\documentclass{article}

\usepackage{xcolor}
\definecolor{red}{RGB}{255,0,0}
\definecolor{blue}{rgb}{0.0, 0.4, 0.65}
\definecolor{orange}{RGB}{253,188,64}
\definecolor{green}{RGB}{154,205,50}

\def\darkgreen{green}
\def\lightgreen{green!30}
\def\darkorange{orange}
\def\neutralorange{orange!60}
\def\lightorange{orange!20}

\usepackage{fullpage}
\usepackage{authblk} 
\usepackage[colorlinks=true, allcolors=blue]{hyperref}
\usepackage{csquotes} 

\usepackage{enumitem}
\setlist[itemize]{itemsep=0pt}
\setlist[enumerate]{itemsep=0pt}
\newcommand*{\normaltextpluscolon}[1]{{#1:}}
\setlist[description]{itemsep=0pt, font=\normalfont\normaltextpluscolon}

\usepackage[xindy]{glossaries}	
\setkeys{glslink}{hyper=false}
\newacronym{oi}{OI}{order-independent}
\newacronym{pands}{P\&S}{pass-and-swap}
\newacronym{fcfs}{FCFS}{first-come-first-served}
\newacronym{alis}{ALIS}{assign-to-the-longest-idle-slot}
\newacronym{fcfs-alis}{FCFS-ALIS}{first-come-first-served and assign-to-the-longest-idle-slot}

\usepackage{amsmath,amssymb}
\allowdisplaybreaks

\newcommand\N{\mathbb{N}}
\newcommand\R{\mathbb{R}}

\newcommand\I{{\mathcal{I}}}
\newcommand\Se{{\mathcal{S}}}
\newcommand\K{{\mathcal{K}}}
\newcommand\T{{\mathcal{T}}}
\newcommand\A{{\mathcal{A}}}
\newcommand\limmu{{\bar\mu}}
\newcommand\X{{\mathcal{X}}}
\newcommand\Y{{\mathcal{Y}}}
\newcommand\C{{\mathcal{C}}}
\newcommand\D{{\mathcal{D}}}

\usepackage{amsthm}

\newtheorem{theorem}{Theorem}
\newtheorem{lemma}{Lemma}
\newtheorem{proposition}{Proposition}

\theoremstyle{definition}

\theoremstyle{remark}
\newtheorem{remark}{Remark}
\newtheorem{example}{Example}

\usepackage{graphics,graphicx,subfig}

\usepackage{tikz,tikz-3dplot}
\usetikzlibrary{calc,shapes}

\tikzset{
	class/.style={draw, minimum size=.7cm},
	server/.style={draw, circle, minimum size=.85cm},
	token/.style={draw, minimum size=.4cm},
	fcfs/.style={
		draw,
		rectangle split,
		rectangle split parts=#1,
		rectangle split horizontal,
		rectangle split empty part width=-.17cm,
		rectangle split empty part height=.65cm,
		inner ysep=.27cm,
	},
	darkgreen/.style={fill=\darkgreen},
	lightgreen/.style={fill=\lightgreen},
	darkorange/.style={fill=\darkorange},
	neutralorange/.style={fill=\neutralorange},
	lightorange/.style={fill=\lightorange},
}

\begin{document}

\title{Pass-and-Swap Queues\footnote{To appear in \textit{Queueing Systems: Theory and Applications}.}}
\author[1]{C\'eline Comte}
\author[2]{Jan-Pieter Dorsman\thanks{Corresponding author:
		Jan-Pieter Dorsman (\href{mailto:j.l.dorsman@uva.nl}{j.l.dorsman@uva.nl}).}}
\affil[1]{Eindhoven University of Technology,
	The Netherlands}
\affil[2]{University of Amsterdam,
	The Netherlands}
\maketitle

\begin{abstract}
	Order-independent (OI) queues, introduced by Berezner, Kriel, and Krzesinski in 1995, expanded the family of multi-class queues that are known to have a product-form stationary distribution by allowing for intricate class-dependent service rates. This paper further broadens this family by introducing pass-and-swap (P\&S) queues, an extension of OI queues where, upon a service completion, the customer that completes service is not necessarily the one that leaves the system. More precisely, we supplement the OI queue model with an undirected graph on the customer classes, which we call a swapping graph, such that there is an edge between two classes if customers of these classes can be \emph{swapped} with one another. When a customer completes service, it passes over customers in the remainder of the queue until it finds a customer it can swap positions with, that is, a customer whose class is a neighbor in the graph. In its turn, the customer that is ejected from its position takes the position of the next customer it can be swapped with, and so on. This is repeated until a customer can no longer find another customer to be swapped with; this customer is the one that leaves the queue. After proving that P\&S queues have a product-form stationary distribution, we derive a necessary and sufficient stability condition for (open networks of) P\&S queues that also applies to OI queues. We then study irreducibility properties of closed networks of P\&S queues and derive the corresponding product-form stationary distribution. Lastly, we demonstrate that closed networks of P\&S queues can be applied to describe the dynamics of new and existing load-distribution and scheduling protocols in clusters of machines in which jobs have assignment constraints. \\[.1cm]
	\textbf{Keywords:} Order-independent queue, product-form stationary distribution, network of queues, quasi-reversibility, first-come-first-served, assign-to-the-longest-idle-server.
\end{abstract}

\section{Introduction} \label{sec:intro}

Since the pioneering work of Jackson~\cite{Jackson1} in the 1950s, queueing networks with a product-form stationary distribution have played
a central role in the development of queueing theory~\cite{BD11,S99}.
In general,
the stationary distribution of a network
is said to have a \emph{product form} if it
can be written as the product of the stationary distributions of the queues
that compose this network.
Further examination of queueing networks with this property led to several breakthroughs, such as the discovery of BCMP~\cite{bcmp} and Kelly~\cite{KW82} networks, which demonstrated the broad applicability of these models.
In addition to implying statistical independence between queues, this product-form property is appealing for its potential for further performance analysis.

In a product-form queueing network, the notion of product form is also relevant at the level of an individual queue in the sense that, aside from the normalization constant, the stationary distribution of each queue is a product of factors, each of which corresponds to a customer in the queue~\cite{D11}.
Dedicated study of this type of product form has gained momentum recently, mainly because of the rising interest in queueing models with arbitrary customer-server compatibilities, in which not every server is able to fulfill the service requirement of any customer.
Such compatibility constraints
are often described by a bipartite graph
between customer classes and servers,
like that of
\figurename~\ref{fig:msq-graph}.
These models arise naturally in many timely applications, such as redundancy scheduling~\cite{BC17,G16} and load balancing~\cite{C19-1,C19-2}
in computer systems, resource management
in manufacturing systems and call centers~\cite{AW12,AW14}, and multiple instances of stochastic matching models~\cite{AKRW18,MBM20}.
Although the dynamics of these queues are rather intricate, their stationary distributions all have a product form, which facilitates exact derivation
of performance measures. A more complete overview of these results can be found in~\cite{GR20}.

\begin{figure}[ht]
	\centering
	\begin{tikzpicture}
		\def\width{1.8cm}
		\def\height{1.4cm}
		
		\node[class, fill=green!60] (c1) {1};
		\node[class, fill=orange!60] (c2)
		at ($(c1)+(\width,0)$) {2};
		
		\node[server] (s1)
		at ($(c1)-(.5*\width,0)-(0,\height)$) {1};
		\node[server] (s3)
		at ($(c1)!.5!(c2)-(0,\height)$) {3};
		\node[server] (s2)
		at ($(c2)+(.5*\width,0)-(0,\height)$) {2};
		
		\draw (s1) -- (c1) -- (s3) -- (c2) -- (s2);
		
		\node[align=left, anchor=west]
		at ($(s2.east)+(.8cm,0)+(0,\height)$)
		{Customer (or job) classes};
		\node[anchor=west]
		at ($(s2.east)+(.8cm,0)$)
		{Servers (or machines)};
		\phantom{
			\node[align=right, anchor=east]
			at ($(s1.west)-(.8cm,0)+(0,\height)$)
			{Customer or job classes};
			\node[anchor=east]
			at ($(s1.west)-(.8cm,0)$)
			{Servers or machines};
		}
	\end{tikzpicture}
	\caption{A compatibility graph between
		two customer classes
		and three servers.}
	\label{fig:msq-graph}
\end{figure}
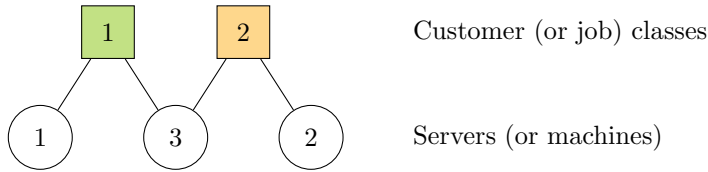

\paragraph{Order-independent queues and pass-and-swap queues}

A remarkable class of queues
that exhibit such a product form
is the class of \gls{oi} queues~\cite{BK95,BK96,K11}.
These are multi-class queues in which, at any point in time, the rate at which any customer completes service may depend on its own class
and the classes of the customers that arrived earlier. \gls{oi} queues owe their name to the fact that the overall service rate of all present customers, although it may depend on their classes, cannot depend on their arrival order.
It was shown in~\cite{BK95} that \gls{oi} queues have a product-form stationary distribution.
Furthermore, as observed in \cite{GR20},
many product-form results found
for the above-mentioned applications
resemble analogous results for \gls{oi} queues.
Several other studies
pointed out connections between
the above-mentioned applications
and placed them in a more general context
of product-form queueing models;
cf.\ \cite{AKRW18, ABV18}.

In this paper, we show that this class of queueing models can be extended in yet another direction, namely the routing of customers within the queue.
We do so by introducing pass-and-swap (P\&S) queues, which preserve the product-form stationary distribution of \gls{oi} queues while covering a wider range of applications.
The distinguishing feature of P\&S queues is a so-called \emph{swapping graph} on customer classes, such that there is an edge
between two classes if customers of these classes can be \emph{swapped} with one another.
Whenever a customer completes service, it scans the remainder of the queue, passes over subsequent customers that are of a non-swappable class, and swaps roles with the first customer of a swappable class, in the sense that it takes the place of this customer, to start another round of service. The ejected customer, in turn, scans the rest of the queue, possibly swapping with yet another customer, and so on. This is repeated until a customer can no longer find a customer to be swapped with. This is the customer that leaves the queue.

\paragraph{Applications to machine clusters}

We shall also see that,
in addition to their theoretical appeal,
P\&S queues can
be used to describe the dynamics of load-distribution and scheduling protocols in
a cluster of machines in which jobs have assignment constraints.
This cluster model, which
can represent various computer clusters
or manufacturing systems
in which not every machine is able
to fulfill the service requirement
of any job,
has played a central role in
several studies of product-form queueing models
over the past decade; see e.g.\ \cite{AW12,AW14,ABDV19,ABV18,BC17,C19-1,C19-2,GR20,G16}.
Roughly speaking, this cluster model can be interpreted as a P\&S queue where customers represent jobs and servers represent machines. To illustrate the diversity of protocols that
can be modeled this way,
we consider a machine cluster
described by the graph
of \figurename~\ref{fig:msq-graph}, with
two job classes
and three machines.
Machines~1 and~2 are dedicated
to classes~1 and~2, respectively,
while machine~3 is compatible with both classes.
These compatibility constraints, which
may for instance result from
data locality in data centers, lead to
different product-form queueing models
depending on the load-distribution
or scheduling protocol.
We now mention two protocols that
were studied in the literature and
can be described using P\&S queues.

With the first protocol,
each incoming job is immediately added
to the buffers of all its compatible machines,
and each machine applies
the \gls{fcfs} service discipline
to the jobs in its buffer.
A job may therefore be in service
on its two compatible machines
at the same time,
in which case its departure rate
is the sum of the service rates of
these two machines.
If the jobs of each class arrive
according to an independent Poisson process
and have independent
and exponentially distributed service requirements,
the corresponding queueing model
is an~\gls{oi} queue
(and therefore also a P\&S queue).
With the second protocol,
each incoming job
enters service on
its compatible machine
that has been idle the longest, if any,
otherwise the job is kept waiting in a central queue
until one of its compatible
machines becomes idle.
This protocol is called
first-come-first-served
and assign-to-the-longest-idle-server (FCFS-ALIS).
With similar assumptions on the arrival process
and job size distribution as before,
it was shown in~\cite{AW14} that,
although the corresponding queueing model
is \emph{not} an \gls{oi} queue,
we again obtain a product-form stationary distribution
that is similar to that
obtained under the first variant.

Although the two protocols
that we have just described
have already been analyzed,
we will shed new light on these protocols
by showing that they can be described
using P\&S queues. Even stronger, we will see that the framework of P\&S queues allows for extensions in different directions
that, as far as the authors are aware, have not been analyzed in the literature. For instance, going back to the toy example
of \figurename~\ref{fig:msq-graph}, it may happen that each incoming job is \emph{a priori} compatible with all machines,
but that it can only be assigned to two neighboring machines due to limited parallelism or other operational constraints.
This observation leads us to introduce a new
load-distribution protocol
in which an incoming job is only assigned to \emph{some} of its compatible machines and is subsequently processed in parallel on any subset of these machines.
This extension falls within the framework
of P\&S queues and, in fact,
it is the P\&S mechanism
that guided the design of this protocol.

In a similar vein, P\&S queues can be used to model redundancy scheduling \cite{G16,ABV18,ABDV19} in clusters of machines, where replicas of a job may be routed to (the buffers of) multiple machines, and redundant replicas are canceled
whenever a replica completes service at a machine (cancel-on-completion) or enters service at a machine (cancel-on-start). One can verify that the dynamics of the cancel-on-completion protocol are tantamount to the first aforementioned protocol, and, as such, can be modeled by an OI queue (and thus also by a P\&S queue). The cancel-on-start protocol is also covered by P\&S queues, as its dynamics are tantamount to those of the second aforementioned protocol~\cite{ABDV19}. Once again, in addition to covering these two existing protocols, P\&S queues can also be used to model new redundancy scheduling protocols. As an example, we introduce the \emph{cancel-on-commit} protocol which generalizes the cancel-on-start protocol. In this protocol, whenever a replica of a job becomes one of the $\ell_s$ oldest replicas in the buffer of a machine~$s$, this replica commits to machine~$s$ and all other replicas of the job are canceled. The cancel-on-start protocol corresponds to the special case where $\ell_s = 1$ for every machine.

\paragraph{Contributions}

Our contributions are as follows. We introduce P\&S queues and establish that, although these queues are a non-trivial generalization of \gls{oi} queues, the product form of the stationary distribution is preserved; this result is proved by careful inspection of the partial balance equations of the underlying Markov chain. This result paves the way for the performance analysis of several applications, such as those we described above, without resorting to scaling regimes. We also provide an easily verifiable necessary and sufficient stability condition for P\&S queues that also holds for \gls{oi} queues.
In addition, we study networks of P\&S queues. By establishing that P\&S queues are quasi-reversible~\cite{kelly,M72}, we show that open networks of P\&S queues exhibit a product-form stationary distribution under mild conditions on the routing process. We also study irreducibility properties of closed networks of P\&S queues and demonstrate that,
under particular assumptions, the stationary distribution of such closed networks also has a product form.
These closed networks form a class of independent interest, since we show later that they can be used to model finite-capacity queues
with token-based structures, akin to those of~\cite{ABDV19} and~\cite{C19-1,C19-2}.

\paragraph{Structure of the paper}

The remainder of the paper is organized as follows. Section~\ref{sec:oi} recalls results on \gls{oi} queues that were derived in~\cite{BK95,K11}. The P\&S queue is introduced in Section~\ref{sec:pands}, where the product form of its stationary distribution is also established.
After deriving complementary results on open (networks of) P\&S queues in Section~\ref{sec:results}, we turn to the analysis of
closed networks of P\&S queues in Section~\ref{sec:closed}.
We demonstrate the applicability of these models
to the modeling of resource-management protocols in Section~\ref{sec:app}.
Section~\ref{sec:ccl} concludes the paper.

\section{Order-independent queues} \label{sec:oi}

This section gives an overview of \gls{oi} queues, introduced in \cite{BK95} and later studied in \cite{K11}.
The results of this section were derived in these two seminal papers and act as a basis for extension in the remainder of the paper.

\subsection{Definition} \label{subsec:oi-def}

We consider a multi-class queue with
a finite set $\I = \{1, \ldots, I\}$ of customer classes.
For each $i \in \I$, class-$i$ customers enter the queue according to
an independent Poisson process with intensity $\lambda_i > 0$.
Customers are queued in their arrival order,
with the oldest customer at the head of the queue,
and are identified by their class.
For now, we assume that the queue has an infinite capacity
and that each customer leaves
the queue immediately upon service completion.

\paragraph{State descriptors}

We consider two state descriptors of this multi-class queue.
The queue \emph{state} represents the classes of customers in the queue in their arrival order.
More specifically, we consider the sequence
$c = (c_1,\ldots,c_n)$,
where $n$ is the total number of customers in the queue
and $c_p$ is the class of the $p$-th oldest customer,
for each $p \in \{1, \ldots, n\}$.
In particular, $c_1$ is the class of the oldest customer,
at the head of the queue. 
The empty state, with $n = 0$, is denoted by $\emptyset$.
The corresponding state space is
the Kleene closure~$\I^*$ of the set $\I$, that is, the set of sequences of finite length made up of elements of $\I$.
To each state $c \in \I^*$,
we associate a \emph{macrostate}
$|c| = (|c|_1, \ldots, |c|_I) \in \N^I$
that only retains the \emph{number} of present customers of each class, and does not keep track of their order in the queue.
As a result, for each $c \in \I^*$ and each $i \in \I$,
the integer $|c|_i$ gives the number of class-$i$ customers in state~$c$.
For each $x, y \in \N^I$,
we write $x \le y$
if $x_i \le y_i$ for each $i \in \I$.

\paragraph{Service rates}
We now explain the way in which service is provided to customers in an OI queue. This is done in such a way that the evolution of the state of the queue over time exhibits a memoryless property (and thus represents a Markov process). 
The overall service rate in state~$c$ is denoted by $\mu(c)$, for each $c \in \I^*$. This function $\mu$, defined on $\I^*$, is called the \emph{rate function} of the queue. Along with the individual rates of service provided to the customers in the queue, it satisfies the following two conditions.
First, the overall rate of service $\mu(c)$ provided
when the queue is in state~$c$ depends only
on the number of customers of each class that are present
and not on their arrival order.
In other words, for each $c, d \in \I^*$,
we have $\mu(c) = \mu(d)$ whenever $|c| = |d|$. 
For this reason, we shall also refer to $\mu(c)$ as $\mu(x)$
when $x$ is the macrostate corresponding to state~$c$.
Second, the service rate of each customer
is independent of (the number and classes of the) customers that are 
behind this customer in the queue.
In particular, for each $c = (c_1,\ldots,c_n) \in \I^*$ and $p \in \{1, \ldots, n\}$, the service rate of the customer in position~$p$ in state~$c$, of class $c_p$, is equal to the increment of the overall service rate induced by the arrival of this customer, denoted by
$$
\Delta\mu(c_1,\ldots,c_p) = \mu(c_1,\ldots,c_p) - \mu(c_1,\ldots,c_{p-1}),
$$
where we use the convention that
$(c_1, \ldots, c_{p-1}) = \emptyset$
when $p = 0$.
This implies in particular
that the function~$\mu$
is non-decreasing, in the sense that
$$
\mu(c_1,\ldots,c_n,i) \ge \mu(c_1,\ldots,c_n),
\quad \forall c = (c_1,\ldots,c_n) \in \I^*,
\quad \forall i \in \I.
$$
The service rate of the first $p$ customers in the queue, given by $\mu(c_1, \ldots, c_p) = \sum_{q=1}^p \Delta\mu(c_1, \ldots, c_q)$, depends neither on the classes of the customers in positions $p+1$ to $n$ nor even on the total number $n$ of customers in the queue.

We set $\mu(\emptyset) = 0$ since the queue exhibits
a zero departure rate when there are no customers in the queue.
We additionally assume that $\mu(c) > 0$ for each $c \neq \emptyset$.
In other words, we assume that the oldest customer
always receives a positive service rate, to ensure irreducibility of the Markov process describing the evolution of the state over time.

\begin{remark}
	The definition of \gls{oi} queues that we presented above is slightly more restrictive than that of \cite{BK95,K11}. Indeed, in these two papers, the overall service rate function $\mu$ is scaled by a factor that depends on the total number of customers in the queue.
	We omit this scaling factor for simplicity of notation and assume it to equal one. However, unless stated otherwise, the results in the sequel of this paper can be straightforwardly generalized to account for this factor.
\end{remark}

\paragraph{Examples}

As observed in \cite{BK95,K11},
the framework of \gls{oi} queues
encompasses several classical queueing models, such as the \gls{fcfs} and infinite-server queues of BCMP networks~\cite{bcmp}, as well as multiple queues with class-based compatibilities (see \cite{msccc, mshcc} for example).
In this paper, we will be especially interested in the following multi-server queues, introduced in \cite{G16} and identified as \gls{oi} queues in~\cite{BC17}.

\begin{example}[Multi-server queue] \label{ex:oi}
	Consider an infinite-capacity queue with
	a set $\I = \{1, \ldots, I\}$ of customer classes
	and a set $\Se = \{1, \ldots, S\}$ of servers.
	All customers have an exponentially-distributed
	size with unit mean and,
	for each $i \in \I$,
	class-$i$ customers enter the queue according to
	a Poisson process with rate $\lambda_i > 0$
	and can be processed
	by the servers of the set $\Se_i \subseteq \Se$.
	This defines a bipartite compatibility graph
	between customer classes and servers,
	in which there is an edge between a class and a server
	if this server can process customers of this class.
	In the example of \figurename~\ref{fig:msq-graph},
	servers~1 and~2 are dedicated
	to classes~1 and~2, respectively,
	while server~3
	is compatible with both classes.
	In this way, we have
	$\Se_1 = \{1,3\}$ and $\Se_2 = \{2,3\}$.
	
	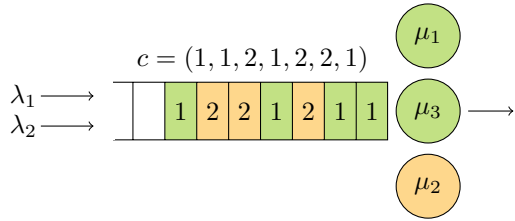
\begin{figure}[b]
		\centering
		\begin{tikzpicture}
			\def\width{1.2cm}
			\def\height{1cm}
			
			\node[fcfs=9,
			rectangle split part fill={
				white, white, green!60, orange!60, orange!60,
				green!60, orange!60, green!60, green!60
			}] (queue) {
				\nodepart{one}{\phantom{1}}
				\nodepart{two}{\phantom{1}}
				\nodepart{three}{1}
				\nodepart{four}{2}
				\nodepart{five}{2}
				\nodepart{six}{1}
				\nodepart{seven}{2}
				\nodepart{eight}{1}
				\nodepart{nine}{1}
			};
			\fill[white] ([xshift=-\pgflinewidth-1pt,
			yshift=-\pgflinewidth+1pt]queue.north west)
			rectangle ([xshift=\pgflinewidth+4pt,
			yshift=\pgflinewidth-1pt]queue.south west);
			\fill[white] ([xshift=-\pgflinewidth-1pt,
			yshift=-\pgflinewidth-.01pt]queue.north west)
			rectangle ([xshift=\pgflinewidth+5pt,
			yshift=\pgflinewidth+.01pt]queue.south west);
			
			\node[server, anchor=west, fill=green!60]
			(mu1)
			at ($(queue.east)+(.1cm,0)+(0,\height)$)
			{$\mu_1$};
			\node[server, anchor=west, fill=green!60]
			(mu2)
			at ($(queue.east)+(.1cm,0)$)
			{$\mu_3$};
			\node[server, anchor=west, fill=orange!60]
			(mu2)
			at ($(queue.east)+(.1cm,0)-(0,\height)$)
			{$\mu_2$};
			
			\node[anchor=south]
			at ($(queue.north)+(.1cm,0)$)
			{$c = (1,1,2,1,2,2,1)$};
			
			\draw[->] ($(queue.west)-(.8cm,0)+(0,.2cm)$)
			-- node[near start, anchor=east, xshift=-.1cm] (lambda1) {$\lambda_1$} ($(queue.west)-(.1cm,0)+(0,.2cm)$);
			\draw[->] ($(queue.west)-(.8cm,0)-(0,.2cm)$)
			-- node[near start, anchor=east, xshift=-.1cm] {$\lambda_2$}
			($(queue.west)-(.1cm,0)-(0,.2cm)$);
			
			\draw[->]
			($(mu1.east)!.5!(mu2.east)+(.1cm,0)$)
			-- ($(mu1.east)!.5!(mu2.east)+(.7cm,0)$);
			
			\node at ($(lambda1.west)-(.21cm,0)$) {};
			\node at ($(mu1.east)+(.7cm,0)+(.21cm,0)$) {};
			\phantom{%
				\node[class]
				at ($(mu2)-(0,.2*\height)$) {};
			}
		\end{tikzpicture}
		\caption{A queue state.
		The color of a server is a visual aid that
		indicates the class of the customer
		currently in service on this server.}
		\label{fig:msq-state}
	\end{figure}
	
	Each server applies the \gls{fcfs}
	discipline to the customers it can serve,
	so that each customer is in service
	on all the servers that can
	process this customer
	but not the customers that
	arrived earlier in the queue.
	For each $s \in \Se$,
	the service rate of server~$s$ is denoted by $\mu_s > 0$.
	When a class-$i$ customer is
	in service on a subset $\T \subseteq \Se_i$
	of its compatible servers,
	its service rate is $\sum_{s \in \T} \mu_s$.
	The overall service rate,
	equal to the sum of the service rates
	of the servers that can process at least one customer
	in the queue, is given by
	\begin{equation} \label{eq:mu}
		\mu(c_1, \ldots, c_n)
		= \sum_{s \in \bigcup_{p=1}^n \Se_{c_p}} \mu_s,
		\quad \forall (c_1, \ldots, c_n) \in \C.
	\end{equation}
	For each $p \in \{1, \ldots, n\}$,
	the service rate of the customer
	in position~$p$ is given by
	\begin{equation} \label{eq:delta-mu}
		\Delta\mu(c_1, \ldots, c_p)
		= \mu(c_1, \ldots, c_p) - \mu(c_1, \ldots, c_{p-1})
		= \sum_{s \in \Se_{c_p} \setminus \bigcup_{q=1}^{p-1}
			\Se_{c_q}} \mu_s.
	\end{equation}
	One can verify that
	this multi-server queue is an \gls{oi} queue.
	In the example
	of \figurename~\ref{fig:msq-state},
	the queue state
	is $c = (1,1,2,1,2,2,1)$.
	The oldest customer, of class~1,
	is in service on servers~1 and~3,
	at rate
	$\Delta\mu(1)
	= \mu(1) - \mu(\emptyset)
	= (\mu_1 + \mu_3) - 0
	= \mu_1 + \mu_3$.
	The second oldest customer, of class~1,
		is not in service on any server, and indeed we have $\Delta\mu(1,1) = \mu(1,1) - \mu(1) = 0.$
	The third oldest customer, of class~2,
	is in service on server~2, at rate
	$\Delta\mu(1,1,2)
	= \mu(1,1,2) - \mu(1,1)
	= (\mu_1 + \mu_2 + \mu_3) - (\mu_1 + \mu_3)
	= \mu_2$.
	The other customers
	have a zero service rate.
\end{example}

\subsection{Stationary analysis}
\label{subsec:oi-steady}

The evolution of the queue state leads to a Markov process with state space $\I^*$, and this Markov process is irreducible. Indeed, for any two states $c = (c_1, \ldots, c_n) \in \I^*$ and $d = (d_1, \ldots, d_m) \in \I^*$, the Markov process can first jump from state~$c$ to state $\emptyset$ as a result of $n$ transitions corresponding to departures of the customer at the head of the queue, and then from state $\emptyset$ to state~$d$ as a result of $m$ transitions corresponding to customer arrivals.

Theorem~\ref{theo:oi} below recalls that this Markov process has a product-form stationary distribution and that the \gls{oi} queue satisfies the quasi-reversibility property (cf.\ \cite[Section 3.2]{kelly}). This property implies that, when the Markov process associated with the queue state is stationary, the departure instants of the customers of each class form independent and stationary Poisson processes and that, at every instant, the current queue state is independent of the departure instants of customers prior to that instant.
Quasi-reversibility also implies that an open network of \gls{oi} queues connected by a Markovian routing policy has a product-form stationary distribution \cite[Theorem 3.7]{kelly} under mild conditions on the routing process
(namely, each customer can become part of any given class with a positive probability and eventually leaves the network with probability one). A similar result holds for closed networks of \gls{oi} queues under some irreducibility assumptions
\cite[Section 3.4]{kelly}. The interested reader is referred to \cite[Sections 3.2 and 3.4]{kelly} and \cite[Chapter 8]{S99} for a more complete account on quasi-reversibility.
The proof below can be found in \cite{BK95, K11} but we present it here for ease of later reference when we introduce the P\&S queue.

\begin{theorem} \label{theo:oi}
	Consider an \gls{oi} queue
	with a set $\I = \{1, \ldots, I\}$ of customer classes,
	per-class arrival rates $\lambda_1, \ldots, \lambda_I$,
	and a rate function $\mu$.
	A stationary measure of the Markov process
	associated with the state of this \gls{oi} queue is of the form
	\begin{equation} \label{eq:oi-pic}
		\pi(c_1, \ldots, c_n)
		= \pi(\emptyset)
		\prod_{p = 1}^{n}
		\frac{\lambda_{c_p}}{\mu(c_1, \ldots, c_p)}
		= \pi(\emptyset)
		\Phi(c) \prod_{i \in \I} {\lambda_i}^{|c|_i},
		\quad \forall (c_1, \ldots, c_n) \in \I^*,
	\end{equation}
	where $\Phi$ is the \emph{balance function} of the \gls{oi} queue,
	defined on $\I^*$ by
	\begin{equation} \label{eq:oi-Phic}
		\Phi(c_1, \ldots, c_n) = \prod_{p=1}^n \frac1{\mu(c_1, \ldots, c_p)},
		\quad \forall (c_1, \ldots, c_n) \in \I^*,
	\end{equation}
	and $\pi(\emptyset)$ is an arbitrary positive constant.
	The queue is stable if and only if
	\begin{equation} \label{eq:oi-stability}
		\sum_{c \in \I^*}
		\Phi(c) \prod_{i \in \I} {\lambda_i}^{|c|_i}
		< +\infty,
	\end{equation}
	in which case
	the queue is quasi-reversible
	and the stationary distribution
	of the Markov process
	associated with its state
	is given by~\eqref{eq:oi-pic}
	with
	\begin{equation} \label{eq:oi-empty}
		\pi(\emptyset) =
		\left(\sum_{c \in \I^*}		\Phi(c) \prod_{i \in \I} {\lambda_i}^{|c|_i}\right)^{-1}.
	\end{equation}
\end{theorem}

\begin{proof}
	We will first verify that
	any measure $\pi$
	of the form~\eqref{eq:oi-pic}
	satisfies the following
	partial balance equations
	in each state
	$c = (c_1,\ldots,c_n) \in \I^*$:
	\begin{itemize}
		\item Equalize the flow
		out of state~$c$
		due to a departure
		with the flow
		into state~$c$
		due to an arrival (if $c \neq \emptyset$):
		\begin{equation} \label{eq:oi-partial-balance-1}
			\pi(c) \mu(c) = \pi(c_1,\ldots,c_{n-1}) \lambda_{c_n}.
		\end{equation}
		\item Equalize, for each $i \in \I$,
		the flow out of state~$c$
		due to the arrival of a class-$i$ customer
		with the flow into state~$c$
		due to the departure of a customer of this class:
		\begin{align} \label{eq:oi-partial-balance-2}
			\pi(c) \lambda_i
			= \sum_{p=1}^{n+1} \pi(c_1, \ldots, c_{p-1}, i, c_p, \ldots, c_n)
			\, \Delta\mu(c_1, \ldots, c_{p-1}, i).
		\end{align}
	\end{itemize}
	This verification will readily imply
	that the stationary measures
	of the Markov process associated with
	the queue state
	are of the form~\eqref{eq:oi-pic}
	and that the queue, when stable,
	is quasi-reversible.
	Indeed, the global balance equations
	of the Markov process
	associated with the queue state
	follow from the partial balance equations~\eqref{eq:oi-partial-balance-1}
	and~\eqref{eq:oi-partial-balance-2}
	by summation.
	Moreover,
	to prove that the queue is quasi-reversible,
	it suffices to verify
	that the stationary measures of
	the Markov process
	associated with the queue state
	satisfy the partial balance equations~\eqref{eq:oi-partial-balance-2}.
	This is a consequence of
	\cite[Equations (3.8) to (3.11)]{kelly}
	and of the fact that
	the customers of each class
	enter the queue according to an
	independent and stationary Poisson process.
	
	We now verify that
	the measures of the form~\eqref{eq:oi-pic}
	satisfy the partial balance
	equations~\eqref{eq:oi-partial-balance-1}
	and~\eqref{eq:oi-partial-balance-2}.
	Equation \eqref{eq:oi-partial-balance-1} follows immediately from \eqref{eq:oi-pic}	and~\eqref{eq:oi-Phic}.
	The case
	of~\eqref{eq:oi-partial-balance-2}
	is more intricate.
	First observe that
	the measures~$\pi$ of the form~\eqref{eq:oi-pic}
	satisfy~\eqref{eq:oi-partial-balance-2}
	if and only if
	the balance function~$\Phi$
	given by~\eqref{eq:oi-Phic}
	satisfies the following equation
	in each state
	$c = (c_1, \ldots, c_n) \in \I^*$:
	\begin{align} \label{eq:oi-partial-balance-3}
		\Phi(c)
		= \sum_{p=1}^{n+1}
		\Phi(c_1, \ldots, c_{p-1}, i, c_p, \ldots, c_n)
		\, \Delta\mu(c_1, \ldots, c_{p-1}, i),
		\quad \forall i \in \I.
	\end{align}
	We show that $\Phi$ satisfies this equation
	by induction over the queue length~$n$.
	For the base step, with $n = 0$,
	it suffices to observe that,
	for each $i \in \I$,
	the right-hand side of~\eqref{eq:oi-partial-balance-3}
	simplifies to $\Phi(i) \mu(i)$,
	which is equal to $\Phi(\emptyset)$
	by~\eqref{eq:oi-Phic}.
	Now let $n \ge 1$ and assume that~\eqref{eq:oi-partial-balance-3}
	is satisfied
	for each state~$c$ of length $n-1$.
	Consider a state~$c$ of length $n$
	and let $i \in \I$.
	The first $n$ terms in the sum
	on the right-hand side of~\eqref{eq:oi-partial-balance-3}
	can be rewritten as follows:
	\begin{align}
		\nonumber
		&\sum_{p=1}^{n}
		\Phi(c_1, \ldots, c_{p-1}, i, c_p, \ldots, c_n)
		\, \Delta\mu(c_1, \ldots, c_{p-1}, i) \\
		\nonumber
		&= \frac1{\mu(c_1, \ldots, c_n, i)}
		\sum_{p=1}^{n}
		\Phi(c_1, \ldots, c_{p-1}, i, c_p, \ldots, c_{n-1})
		\, \Delta\mu(c_1, \ldots, c_{p-1}, i), \\
		\nonumber
		&= \frac1{\mu(c_1, \ldots, c_n, i)}
		\Phi(c_1, \ldots, c_{n-1}), \\
		\label{eq:oi-partial-balance-4}
		&= \Phi(c_1, \ldots, c_n)
		\frac{\mu(c_1, \ldots, c_n)}{\mu(c_1, \ldots, c_n, i)},
	\end{align}
	where the first and last equalities follow
	from~\eqref{eq:oi-Phic}
	and the order independence of $\mu$,
	while the second equality is obtained by
	applying the induction assumption
	to state $(c_1, \ldots, c_{n-1})$
	and class~$i$.
	By~\eqref{eq:oi-Phic}, we also have that
	\begin{align} \label{eq:oi-partial-balance-5}
		\Phi(c_1, \ldots, c_n, i)
		\, \Delta\mu(c_1, \ldots, c_n, i)
		= \Phi(c_1, \ldots, c_n)
		\, \frac{\mu(c_1, \ldots, c_n, i) 
			- \mu(c_1, \ldots, c_n)}
		{\mu(c_1, \ldots, c_n, i)},
	\end{align}
	so that summing~\eqref{eq:oi-partial-balance-4} and~\eqref{eq:oi-partial-balance-5} yields \eqref{eq:oi-partial-balance-3}.
	This concludes the proof by induction.
	
	The stability condition~\eqref{eq:oi-stability} is equivalent to the statement that $\sum_{c\in\I^*} \pi(c) / \pi(\emptyset)$ is finite for any stationary measure $\pi$, which is indeed necessary
	and sufficient for ergodicity.
	Equation~\eqref{eq:oi-empty} guarantees that the stationary distribution sums to unity.
\end{proof}

\section{Pass-and-swap queues} \label{sec:pands}

This section contains our first main contribution.
Pass-and-swap (P\&S) queues,
obtained by supplementing
\gls{oi} queues with an
additional mechanism
when customers complete service,
are defined
in Section~\ref{subsec:pands-def}.
In contrast to signals
and negative customers
considered for quasi-reversible
queues~\cite{C11},
the P\&S mechanism
occurs upon a service completion
and can move several customers
at the same time within the queue.
Section~\ref{subsec:pands-steady}
shows that both
the product-form nature of the
stationary measure of the Markov process
associated with the state
and the quasi-reversibility property
of the queue
are preserved by this mechanism.

\subsection{Definition} \label{subsec:pands-def}

As before, the set of customer classes is denoted by $\I = \{1, \ldots, I\}$
and we adhere to the state descriptors:
the state $c = (c_1, \ldots, c_n) \in \I^*$ gives
the classes of customers
as they are ordered in the queue
and the macrostate
$|c| = (|c|_1, \ldots, |c|_I) \in \N^I$
gives the numbers of customers of each class.
Likewise, the customer arrival processes
and completion times
are as defined in Section~\ref{subsec:oi-def}.
We however part with the assumption that a customer
that completes service leaves the queue directly
and that all customers behind move forward one position.
With the mechanism that we will now define,
each service completion will potentially trigger
a chain reaction within the queue.
More precisely, a customer that completes service
may take another customer's position further down the queue
and require a new round of service in this position.
The customer ejected from this position may in turn take the position of another customer further down the queue, and so on.
The decision of which customer
replaces which other customer
is driven by the pass-and-swap mechanism
described below.

\paragraph{Pass-and-swap mechanism}

We supplement the \gls{oi} queue
with an undirected graph that
will be called the \emph{swapping graph} of the queue.
The vertices of this graph represent the customer classes. For each $i, j \in \I$, when there is an edge between classes~$i$ and~$j$ in the graph, a class-$i$ customer can take the position of
a class-$j$ customer upon service completion. In that case, we say that (the customers of) classes~$i$ and $j$ are mutually \emph{swappable}.
Observe that the graph is undirected, meaning that the swapping relation is symmetric: if class~$i$ can be swapped with class~$j$, then class~$j$ can also be swapped with class~$i$.
Also observe that the graph
may contain a loop,
that is, an edge that connects a node to itself,
in which case two customers
of the corresponding class
can also be swapped with one another.
For each $i \in \I$, we let $\I_i \subseteq \I$ denote the set of neighbors of class~$i$ in the graph,
that is, the set of classes that are swappable with class~$i$.

Based on this graph,
the pass-and-swap mechanism
is defined as follows.
A customer whose service is complete scans the rest of the queue, passes over subsequent customers that it cannot swap positions with, and replaces the first swappable customer, if any.
This ejected customer, in turn, scans the rest of the queue and replaces the first swappable customer afterwards. This is repeated until an ejected customer finds no customers in the remainder of the queue it can swap with.
In this case, the customer leaves the queue.

More formally, let
$c = (c_1, \ldots, c_n) \in \I^*$
denote a queue state.
Assume that the service of the customer
in some position $p_1 \in \{1, \ldots, n\}$ completes
and let $i_1 = c_{p_1}$ denote the class of this customer.
If there is at least one position
$q \in \{p_1 + 1, \ldots, n\}$
such that $c_q \in \I_{i_1}$,
we let $p_2$ denote the smallest of these positions
and $i_2 = c_{p_2}$ the class of the corresponding customer.
The class-$i_1$ customer
that was originally in position~$p_1$
replaces the class-$i_2$ customer
in position~$p_2$,
and this class-$i_2$ customer is ejected.
If there is at least one position
$q \in \{p_2 + 1, \ldots, n\}$
such that $c_q \in \I_{i_2}$,
we let $p_3$ denote
the smallest of these positions
and $i_3 = c_{p_3}$ the class
of the corresponding customer.
The class-$i_2$ customer
that was originally in position~$p_2$
replaces the class-$i_3$ customer
in position~$p_3$,
and this class-$i_3$ customer is ejected.
Going on like this, we define recursively
$p_{v+1} = \min\{ q \ge p_v + 1:
c_q \in \I_{c_{p_v}} \}$
for each $v \in \{1, \ldots, u-1\}$,
where $p_u \in \{p, \ldots, n\}$
is the position of
the first ejected customer
that can no longer replace another
subsequent customer in the queue,
that is, for which there is
no $q \in \{p_u + 1, \ldots, n\}$
such that $c_q \in \I_{c_{p_u}}$.
This customer is the one
that leaves the queue.
The integer $u \in \{1, \ldots, n - p + 1\}$ gives
the total number of customers
that are involved in the transition,
and the state reached
after this transition is
$$
(c_1,\ldots,c_{p_1-1},c_{p_1+1},\ldots,c_{p_2-1},i_1,
c_{p_2+1},\ldots,c_{p_3-1},i_2,c_{p_3+1},\ldots,
c_{p_u-1},i_{u-1},c_{p_u+1},\ldots,c_n).
$$
Observe that the transition is recorded
as a departure of a class-$i_u$ customer
and not as a departure of
a class-$i_1$ customer in general.
In the special case where $u = 1$,
the customer that completes service
cannot replace any subsequent customer,
so that this customer leaves the queue.
An \gls{oi} queue
supplemented with the pass-and-swap mechanism is called a \gls{pands} queue. The stochastic process keeping track of the state over time in the \gls{pands} queue has the Markov property, just like the \gls{oi} queue. For both the \gls{oi} and the \gls{pands} queue, however, the stochastic process that describes the macrostate over time does not yield any such Markov property in general.

\begin{example}[Multi-server queue] \label{ex:pands}
	We give a toy example
	that illustrates the \gls{pands} mechanism.
	More concrete applications
	will be described in Section~\ref{sec:app}.
	Consider a multi-server queue
	with the compatibility graph
	shown in \figurename~\ref{subfig:pands-compatibility} and
	the swapping graph
	shown in \figurename~\ref{subfig:pands-toy-graph}.
	Customers of
		classes~1 and~2 can be processed
	only by servers~1 and~2, respectively,
	while class-3 customers
	can be processed by both servers.
	Class-$2$ customers
	can be swapped with
	customers of classes $1$ and $3$
	but customers of
	classes $1$ and $3$ cannot
	be swapped with one another.
	Assume that the queue is
	in the state $c = (1,3,3,2,2,3,1,2)$
	depicted in \figurename~\ref{subfig:pands-toy-state},
	and that the customer in first position,
	of class~$1$, completes service.
	The corresponding chain reaction
	is depicted on the same figure by arrows.
	Class~1 can only be swapped with class~2
	and the first subsequent
	class-2 customer is in the fourth position.
	Therefore, the class-1 customer
	that completes service
	is passed along the queue
	up until the fourth position
	and is swapped
	with the class-2 customer
	at this position.
	The ejected customer is of class~2,
	and class~2 can be swapped with
	classes~1 and 3.
	Therefore, the ejected customer
	is passed
	along the queue up until the sixth position
	and is swapped with
	the class-3 customer
	at this position.
	We repeat this with
	the ejected class-3 customer,
	which replaces the last class-2 customer,
	which leaves the queue.
	The transition is labeled as
	a departure of a class-$2$ customer
	and the new queue state
	is $d = (3,3,1,2,2,1,3)$.
	
	\begin{figure}[ht]
		\centering
		\subfloat[Compatibility graph.
		\label{subfig:pands-compatibility}]{%
			\begin{tikzpicture}
				\def\width{1.8cm}
				\def\height{1.4cm}
				
				\node[class, fill=green!60] (c1) {1};
				\node[class, fill=yellow!60] (c3)
				at ($(c1)+(\width,0)$) {3};
				\node[class, fill=orange!60] (c2)
				at ($(c3)+(\width,0)$) {2};
				
				\node[server] (s1)
				at ($(c1)!.5!(c3)-(0,\height)$) {1};
				\node[server] (s2)
				at ($(c2)!.5!(c3)-(0,\height)$) {2};
				
				\draw (c1) -- (s1) -- (c3) -- (s2) -- (c2);
			\end{tikzpicture}
		}
		\hfill
		\subfloat[Swapping graph.
		\label{subfig:pands-toy-graph}]{
			\begin{tikzpicture}
				\def\height{1cm}
				
				\node[class, fill=green!60] (1) {1};
				\node[class, fill=orange!60] (2)
				at ($(1)-(0,\height)$) {2};
				\node[class, fill=yellow!60] (3)
				at ($(2)-(0,\height)$) {3};
				
				\draw[-] (1) -- (2) -- (3);
				
				\node at ($(1.west)-(.75cm,0)$) {};
				\node at ($(1.east)+(.75cm,0)$) {};
			\end{tikzpicture}
		}
		\hfill
		\subfloat[A queue state and the transition that occurs upon service completion of the first class-1 customer. \label{subfig:pands-toy-state}]{
			\begin{tikzpicture}
				\def\width{1.2cm}
				\def\height{1cm}
				
				\node[fcfs=10,
				rectangle split part fill
				={white, white, orange!60, green!60, yellow!60, orange!60, orange!60, yellow!60, yellow!60, green!60},
				] (queue) {
					\nodepart{one}{\phantom{1}}
					\nodepart{two}{\phantom{1}}
					\nodepart{three}{2}
					\nodepart{four}{1}
					\nodepart{five}{3}
					\nodepart{six}{2}
					\nodepart{seven}{2}
					\nodepart{eight}{3}
					\nodepart{nine}{3}
					\nodepart{ten}{1}
				};
				\fill[white] ([xshift=-\pgflinewidth-1pt,yshift=-\pgflinewidth+1pt]queue.north west)
				rectangle ([xshift=\pgflinewidth+4pt,yshift=\pgflinewidth-1pt]queue.south west);
				\fill[white] ([xshift=-\pgflinewidth-1pt,yshift=-\pgflinewidth-.01pt]queue.north west)
				rectangle ([xshift=\pgflinewidth+5pt,yshift=\pgflinewidth+.01pt]queue.south west);
				
				\node[server, anchor=west, fill=green!60]
				(mu1)
				at ($(queue.east)+(.1cm,0)+(0,.5*\height)$)
				{$\mu_1$};
				\node[server, anchor=west, fill=yellow!60]
				(mu2)
				at ($(queue.east)+(.1cm,0)-(0,.5*\height)$)
				{$\mu_2$};
				
				\draw[->] ($(queue.west)-(.8cm,0)+(0,.25cm)$) --
				node[near start, anchor=east, xshift=-.1cm] (lambda1)
				{$\lambda_1$} ($(queue.west)-(.1cm,0)+(0,.25cm)$);
				\draw[->] ($(queue.west)-(.8cm,0)-(0,.25cm)$) --
				node[near start, anchor=east, xshift=-.1cm]
				{$\lambda_2$} ($(queue.west)-(.1cm,0)-(0,.25cm)$);
				\draw[->]
				($(mu1.east)!.5!(mu2.east)+(.1cm,0)$)
				-- node[pos=1] (departure) {}
				($(mu1.east)!.5!(mu2.east)+(.7cm,0)$);
				
				\node[anchor=south, token, fill=green!60] (t1)
				at ($(queue.ten north)+(0,.6cm)$)
				{\scriptsize 1};
				\node[anchor=south, token] (t2)
				at ($(queue.seven north)+(0,.6cm)$) {};
				\node[anchor=south, token] (t3)
				at ($(queue.five north)+(0,.6cm)$) {};
				\node[anchor=south, token] (t4)
				at ($(queue.three north)+(0,.6cm)$) {};
				\node[anchor=south, token, fill=orange!60] (t5)
				at ($(queue.one north)+(0,.6cm)$)
				{\scriptsize 2};
				
				\draw[->] (queue.ten north) -- (t1);
				
				\draw[->] (t1)
				-- node[midway, fill=white, inner sep=.04cm]
				{\scriptsize pass}
				(t2);
				\draw ($(t2.south)+(0.08cm,0)$)
				edge[bend left, ->]
				($(queue.seven north)+(0.08cm,0)$);
				\draw ($(queue.seven north)-(0.08cm,0)$)
				edge[bend left, ->]
				node[fill=white, text=white,
				midway, xshift=.14cm, yshift=.03cm,
				inner sep=.044cm]
				{\scriptsize swaa}
				node[midway, xshift=.14cm, inner sep=.04cm]
				{\scriptsize swap}
				($(t2.south)-(0.08cm,0)$);
				
				\draw[->] (t2) -- (t3);
				\draw ($(t3.south)+(0.08cm,0)$)
				edge[bend left, ->]
				($(queue.five north)+(0.08cm,0)$);
				\draw ($(queue.five north)-(0.08cm,0)$)
				edge[bend left, ->]
				($(t3.south)-(0.08cm,0)$);
				
				\draw[->] (t3) -- (t4);
				\draw ($(t4.south)+(0.08cm,0)$)
				edge[bend left, ->]
				($(queue.three north)+(0.08cm,0)$);
				\draw ($(queue.three north)-(0.08cm,0)$)
				edge[bend left, ->]
				($(t4.south)-(0.08cm,0)$);
				
				\draw[->] (t4) -- (t5);
				
				\node at ($(lambda1.west)-(.0cm,0)$) {};
				\node at ($(departure.east)+(.0cm,0)$) {};
			\end{tikzpicture}
		}
		\caption{Toy example of a \gls{pands} queue.}
		\label{fig:pands-toy}
	\end{figure}
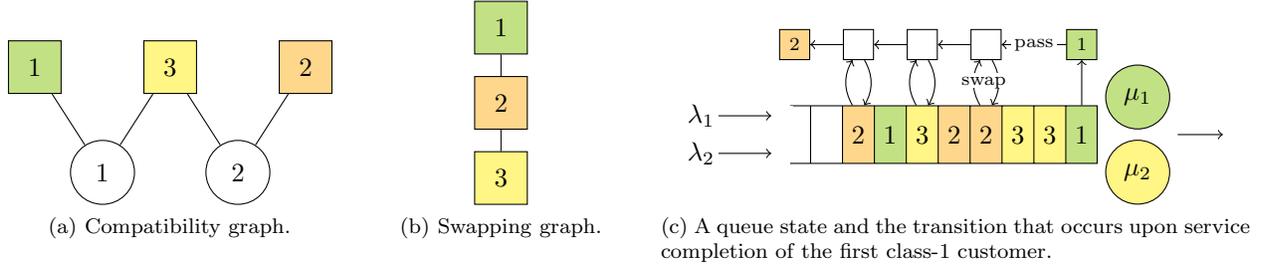
\end{example}

\begin{example}[Single-server
	last-come-first-served queue]
	The following example shows that
	the \gls{pands} mechanism
	can also be used to emulate
	a single-server queue with the
	last-come-first-served preemptive-resume policy.
	Consider a \gls{pands} queue
	with a set $\I = \{1, \ldots, I\}$
	of customer classes.
	Assume that the service rate
	is a positive constant $\mu$,
	independent of the numbers
	of customers of each class in the queue
	(provided that the queue is not empty),
	so that only the customer
	at the head of the queue
	has a positive service rate.
	If the swapping graph is empty,
	then the customer at the head of the queue
	leaves immediately upon service completion,
	and the queue behaves
	like a multi-class single-server queue
	with the first-come-first-served policy.
	On the contrary,
	if the swapping graph is complete
	(with in particular loops at all nodes),
	then the service completion
	of the customer at the head of the queue
	triggers a cascading effect whereby
	every customer swaps position
	with its successor in the queue,
	so that the customer that leaves
	is the one at the tail of the queue.
	Therefore, the queue behaves
	like a multi-class single-server queue
	with the \emph{last}-come-first-served
	preemptive-resume policy.
	Note that, unfortunately, this example
	only works because the service time
	of each customer is exponentially distributed
	with a mean that is independent of its class;
	in particular, this example does not explain
	the insensitivity of the last-come-first-served
	preemptive-resume policy
	to the distribution of the service times.
\end{example}

\paragraph{Additional comments}

The \gls{pands} mechanism is
applied instantaneously
upon a service completion.
Replacements are performed
from the front to the back of the queue,
so that,
when a customer is ejected
from some position $p \in \{1,\ldots,n\}$,
this customer never replaces a customer
at a position $q \in \{1, \ldots, p-1\}$,
even if $c_q \in \I_{c_p}$.
Also note that customer classes
now have a dual role:
they determine not only the service rate
received by each customer
through the rate function $\mu$,
but also the chain reaction
that happens upon each service completion
through the swapping graph.
One could also associate two classes
with each customer,
one that determines its service rate
and another that determines its swapping relations.
Finally, the original \gls{oi} queue,
in which the customer that completes service
is the one that leaves the queue,
is obtained by applying
the \gls{pands} mechanism based on a swapping graph without edges.
Therefore, an \gls{oi} queue is
also a \gls{pands} queue, so that the results that
we will derive for \gls{pands} queues in the sequel
also apply to \gls{oi} queues.

The following notation will be useful.
We write
$\delta_p(c) = (d, i)$
if the service completion of the customer
in position~$p$ in state~$c$
leads to state~$d$
and triggers the departure
of a class-$i$ customer,
for each $n \ge 1$,
$c = (c_1, \ldots, c_n) \in \I^*$,
	$d = (d_1, \ldots, d_{n-1}) \in \I^*$,
$p \in \{1, \ldots, n\}$,
and $i \in \I$.
In Example~\ref{ex:pands} for instance,
we have
$\delta_1(c) = ((3,3,1,2,2,1,3), 2)$.
With a slight abuse of notation,
we also write $\delta_p(c) = i$
if we only want to specify
that the departing customer
is of class~$i$.

\subsection{Stationary analysis}
\label{subsec:pands-steady}

The Markov process associated
with the queue state
on the state space $\I^*$ is irreducible.
Indeed, given any states
$c = (c_1, \ldots, c_n) \in \I^*$
and $d = (d_1, \ldots, d_m) \in \I^*$,
the Markov process can again jump
from state~$c$ to state~$\emptyset$
thanks to $n$ transitions corresponding to departures
(for instance, triggered by
the service completion
of the customer at the head of the queue),
and then from state~$\emptyset$ to state~$d$
thanks to $m$ transitions corresponding to arrivals.

Theorem~\ref{theo:pands} below
shows that introducing
the \gls{pands} mechanism
does actually not change
the stationary distribution
of this Markov process
compared to the original \gls{oi} queue.
In particular,
this stationary distribution
is independent of the swapping graph.
Intuitively,
this can be thought of
as a consequence of
the symmetric property
of the swapping relation.
Another implication of
Theorem~\ref{theo:pands} is that
the consequences of
the quasi-reversibility property
stated in Section~\ref{subsec:oi-steady}
for \gls{oi} queues
also apply to \gls{pands} queues.
As a result,
an open network of \gls{pands} queues
connected by a random routing process
has a product-form stationary distribution.
The only peculiarity is that,
when a customer completes service
in a \gls{pands} queue,
this customer is not necessarily the one
that leaves this queue.
The case of closed networks
is more complicated and
will be considered
in Section~\ref{sec:closed}.
The sketch of the proof given below
is completed in Appendix~\ref{app:pands}.

\begin{theorem} \label{theo:pands}
	The results of Theorem~\ref{theo:oi} remain valid
	if we replace ``\gls{oi} queue'' with ``\gls{pands} queue''.
\end{theorem}

\begin{proof}[Sketch of proof]
	The proof resembles that
	of Theorem~\ref{theo:oi},
	except that the second set
	of partial balance equations \eqref{eq:oi-partial-balance-2} has a different form.
	More specifically, we will verify that
	any measure $\pi$
	of the form~\eqref{eq:oi-pic}
	satisfies the following
	partial balance equations
	in each state
	$c = (c_1,\ldots,c_n) \in \I^*$:
	\begin{itemize}
		\item Equalize the flow
		out of state~$c$ due to a departure
		with the flow
		into state~$c$ due to an arrival (if $c \neq \emptyset$):
		\begin{equation} \label{eq:pands-partial-balance-1}
			\pi(c) \mu(c) = \pi(c_1,\ldots,c_{n-1}) \lambda_{c_n}.
		\end{equation}
		\item Equalize, for each $i \in \I$,
		the flow out of state~$c$
		due to the arrival of a class-$i$ customer
		with the flow into state~$c$
		due to the departure of
		a customer of this class:
		\begin{align} \label{eq:pands-partial-balance-2}
			\pi(c) \lambda_i
			= \sum_{d \in \I^*}
			\sum_{\substack{
					p = 1 \\
					\delta_p(d) = (c,i)
			}}^{n+1}
			\pi(d) \, \Delta\mu(d_1,\ldots,d_p).
		\end{align}
	\end{itemize}
	Since \eqref{eq:pands-partial-balance-1}
	represents the same set of equations as \eqref{eq:oi-partial-balance-1}, it is already known that \eqref{eq:oi-pic} satisfies \eqref{eq:pands-partial-balance-1}.
	Showing that \eqref{eq:oi-pic} satisfies \eqref{eq:pands-partial-balance-2} is equivalent to showing that the balance function~$\Phi$ in \eqref{eq:oi-Phic}
	satisfies the following equation
	in each state
	$c = (c_1, \ldots, c_n) \in \I^*$:
	\begin{align} \label{eq:pands-partial-balance-3}
		\Phi(c)
		= \sum_{d \in \I^*}
		\sum_{\substack{
				p = 1 \\
				\delta_p(d) = (c, i)
		}}^{n+1}
		\Phi(d) \, \Delta\mu(d_1,\ldots,d_p),
		\quad \forall i \in \I.
	\end{align}
	This is shown by induction
	over the queue length~$n$
	in Appendix~\ref{app:pands}.
	The rest of the proof follows
	along the same lines as Theorem~\ref{theo:oi}.
\end{proof}

\begin{remark}
	In this section, we have seen that the introduction of the \gls{pands} mechanism to the \gls{oi} queue induces a different set of partial balance equations, while product-form properties are retained. Even stronger, Theorem~\ref{theo:pands} shows that replacing~\eqref{eq:oi-partial-balance-2} with~\eqref{eq:pands-partial-balance-2} does not alter the stationary distribution of the queue at all. This begs the question of whether there exist other intra-queue routing mechanisms that also lead to this product-form stationary distribution. A partial answer can be found in the actual proof of Theorem~\ref{theo:pands}, given in Appendix~\ref{app:pands}. A careful analysis of this proof suggests that other routing mechanisms, resulting in a different completion order $q_1, q_2, \ldots, q_u$ as defined in Appendix~\ref{app:pands}, could also lead to the stationary measure given in Theorem~\ref{theo:oi}, as long as this completion order adheres to an equation of the same form as~\eqref{eq:proof-partial-balance-2}. However, identifying such routing mechanisms does not seem straightforward.
\end{remark}

\section{Complementary results on pass-and-swap queues}
\label{sec:results}

We now use Theorem~\ref{theo:pands}
to derive further results
on the stationary behavior
of \gls{pands} queues.
Section~\ref{subsec:results-stability}
gives an alternative stability condition
that is simpler to verify
than~\eqref{eq:oi-stability}.
The result of this section
extends that obtained for \gls{oi} queues
in the Ph.D.\ thesis \cite{C19-2}
and before that for multi-server queues in~\cite{BC17}.
In Section~\ref{subsec:results-rates},
we use the quasi-reversibility property
to prove that
the average service and departures rates
of each class are equal to each other
and independent of the swapping graph.
Recall that, since an \gls{oi} queue
is also a \gls{pands} queue,
the results of this section readily
apply to \gls{oi} queues.

\subsection{Stability condition}
\label{subsec:results-stability}

Theorem~\ref{theo:stability} below
gives a necessary and sufficient condition
for the stability of \gls{pands} queues.
This condition is simpler
than~\eqref{eq:oi-stability}
as it only compares
the per-class arrival rates
$\lambda_1, \ldots, \lambda_I$
to the rate function~$\mu$.
A first version of this theorem
was stated in the Ph.D.\ thesis \cite[Theorem~3.4]{C19-2}
for \gls{oi} queues.
The proof that we give in Appendix~\ref{app:stability} is different
in that it does not involve Whittle networks~\cite{S99}.

To state the stability condition,
it is worth recalling that,
for each $c \in \I^*$,
the macrostate associated with state~$c$
is the vector
$|c| = (|c|_1, \ldots, |c|_I) \in \N^I$
that counts the number of customers
of each class present in the queue.
Also recall that
the service rate $\mu(c)$ depends
on the number of customers
of each class that are contained
in state~$c$
but not on their order.
This means that,
once the macrostate corresponding
to a state is given,
the service rate function is not sensitive
to the state itself anymore.
Therefore, in the sequel, we also refer to
$\mu(c)$ as $\mu(|c|)$
for simplicity of notation.

For each~$i \in \I$,
we let $e_i$ denote the $I$-dimensional vector
with one in component~$i$ and zero elsewhere.
We define the function $\limmu$
on the power set of $\I$ by
\begin{equation} \label{eq:limmu}
	\limmu(\A) = \lim_{m \to +\infty} \mu(m e_\A),
	\quad \forall \A \subseteq \I,
\end{equation}
where $e_\A = \sum_{i \in \A} e_i$
for each $\A \subseteq \I$.
The monotonicity of $\mu$
ensures that $\limmu$ is well defined,
with values in $\R_+\cup \{+\infty\}$,
and is itself a non-decreasing set function.
If the overall service rate only depends
on the set of active classes,
as is the case in the multi-server queue
of Examples~\ref{ex:oi} and \ref{ex:pands},
we have $\bar\mu(\A) = \mu(x)$
for each $x \in \N^I$
such that $\A = \{i \in \I: x_i > 0\}$,
but in general,
we may have $\bar\mu(\A) > \mu(x)$
for each such $x$.

\begin{theorem} \label{theo:stability}
	Consider a \gls{pands} queue
	with a set $\I = \{1, \ldots, I\}$ of customer classes,
	per-class arrival rates $\lambda_1, \ldots, \lambda_I$,
	and a rate function $\mu$.
	This \gls{pands} queue is stable if and only if
	\begin{equation} \label{eq:stability}
		\sum_{i \in \A} \lambda_i < \limmu(\A),
		\quad \forall \A \subseteq \I: \A \neq \emptyset.
	\end{equation}
\end{theorem}

\begin{proof}
	See Appendix~\ref{app:stability}.
\end{proof}

\noindent This result is the only one
in this paper
that cannot be straightforwardly extended
to \gls{pands} queues
with an arbitrary scaling factor
as considered in \cite{BK95,K11}.
However, it can be extended
to \gls{pands} queues
with a non-decreasing scaling factor
by including this scaling rate
into the definition of $\limmu$.

\subsection{Departure and service rates}
\label{subsec:results-rates}

We now consider the relation between
the service rates and
departure rates of customers
in the \gls{pands} queue.
Consider a stable \gls{pands} queue,
as defined in Section~\ref{subsec:pands-def},
and let $\pi$ denote
the stationary distribution
of the Markov process tracking
the state over time.
For each $c = (c_1, \ldots, c_n) \in \I^*$
and $i \in \I$,
we also define
the overall departure rate of class~$i$
in state~$c$ to be
\begin{equation} \label{eq:thetac}
	\phi^d_i(c)
	= \sum_{\substack{p=1 \\ \delta_p(c) = i}}^n
	\Delta\mu(c_1,\ldots,c_p),
\end{equation}
while the overall service rate
of class~$i$ in this state is defined as
\begin{equation} \label{eq:phic}
	\phi^s_i(c) =
	\sum_{\substack{p=1 \\ c_p = i}}^n \Delta\mu(c_1,\ldots,c_p).
\end{equation}
While it is not necessarily true that
$\phi^d_i(c) = \phi^s_i(c)$,
the next proposition states that,
for each $x \in \N^I$ and $i \in \I$,
the overall probability flow
out of macrostate~$x$
due to a departure
of a class-$i$ customer
is equal to the overall probability
flow out of macrostate~$x$
due to a service completion
of a class-$i$ customer.

\begin{proposition} \label{prop:pands-oi}
	For each $x \in \N^I$ and $i \in \I$, we have
	\begin{equation} \label{eq:pands-oi}
		\sum_{\substack{c \in \I^*: |c|=x}} \pi(c) \phi^d_i(c)
		= \sum_{\substack{c \in \I^*: |c|=x}} \pi(c) \phi^s_i(c).
	\end{equation}
\end{proposition}

\begin{proof}
	If $x_i=0$, the result is immediate
	since $\phi^d_i(c) = \phi^s_i(c) = 0$
	for any state~$c$
	for which $|c| = x$.
	For the case $x_i>0$, note that
	the stationary distribution of
	the \gls{pands} queue satisfies the partial balance equations~\eqref{eq:oi-partial-balance-2}
	by Theorems~\ref{theo:oi} and~\ref{theo:pands}
	and~\eqref{eq:pands-partial-balance-2}
	by Theorem~\ref{theo:pands}.
	Therefore, the right-hand sides
	of~\eqref{eq:oi-partial-balance-2}
	and~\eqref{eq:pands-partial-balance-2} are equal.
	Equating these two sides and summing
	over all states~$c = (c_1, \ldots, c_n)$
	for which $|c| = x-e_i$, we obtain
	\begin{equation*}
		\sum_{\substack{c \in \I^*: \\ |c|=x-e_i}}
		\, \sum_{\substack{d \in \I^*: \\ |d|=x}}
		\, \sum_{\substack{
				p = 1 \\
				\delta_p(d) = (c,i)
		}}^{n+1}
		\pi(d) \, \Delta\mu(d_1,\ldots,d_p)
		=
		\sum_{\substack{c \in \I^*: \\ |c| = x - e_i}} \sum_{p=1}^{n+1}
		\pi(c_1,\ldots,c_{p-1},i,c_p,\ldots,c_{n})
		\Delta\mu(c_1,\ldots,c_{p-1},i).
	\end{equation*}
	Rewriting both sides leads to
	\begin{equation*}
		\sum_{\substack{c \in \I^*: |c| = x}}
		\, \sum_{\substack{p=1 \\ \delta_p(c)=i}}^n
		\pi(c) \, \Delta\mu(c_1,\ldots,c_p)
		=
		\sum_{\substack{c \in \I^*: |c| = x}}
		\, \sum_{\substack{p=1 \\ c_p = i}}^{n}
		\pi(c_1,\ldots,c_{n})
		\Delta\mu(c_1,\ldots, c_p).
	\end{equation*}
	Combining this equation with
	\eqref{eq:thetac} and \eqref{eq:phic}
	finalizes the proof.
\end{proof}

Since the overall probability
flow out of macrostate~$x$
due to a \emph{service completion}
of a class-$i$ customer
does not depend on the swapping graph,
this equality implies that,
for each $x \in \N^I$ and $i \in \I$,
the overall probability flow out of macrostate~$x$
due to a \emph{departure} of a class-$i$ customer
does not depend on the swapping graph.
Upon dividing~\eqref{eq:pands-oi}
by $\sum_{c \in \I^*: |c| = x} \pi(c)$,
we also obtain that
the conditional expected
departure and service rates of a class
given the macrostate
are equal to each other. Finally, summing both sides of \eqref{eq:pands-oi} over all $i \in \mathcal{I}$ shows that the aggregate departure rate of customers in any macrostate $x$ equals the aggregate service rate of customers in $x$. This can alternatively be seen to hold true by noting that a service completion in a macrostate $x$ also induces a departure from macrostate $x$, and departures from the system only occur because of service completions.

\begin{remark}
	The result of Proposition~\ref{prop:pands-oi}
	is intuitively not very surprising.
	If, in a state $c = (c_1, \ldots, c_n)$,
	the completion of customer $c_p$ will trigger
	a departure of customer $c_q$,
	then, in state $d = (c_n, \ldots, c_1)$,
	the completion of customer $c_q$ will trigger
	a departure of customer $c_p$,
	as the swapping graph is undirected.
	Moreover, both states lead to the same macrostate~$x$.
	The fact that $\pi(c) \neq \pi(d)$ is offset by
	the nature of the order-independent service rates,
	as formalized in the proof
	of Proposition~\ref{prop:pands-oi}.
\end{remark}

\section{Closed models} \label{sec:closed}

We saw that \gls{pands} queues are quasi-reversible, so that stable open networks of \gls{pands} queues have a product-form stationary distribution under mild conditions on the routing process.
In this section, we consider \emph{closed} networks of \gls{pands} queues in more detail and conclude that, also for closed networks, the stationary distribution has a product form. In contrast to open networks, this does not follow directly from quasi-reversibility since, in general, the obtained Markov process does not meet the irreducibility assumptions posed in~\cite[Section 3.4]{kelly}. In Section~\ref{subsec:one}, we first consider a closed \gls{pands} queue in which the number of customers of each class is fixed and departing customers are appended back to the end of the queue instead of leaving. 
These results are extended to a closed tandem network of two \gls{pands} queues in Section~\ref{subsec:two}. This tandem network turns out to have rich applications, as we will see in Section~\ref{sec:app}.

\subsection{A closed pass-and-swap queue}
\label{subsec:one}

We first consider a closed network
that consists of a single \gls{pands} queue.
In Section~\ref{subsubsec:one-example},
we give an example of such
a closed \gls{pands} queue
to illustrate its dynamics. Section~\ref{subsubsec:one-model} then
gives a more formal description of this model,
including necessary notation.
This section also
studies the structure of the Markov process
underlying this closed \gls{pands} queue
and establishes sufficient conditions
for this Markov process
to be irreducible.
Provided that the Markov process is indeed irreducible, Section~\ref{subsubsec:one-steady} provides the stationary distribution
of the closed \gls{pands} queue
and establishes its product-form nature.

\subsubsection{Introductory example}
\label{subsubsec:one-example}

Consider a closed \gls{pands} queue with six customer classes ($\I = \{1,2,\ldots,6\}$) and the swapping graph shown in \figurename~\ref{fig:toy-hasse}. When a class-$i$ customer departs the queue, this customer does not leave the system. Instead, it is appended back to the queue as a class-$i$ customer. For simplicity, we assume that there is a single customer of each class in the queue. More precisely, we assume that the queue starts in state $c = (1,2,3,4,5,6)$, as depicted in \figurename~\ref{subfig:toy-one-1}.
A possible sequence of transitions is shown in \figurename s~\ref{subfig:toy-one-2} and \ref{subfig:toy-one-3}. Each transition is triggered by the service completion of the customer that is currently at the head of the queue. In particular, in the transition from \figurename~\ref{subfig:toy-one-1} to \figurename~\ref{subfig:toy-one-2}, customer~1 completes service, and this customer replaces customer~3, which replaces customer~6 in accordance with the swapping graph in \figurename~\ref{fig:toy-hasse}. In the transition from \figurename~\ref{subfig:toy-one-2} to \figurename~\ref{subfig:toy-one-3}, customer~2 completes service, and this customer replaces customer~4, which replaces customer~6. In both cases, customer~6 is appended back to the end of the queue, in the last position, so that this customer's position remains unchanged by the transition.

\begin{figure}[ht]
	\begin{minipage}{.45\linewidth}
		\centering
		\subfloat[Swapping graph. \label{fig:toy-hasse}]{%
			\begin{tikzpicture}
				\def\width{2cm}
				\def\height{1.6cm}
				
				\node[class] (1) {1};
				\node[class] (2)
				at ($(1)+(\width,0)$) {2};
				\node[class] (3)
				at ($(1)-(0.5*\width,0)+(0,\height)$)
				{3};
				\node[class] (4)
				at ($(1)!.5!(2)+(0,\height)$) {4};
				\node[class] (5)
				at ($(2)+(0.5*\width,0)+(0,\height)$)
				{5};
				\node[class] (6)
				at ($(4)+(0,\height)$) {6};
				
				\draw (6) -- (3) -- (1) -- (4) -- (2) -- (5) -- (6) -- (4);
			\end{tikzpicture}
		}
	\end{minipage}
	\begin{minipage}{.50\linewidth}
		\centering
		\subfloat[Initial state. \label{subfig:toy-one-1}]{%
			\begin{tikzpicture}
				\def\width{1.2cm}
				\def\height{1.8cm}
				
				\node[fcfs=6] (queue) {
					\nodepart{one}{6}
					\nodepart{two}{5}
					\nodepart{three}{4}
					\nodepart{four}{3}
					\nodepart{five}{2}
					\nodepart{six}{1}
				};
				
				\node[server, anchor=west] (mu) at ($(queue.east)+(.1cm,0)$) {};
				
				\draw[->] ($(mu.east)+(.1cm,0)$) -- ($(mu.east)+(.3cm,0)$)
				|- ($(queue.south west)-(.3cm,.2cm)$) -- ($(queue.west)-(.3cm,0)$)
				-- ($(queue.west)-(.1cm,0)$);
				
				\node at ($(queue.west)-(2cm,0)$) {};
				\node at ($(mu.east)+(2cm,0)$) {};
			\end{tikzpicture}
		}
		\hfill
		\subfloat[State reached after the service completion of customer~$1$. \label{subfig:toy-one-2}]{%
			\begin{tikzpicture}
				\def\width{1.2cm}
				\def\height{1.8cm}
				
				\node[fcfs=6] (queue) {
					\nodepart{one}{6}
					\nodepart{two}{3}
					\nodepart{three}{5}
					\nodepart{four}{4}
					\nodepart{five}{1}
					\nodepart{six}{2}
				};
				
				\node[server, anchor=west] (mu) at ($(queue.east)+(.1cm,0)$) {};
				
				\draw[->] ($(mu.east)+(.1cm,0)$) -- ($(mu.east)+(.3cm,0)$)
				|- ($(queue.south west)-(.3cm,.2cm)$) -- ($(queue.west)-(.3cm,0)$) -- ($(queue.west)-(.1cm,0)$);
				
				\node at ($(queue.west)-(2cm,0)$) {};
				\node at ($(mu.east)+(2cm,0)$) {};
			\end{tikzpicture}
		}
		\hfill
		\subfloat[State reached after the service completion of customer $2$. \label{subfig:toy-one-3}]{%
			\begin{tikzpicture}
				\def\width{1.2cm}
				\def\height{1.8cm}
				
				\node[fcfs=6] (queue) {
					\nodepart{one}{6}
					\nodepart{two}{4}
					\nodepart{three}{3}
					\nodepart{four}{5}
					\nodepart{five}{2}
					\nodepart{six}{1}
				};
				
				\node[server, anchor=west] (mu) at ($(queue.east)+(.1cm,0)$) {};
				
				\draw[->] ($(mu.east)+(.1cm,0)$) -- ($(mu.east)+(.3cm,0)$)
				|- ($(queue.south west)-(.3cm,.2cm)$) -- ($(queue.west)-(.3cm,0)$) -- ($(queue.west)-(.1cm,0)$);
				
				\node at ($(queue.west)-(2cm,0)$) {};
				\node at ($(mu.east)+(2cm,0)$) {};
			\end{tikzpicture}
		}
	\end{minipage}
	\caption{A closed \gls{pands} queue.
		The rate function need not be specified,
		as throughout this example we only
		consider the service completion
		of the customer at the head of the queue.}
	\label{fig:one}
\end{figure}
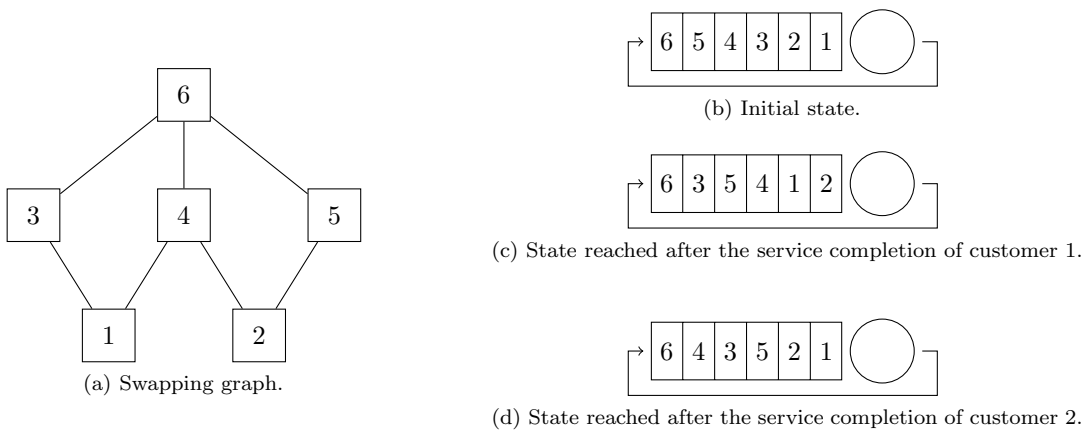

It is worth noting that, in the three states shown in \figurename~\ref{fig:one}, customer~$1$ precedes customers~$3$ and $4$, which precede customer~$6$. This order will be conserved by the \gls{pands} mechanism, as the service completion of customer~$1$ systematically triggers the movement of either customer~$3$ or customer~$4$,
that in turn will replace customer~$6$. Similarly, customer $2$ will always precede customers~$4$ and $5$ that, in their turn, will always precede customer~$6$.
In the sequel, we will formalize this phenomenon
and characterize the communicating classes of the Markov process
associated with the state of a closed \gls{pands} queue.

\subsubsection{Queueing model}\label{subsubsec:one-model}

The closed \gls{pands} queue inherits virtually
all properties and notation from Section~\ref{subsec:pands-def}.
As mentioned before, the only difference is that, upon a service completion, the customer that would have left if the queue were open is, instead, appended back to the end of the queue as a customer of the same class.
There is also no external arrival process, so that the macrostate of the queue is determined by its initial state and does not change over time. We therefore let $\I = \{1, \ldots, I\}$ denote the set of classes of the customers in the initial state of the queue. The (fixed) macrostate of the queue is denoted by $\ell = (\ell_1, \ldots, \ell_I) \in \N^I$ and the total number of customers by $n = \ell_1 + \ldots + \ell_I$. Note that $\ell_i>0$ for each $i \in \I$, as we only consider classes of customers present in the queue.

In Section~\ref{subsubsec:one-example}, we found that it was possible for an ordering of customers to be preserved by the \gls{pands} mechanism. To help formalize this phenomenon, we introduce the notion of a \emph{placement order}.
We first define a
\emph{placement graph} of the queue
as an acyclic orientation
of its swapping graph, that is,
a directed acyclic graph obtained
by assigning an orientation
to each edge of the swapping graph.
This is only possible if the swapping graph contains no loop, which we assume in the remainder of Sections~\ref{sec:closed} and \ref{sec:app}.
A \emph{placement order} of the queue
is then defined as
(the strict partial order associated with)
the reachability relationship
of one of its placement graphs.
In other words,
a strict partial order $\prec$ on $\I$
is said to be a placement order
if there exists a placement graph
such that, for each $i, j \in \I$ with $i \neq j$,
$i \prec j$ if and only if
there is a directed path
from class~$i$ to class~$j$
in the placement graph.
It will be useful later to observe that,
for each classes $i, j \in \I$
that are neighbors in the swapping graph,
we have either $i \prec j$ or $j \prec i$.

We say that a state
$c=(c_1, \ldots, c_n)\in\mathcal{I}^n$
\emph{adheres} to the placement order
if $c_q \nprec c_p$
for each $p, q \in \{1, \ldots, n\}$
such that $p < q$.
Since the placement order is only partial,
there may be pairs of classes
for which neither $c_p \prec c_q$
nor $c_q \prec c_p$ hold.
Adherence is therefore
a weaker property than
having $c_p \prec c_q$
for each $p, q \in \{1, \ldots, n\}$
such that $p < q$.
As a special case,
adherence allows that $c_p = c_q$ when $p < q$.

\begin{example} \label{ex:one}
	We consider the closed \gls{pands} queue
	of Section~\ref{subsubsec:one-example}.
	The placement graph in \figurename~\ref{fig:placementGraphInitialExample}
	is obtained by orienting
	the edges of the swapping graph
	of \figurename~\ref{fig:toy-hasse}
	from bottom to top.
	All states in \figurename~\ref{fig:one}
	adhere to the corresponding placement order.
	For example, the placement graph implies that
	$1 \prec j$ for $j \in \{3,4,6\}$
	and $2 \prec j$ for $j \in \{4,5,6\}$,
	which in turn implies that the customer at the front
	of the queue is either customer~1 or customer~2.
	All states in \figurename~\ref{fig:one}
	indeed satisfy this property.
	Customers~1 and 2 can alternate positions,
	as neither $1 \prec 2$ nor $1 \succ 2$.
\end{example}

\begin{figure}[ht]
	\centering
	\begin{tikzpicture}
		\def\width{2cm}
		\def\height{1.6cm}
		
		\node[class] (1) {1};
		\node[class] (2)
		at ($(1)+(\width,0)$) {2};
		\node[class] (3)
		at ($(1)-(0.5*\width,0)+(0,\height)$)
		{3};
		\node[class] (4)
		at ($(1)!.5!(2)+(0,\height)$) {4};
		\node[class] (5)
		at ($(2)+(0.5*\width,0)+(0,\height)$)
		{5};
		\node[class] (6)
		at ($(4)+(0,\height)$) {6};
		
		\draw[->] (1) -- (3);
		\draw[->] (1) -- (4);
		\draw[->] (2) -- (4);
		\draw[->] (2) -- (5);
		\draw[->] (3) -- (6);
		\draw[->] (4) -- (6);
		\draw[->] (5) -- (6);
	\end{tikzpicture}
	\caption{The placement graph corresponding
		to \figurename~\ref{fig:one}.}
	\label{fig:placementGraphInitialExample}
\end{figure}
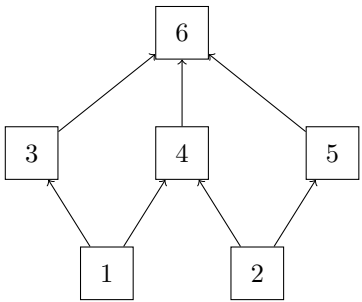

In general,
not all possible states of
a closed \gls{pands} queue
adhere to a placement order.
In Example~\ref{ex:one},
if the initial queue state
is $(3,1,2,3,4,5,6)$,
a class-1 customer is both preceded
and succeeded by a class-3 customer,
making it impossible to
orient the edge between classes~$1$ and $3$
in the swapping graph.
Furthermore, each state
can only adhere to \emph{at most}
one placement order,
so that the sets of states
that adhere to different placement orders
are disjoint.
To prove this,
it suffices to observe that,
for each state
that adheres to a placement order,
the relative placement of customers
within the state
specifies the orientation of all edges of the swapping graph,
which in turn uniquely defines
a placement order.

In the rest of this section
and in Section~\ref{subsubsec:one-steady},
we focus on the case where
the initial state
of the queue does adhere to
a placement order.
Proposition~\ref{prop:one-order-closed}
describes the phenomenon encountered
in Section~\ref{subsubsec:one-example}
in full generality, while
Proposition~\ref{prop:one-order-irreducible}
provides a stronger result,
assuming that all customers
receive a positive service rate.
The proofs of these two propositions are given in Appendix~\ref{app:closed}.
The case of states
that do not adhere to a placement order
is treated in Appendix~\ref{app:irreducibility}.

\begin{proposition} \label{prop:one-order-closed}
	If the initial state of the closed \gls{pands} queue adheres to the placement order $\prec$, then any state reached by applying the \gls{pands} mechanism also adheres to this placement order.
\end{proposition}

\begin{proposition} \label{prop:one-order-irreducible}
	Assume that $\Delta\mu(c) > 0$
	for each $c \in \I^*$.
	All states that
	adhere to the same placement order
	and correspond to the same macrostate
	form a single closed communicating class
	of the Markov process
	associated with the queue state.
\end{proposition}

\begin{remark}
	The assumption in
	Proposition~\ref{prop:one-order-irreducible},
	namely that $\Delta\mu(c) > 0$
	for each $c \in \I^*$,
	is a sufficient condition
	for this result to hold
	but it is not a necessary condition.
	It is for example worth noting that
	this assumption is not satisfied
	by the multi-server queue
	of Example~\ref{ex:oi},
	yet the closed variant of
	this queue satisfies the conclusions
	of these two propositions
	whenever
	$\mu_1$, $\mu_2$, and $\mu_3$
	are positive.
	In general, the construction of weaker sufficient conditions appears to be challenging since the transition described
	in step~3 of the algorithm in the proof of Proposition~\ref{prop:one-order-irreducible} is not guaranteed to occur with a positive probability when there are states~$c$ so that $\Delta\mu(c) = 0$.
\end{remark}

\subsubsection{Stationary analysis}
\label{subsubsec:one-steady}

We now turn to the stationary distribution
of the Markov process underlying
the closed \gls{pands} queue
and establish its product-form nature.
Recall that the initial macrostate
of the queue equals $\ell$,
which cannot change over time due to
the closed nature of the queue.
We assume that the initial state
adheres to a placement order $\prec$.
Since all subsequent states
must also adhere to this placement order
due to Proposition~\ref{prop:one-order-closed},
we restrict the state space of the Markov process
to the state space~$\C$ that consists
of all states $c = (c_1, \ldots, c_n)$
that satisfy $|c| = \ell$
and adhere to the placement order~$\prec$.
The rate function~$\mu$
and balance function~$\Phi$ of the queue
are assumed to be defined
on the whole set $\I^*$ for simplicity,
although we could just as well define them
on a subset of $\I^*$.
Theorem~\ref{theo:one-picd} below provides
the stationary distribution
of the closed \gls{pands} queue
and reveals its product form nature.

\begin{theorem} \label{theo:one-picd}
	Assume that the Markov process
	associated with the state
	of the closed \gls{pands} queue,
	with state space $\C$, is irreducible.
	The stationary distribution
	of this Markov process
	is then given by
	\begin{equation} \label{eq:one-pic}
		\pi(c)
		= \frac{\Phi(c)}
		{\sum_{d \in \C} \Phi(d)},
		\quad \forall c \in \C,
	\end{equation}
	where the function $\Phi$
	is given by~\eqref{eq:oi-Phic}.
\end{theorem}

\begin{proof}
	It suffices to show that the function~$\Phi$ satisfies the balance equations of the Markov process, after which the result follows by normalization.
	Let $c = (c_1, \ldots, c_n) \in \C$
	and $i = c_n$.
	Since a departing customer immediately re-enters the queue as
	a customer of the same class,
	the balance equation for
	any state $c \in \C$ reads
	\begin{equation} \label{eq:one-balance-1}
		\pi(c) \, \mu(c)
		= \sum_{d \in \C}
		\sum_{\substack{
				p = 1 \\
				\delta_p(d)
				= ((c_1, \ldots, c_{n-1}), i)
		}}^n
		\pi(d)
		\, \Delta\mu(d_1, \ldots, d_p),
	\end{equation}
	where we write
	$\delta_p(d) = ((c_1, \ldots, c_{n-1}), i)$
	if, in the open queue, the service completion of the customer in position~$p$ in state $d$ would lead to state $(c_1, \ldots, c_{n-1})$ with a departure of a class-$i$ customer.
	It follows from
	Proposition~\ref{prop:one-order-closed} that
	the set $\C$ contains all states
	$d \in \I^*$ such that
	$\delta_p(d) = ((c_1, \ldots, c_{n-1}), i)$
	for some $p \in \{1, \ldots, n\}$.
	Therefore, it suffices to prove that
	the balance function~$\Phi$ satisfies
	\begin{equation} \label{eq:one-balance-2}
		\Phi(c) \, \mu(c)
		= \sum_{d \in \I^*}
		\sum_{\substack{
				p = 1 \\
				\delta_p(d)
				= ((c_1, \ldots, c_{n-1}),i)
		}}^n
		\Phi(d) \, \Delta\mu(d_1, \ldots, d_p).
	\end{equation}
	By applying~\eqref{eq:pands-partial-balance-3}
	to state~$(c_1, \ldots, c_{n-1})$
	and class~$i$,
	we obtain that $\Phi(c_1, \ldots, c_{n-1})$,
	as defined in~\eqref{eq:oi-Phic},
	is equal to the right-hand side
	of~\eqref{eq:one-balance-2}.
	To conclude, it suffices to observe
	that~\eqref{eq:oi-Phic} implies
	$\Phi(c) \mu(c)
	= \Phi(c_1, \ldots, c_{n-1})$.
\end{proof}

\begin{remark} \label{remark1}
	According to
	Proposition~\ref{prop:one-order-irreducible},
	a sufficient condition for
	the Markov process
	considered in Theorem~\ref{theo:one-picd}
	to be irreducible
	is that $\Delta\mu(c) > 0$
	for each $c \in \I^*$.
	If this process is not irreducible,
	all steps of the proof of
	Theorem~\ref{theo:one-picd} remain valid,
	so that the distribution defined
	by~\eqref{eq:one-pic}
	is still a stationary distribution
	of the Markov process, but it may not be the only one.
	Since $\Phi(c) > 0$ for each $c \in \I^*$
	by \eqref{eq:oi-Phic},
	this observation shows that
	the Markov process
	considered in Theorem~\ref{theo:one-picd}
	always has a positive stationary distribution,
	which implies that this process
	has no transient state,
	that is, all its
	communicating classes are closed.
\end{remark}

\begin{remark} \label{remark2}
	A variant of Theorem~\ref{theo:one-picd}
	can also be derived for closed \gls{pands} queues with initial states	that do not adhere to a placement order. 
	We have deferred derivation of this more general result to Appendix~\ref{app:irreducibility} to simplify the discussion.
	Theorem~\ref{theo:one-picd}	will be sufficient
	for the applications
	of Section~\ref{sec:app}.
\end{remark}

\subsection{A closed tandem network of two pass-and-swap queues}
\label{subsec:two}

Now that the product-form
of the stationary distribution
of a single closed \gls{pands} queue
has been established,
we turn to the study of
a closed tandem of two \gls{pands} queues.
Again, we first explain the model
through an introductory example
in Section~\ref{subsubsec:two-example},
after which we formalize the model
and describe structural properties
in Section~\ref{subsubsec:two-model}. 
We also derive
the stationary distribution
in Section~\ref{subsubsec:two-steady}.

\subsubsection{Introductory example}
\label{subsubsec:two-example}

We consider the closed tandem network of two \gls{pands} queues depicted in \figurename~\ref{fig:dlb-toy}.
\begin{figure}[b]
	\centering
	\subfloat[Initial state. \label{subfig:toy-two-1}]{
		\begin{tikzpicture}
			\def\width{1cm}
			\def\height{1.8cm}
			
			\node[fcfs=6] (queue1) {
				\nodepart{one}{6}
				\nodepart{two}{5}
				\nodepart{three}{4}
				\nodepart{four}{3}
				\nodepart{five}{2}
				\nodepart{six}{1}
			};
			
			\node[server, anchor=west] (mu1) at ($(queue1.east)+(.1cm,0)$) {};
			
			\node[fcfs=6] (queue2) at ($(queue1)+(\width,0)-(0,\height)$) {
				\nodepart{one}{\phantom{1}}
				\nodepart{two}{\phantom{1}}
				\nodepart{three}{\phantom{1}}
				\nodepart{four}{\phantom{1}}
				\nodepart{five}{\phantom{1}}
				\nodepart{six}{\phantom{1}}
			};
			
			\node[server, anchor=east] (mu2) at ($(queue2.west)-(.1cm,0)$) {};
			
			\draw[->] ($(mu2.west)-(.1cm,0)$) -- ($(mu2.west)-(.5cm,0)$) |- ($(queue1.west)-(.1cm,0)$);
			\draw[->] ($(mu1.east)+(.1cm,0)$) -- ($(mu1.east)+(.5cm,0)$) |- ($(queue2.east)+(.1cm,0)$);
			
		\end{tikzpicture}
	}
	\hfill
	\subfloat[State reached after
	the service completion of customer~1.
	\label{subfig:toy-two-2}]{
		\begin{tikzpicture}
			\def\width{1cm}
			\def\height{1.8cm}
			
			\node[fcfs=6] (queue1) {
				\nodepart{one}{\phantom{1}}
				\nodepart{two}{3}
				\nodepart{three}{5}
				\nodepart{four}{4}
				\nodepart{five}{1}
				\nodepart{six}{2}
			};
			
			\node[server, anchor=west] (mu1)
			at ($(queue1.east)+(.1cm,0)$) {};
			
			\node[fcfs=6] (queue2) at ($(queue1)+(\width,0)-(0,\height)$) {
				\nodepart{one}{6}
				\nodepart{two}{\phantom{1}}
				\nodepart{three}{\phantom{1}}
				\nodepart{four}{\phantom{1}}
				\nodepart{five}{\phantom{1}}
				\nodepart{six}{\phantom{1}}
			};
			
			\node[server, anchor=east] (mu2)
			at ($(queue2.west)-(.1cm,0)$) {};
			
			\draw[->] ($(mu2.west)-(.1cm,0)$)
			-- ($(mu2.west)-(.5cm,0)$)
			|- ($(queue1.west)-(.1cm,0)$);
			\draw[->] ($(mu1.east)+(.1cm,0)$)
			-- ($(mu1.east)+(.5cm,0)$)
			|- ($(queue2.east)+(.1cm,0)$);
		\end{tikzpicture}
	}
	\hfill
	\subfloat[State reached after
	service completions
	of several customers
	at the head
	of the first queue. \label{subfig:toy-two-3}]{
		\begin{tikzpicture}
			\def\width{1cm}
			\def\height{1.8cm}
			
			\node[fcfs=6] (queue1) {
				\nodepart{one}{\phantom{1}}
				\nodepart{two}{\phantom{1}}
				\nodepart{three}{\phantom{1}}
				\nodepart{four}{1}
				\nodepart{five}{5}
				\nodepart{six}{2}
			};
			
			\node[server, anchor=west] (mu1) at ($(queue1.east)+(.1cm,0)$) {};
			
			\node[fcfs=6] (queue2) at ($(queue1)+(\width,0)-(0,\height)$) {
				\nodepart{one}{6}
				\nodepart{two}{4}
				\nodepart{three}{3}
				\nodepart{four}{\phantom{1}}
				\nodepart{five}{\phantom{1}}
				\nodepart{six}{\phantom{1}}
			};
			
			\node[server, anchor=east] (mu2) at ($(queue2.west)-(.1cm,0)$) {};
			
			\draw[->] ($(mu2.west)-(.1cm,0)$)
			-- ($(mu2.west)-(.5cm,0)$)
			|- ($(queue1.west)-(.1cm,0)$);
			\draw[->] ($(mu1.east)+(.1cm,0)$)
			-- ($(mu1.east)+(.5cm,0)$)
			|- ($(queue2.east)+(.1cm,0)$);
		\end{tikzpicture}
	}
	\caption{A closed tandem network of two \gls{pands} queues.
		As in \figurename~\ref{fig:one},
		the rate function is not specified because
		we will only consider
		the service completion of the customer
		at the head of a queue.}
	\label{fig:dlb-toy}
\end{figure}
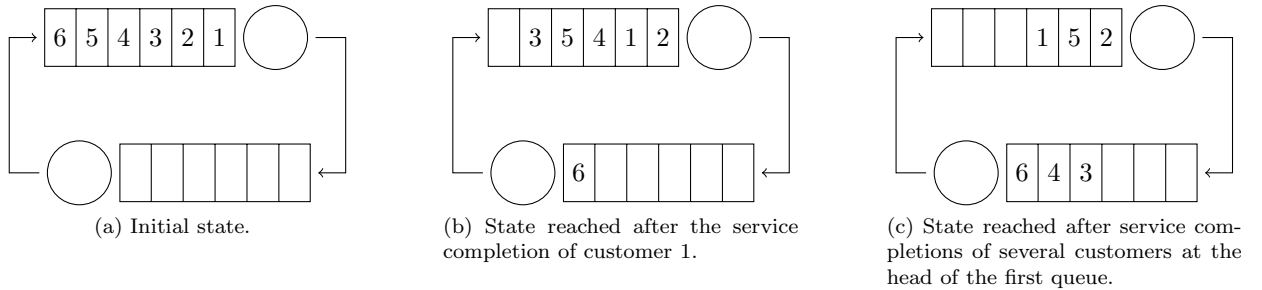%
Both queues have the same set of customer classes and the same swapping graph as the closed queue of Section~\ref{subsec:one}.
Furthermore, only the customer at the head of each queue, if any, receives a positive service rate. We also assume that there is a single customer of each class in the network.
The routing process is as follows: for each $i \in \I$, if a class-$i$ customer departs from a queue, this customer is routed to the back of the \emph{other} queue as a class-$i$ customer. As shown in \figurename~\ref{subfig:toy-two-1}, the initial state of the first queue is $c = (1,2,3,4,5,6)$ and that of the second queue is $d = \emptyset$.
\figurename~\ref{subfig:toy-two-2} shows the state reached after customer~1 (that is, the only customer belonging to class 1) completes service. As in the introductory example of Section~\ref{subsec:one}, this customer replaces customer~$3$, and customer~$3$ replaces customer~$6$.
The difference is that customer~$6$ is now routed to the second queue rather than to the first. \figurename~\ref{subfig:toy-two-3} shows the state reached after several service completions, each time that of the customer at the head of the first queue.

As in Section~\ref{subsubsec:one-example},
the order of customers seems to be preserved
by the \gls{pands} mechanism.
However, the orders of customers
in the two queues are reversed.
For instance, while customer~$6$
comes after customers~$3$ and $4$
in the first queue
(as in Section~\ref{subsubsec:one-example}),
in the second queue customer $6$ always
precedes these customers.
We will show that this symmetry
in customer orders
holds for any closed tandem network
of two \gls{pands} queues.

Before we move on to the analysis,
let us motivate this model by giving a glimpse
of the applications of Section~\ref{sec:app}.
This section is concerned with
token-based load-distribution protocols
for server systems with assignment constraints,
in which each job seizes a token upon arrival
and releases its token upon departure.
The above-defined closed tandem network
models the dynamics of tokens as follows.
The first queue contains available tokens,
while the second queue
contains tokens held by jobs in the system
(either waiting to be served or in service).
Service completions in the first and second queues
correspond to job arrivals and departures,
respectively, and
the \gls{pands} mechanism translates to
a protocol to distribute load across servers.

\subsubsection{Queueing model}
\label{subsubsec:two-model}

Just like in Section~\ref{subsec:one}, the queues in the closed tandem network are ordinary \gls{pands} queues as described in Section~\ref{sec:pands}. There are however no external arrivals or departures: the departure process of one queue now forms the arrival process of the other. In particular, both queues share the same set $\I = \{1, \ldots, I\}$ of customer classes. The state of the first queue is denoted by $c = (c_1, \ldots, c_n)\in\I^*$ and that of the second queue by $d = (d_1, \ldots, d_m)\in\I^*$, where $n$ and $m$ are the number of customers present in the first and second queue, respectively. We refer to $(c; d) = (c_1, \ldots, c_n; d_1, \ldots, d_m)$ as the state of the network. We furthermore assume that both queues have the same swapping graph. The rate functions of the two queues may however differ: whereas customers in the first queue complete service according to the rate function $\mu$, the rate function in the second queue is denoted by $\nu$.
Because of the closedness of the system,
the sum of the macrostates of the two queues, denoted by $\ell = |c|+|d|$, is constant over time.
We refer to this vector $\ell = (\ell_1, \ldots, \ell_{I})$ as the network macrostate and we assume without loss of generality that $\ell_i > 0$ for every $i \in \I$.
Macrostates $|c|$ and $|d|$, however,
do fluctuate over time.
We let $\Phi$ (resp.\ $\Lambda$) denote the balance function
of the first (resp.\ second) queue.
The functions $\mu$, $\nu$, $\Phi$, and $\Lambda$
are assumed to be defined
on $\I^*$ for simplicity.

As mentioned above, both queues have the same swapping graph.
Similarly to Section~\ref{subsubsec:one-model},
we define a \emph{placement order} $\prec$
of the network
by directing the edges
of this swapping graph
so that a directed acyclic graph,
called the \emph{placement graph} of the network,
arises;
we again write $i \prec j$
if and only if there exists a directed path
from class~$i$ to class~$j$ in the placement graph.
We now say that a network state $(c;d)$
adheres to the placement order~$\prec$ if
the following three conditions are satisfied:
\begin{enumerate}[label=(\roman*)]
	\item \label{order-1} $(c_1, \ldots, c_n)$ adheres
	to the placement order~$\prec$
	in the sense of Section~\ref{subsubsec:one-model},
	\item \label{order-2} $(d_m, \ldots,d_1)$ adheres
	to the placement order~$\prec$
	in the sense of Section~\ref{subsubsec:one-model}, and
	\item \label{order-3} $c_p \nsucc d_q$ for each
	$p \in \{1, \ldots, n\}$
	and $q \in \{1, \ldots, m\}$.
\end{enumerate}
Reversing the order of state~$d$
in property~\ref{order-2}
is consistent with the observation
of Section~\ref{subsubsec:two-example}
that the order of customers
in the second queue is reversed
compared to the first queue.
Equivalently, we say that the network state $(c;d)$
adheres to the placement order $\prec$ if and only if the state $(c_1, \ldots, c_n, d_m, \ldots, d_1)$ would adhere to $\prec$ in the single-queue setting of Section \ref{subsubsec:one-model}.
For example, the three network states
shown in \figurename~\ref{fig:dlb-toy}
adhere to the placement order $\prec$
defined by the placement graph
of \figurename~\ref{fig:placementGraphInitialExample}.
Focusing on the state shown in
\figurename~\ref{subfig:toy-two-3},
we have $c = (2,5,1)$ and $d = (6,4,3)$.
The equivalent state in Section \ref{subsubsec:one-model} would be $(2,5,1,3,4,6)$; note the absence of a semicolon and the reverse order of the customers of the second queue. It is indeed easily verified that this state adheres to the same placement order $\prec$ in \figurename~\ref{fig:dlb-toy}.
The definition of adherence in the two-queue setting is symmetric
with respect to the queues in the sense that
state~$(c;d)$ adheres
to the placement order~$\prec$
if and only if
state~$(d;c)$ adheres
to the reverse placement order~$\succ$
defined as follows: for each $i, j \in \I$,
$i \succ j$ if and only if $j \prec i$.

As for the case of a single queue,
we focus here on the case where the initial state
adheres to a placement order.
Propositions~\ref{prop:two-order-closed}
and \ref{prop:two-order-irreducible}
below are the counterparts
of Propositions~\ref{prop:one-order-closed}
and \ref{prop:one-order-irreducible}
for closed tandem networks of two queues.
The proofs of these two propositions
are given in Appendix~\ref{app:closed} and
rely on the proof of their single-queue counterparts.

\begin{proposition}\label{prop:two-order-closed}
	If the initial network state
	adheres to the placement order $\prec$,
	then any state reached
	by applying the \gls{pands} mechanism
	to either of the two queues
	also adheres to this placement order.
\end{proposition}

\begin{proposition} \label{prop:two-order-irreducible}
	Assume that either
	$\Delta\mu(c) > 0$ for each $c \in \I^*$
	or $\Delta\nu(d) > 0$ for each $d \in \I^*$ (or both).
	All states that adhere
	to the same placement order
	and correspond to the same macrostate
	form a single closed communicating class
	of the Markov process
	associated with the network state.
\end{proposition}

\subsubsection{Stationary analysis}
\label{subsubsec:two-steady}

We now derive the stationary distribution
of the Markov process associated with
the network state.
As in Section~\ref{subsubsec:one-steady},
we focus on the special case
where the initial state
adheres to the placement order~$\prec$,
so that, by Proposition~\ref{prop:two-order-closed},
all subsequent states also adhere
to this placement order.
Therefore, we restrict
the state space of the Markov process
to the set~$\Sigma$
that consists of all network states~$(c;d)$
that adhere to the placement order~$\prec$
and satisfy $|c| + |d| = \ell$,
where $\ell$ denotes the initial network macrostate.
Theorem~\ref{theo:two-steady-cd} below
gives the stationary distribution
of the Markov process
associated with the network state.

\begin{theorem} \label{theo:two-steady-cd}
	Assume that the Markov process
	associated with the state
	of the closed tandem network,
	with state space $\Sigma$,
	is irreducible.
	The stationary distribution
	of this Markov process
	is then given by
	\begin{equation} \label{eq:two-picd}
		\pi(c;d) = \frac1G \Phi(c) \Lambda(d),
		\quad \forall (c;d) \in \Sigma,
	\end{equation}
	where $\Phi$ and $\Lambda$ are
	the balance functions of
	the first and second queues,
	respectively,
	and the normalization constant~$G$
	is given by
	\begin{equation} \label{eq:two-Gcd}
		G = \sum_{(c;d) \in \Sigma}
		\Phi(c) \Lambda(d).
	\end{equation}
\end{theorem}
\begin{proof}
	Before writing down the balance equations,
	we introduce some useful notation.
	Let $\X$ denote the subset of $\N^I$
	that consists of
	the vectors $x = (x_1, \ldots, x_I)$
	such that $x \le \ell$ and,
	for each $i, j \in \I$ with $i \prec j$,
	$x_j = 0$ whenever $x_i = 0$.
	This is the set of possible
	macrostates of the first queue.
	For each $x \in \X$,
	let $\C_x$ denote the set
	of states $c = (c_1, \ldots, c_n) \in \I^*$
	that adhere to the placement order
	and satisfy $|c| = x$.
	The set of possible states
	of the first queue is
	$\C = \bigcup_{x \in \X_\ell} \C_x$.
	Similarly, let
	$\Y$ denote the subset of $\N^I$
	that consists of
	the vectors $y = (y_1, \ldots, y_I)$
	such that $|y| \le \ell$ and,
	for each $i, j \in \I$
	with $i \prec j$,
	$y_i = 0$ whenever $y_j = 0$.
	This is the set of possible macrostates
	of the second queue.
	Also, for each $y \in \Y$, let
	$\D_y$ denote the set of
	states $d = (d_1, \ldots, d_m)$
	such that $(d_m, \ldots, d_1)$
	adheres to the placement order
	and $|d| = y$.
	The set of possible states
	of the second queue is
	$\D = \bigcup_{y \in \Y_\ell} \D_y$.
	As a result, the state space $\Sigma$ 
	can be partitioned as follows:
	\begin{equation} \label{eq:two-partition}
		\Sigma
		= \bigcup_{x \in \X}
		\C_x \times \D_{\ell - x}
		= \bigcup_{y \in \Y}
		\C_{\ell - y} \times \D_y,
	\end{equation}
	where the symbol $\times$
	stands for the Cartesian product.
	In particular, if the state
	of the first queue is equal to $c \in \C$,
	then the set of possible states
	of the second queue is $\Y_{\ell - |c|}$,
	and vice versa.
	
	To prove the theorem,
	it suffices to verify that any measure
	given by~\eqref{eq:two-picd} satisfies
	the following partial balance equations
	in each state $(c;d) \in \Sigma$,
	with $c = (c_1, \ldots, c_n)$,
	$d = (d_1, \ldots, d_m)$,
	$x = |c|$, and $y = |d|$:
	\begin{itemize}
		\item Equalize the flow
		out of state $(c;d)$
		due to a service completion
		at the first queue
		with the flow into that state
		due to an arrival at this queue,
		that is, to a service completion
		at the second queue
		(if $c \neq \emptyset$):
		\begin{equation}
			\label{eq:two-partial-balance-1}
			\pi(c;d) \, \mu(c)
			= \sum_{d' \in \D_{y + e_{c_n}}}
			\sum_{\substack{
					p = 1 \\
					\delta_p(d') = (d, c_n)
			}}^{m+1}
			\pi(c_1, \ldots, c_{n-1}; d')
			\, \Delta\nu(d'_1, \ldots, d'_p).
		\end{equation}
		\item Equalize the flow
		out of state $(c;d)$
		due to a service completion
		at the second queue
		with the flow into that state
		due to an arrival at this queue,
		that is, to a service completion
		at the first queue
		(if $d \neq \emptyset$):
		\begin{equation}
			\label{eq:two-partial-balance-2}
			\pi(c;d) \, \nu(d)
			= \sum_{c' \in \C_{x + e_{d_m}}}
			\sum_{\substack{
					p = 1 \\
					\delta_p(c') = (c,d_m)
			}}^{n+1}
			\pi(c'; d_1, \ldots, d_{m-1})
			\, \Delta\mu(c'_1, \ldots, c'_p),
		\end{equation}
	\end{itemize}
	We focus on~\eqref{eq:two-partial-balance-2}
	because~\eqref{eq:two-partial-balance-1}
	follows by symmetry.
	Assuming that $d \neq \emptyset$,
	the main argument consists
	of observing that,
	since $\Lambda(d) \, \nu(d)
	= \Lambda(d_1, \ldots, d_{m-1})$,
	a stationary measure
	given by~\eqref{eq:two-picd}
	satisfies~\eqref{eq:two-partial-balance-2}
	if and only if the balance function $\Phi$ defined by~\eqref{eq:oi-Phic} satisfies
	\begin{equation*}
		\Phi(c)
		= \sum_{c' \in \C_{x + e_{d_m}}}
		\sum_{\substack{
				p = 1 \\
				\delta_p(c') = (c, d_m)
		}}^{n+1}
		\Phi(c')
		\, \Delta\mu(c'_1, \ldots, c'_p).
	\end{equation*}
	Up to a normalization constant,
	the right-hand side
	of this equation is also that
	of the partial balance
	equation~\eqref{eq:one-balance-1}
	applied to state
	$(c_1, \ldots, c_n, d_m)$,
	since the domains of
	the outer sums
	are the same.
	The proof of Theorem~\ref{theo:one-picd} already showed that $\Phi$ satisfies this equation.
	To conclude, it suffices to observe that
	the left-hand side
	of~\eqref{eq:one-balance-1} is
	$\Phi(c_1, \ldots, c_n, d_m)
	\, \mu(c_1, \ldots, c_n, d_m)
	= \Phi(c)$.
\end{proof}

\begin{remark}
	Mutatis mutandis, Remarks~\ref{remark1}	and \ref{remark2} also apply to a closed tandem network	of two \gls{pands} queues. In particular, an equivalent of Appendix~\ref{app:irreducibility} can be derived in case the initial network state does not adhere to a placement order.
\end{remark}

\section{Application to resource management in machine clusters} \label{sec:app}

In the introduction, we already mentioned that \gls{pands} queues can be used to model several load-distribution and scheduling protocols, such as the FCFS-ALIS and \gls{fcfs} redundancy scheduling protocols, in clusters of machines in which jobs have assignment constraints.
This cluster model
can represent various queueing systems,
like the computer clusters
or manufacturing systems
mentioned in the introduction,
in which not every machine is able
to fulfill the service requirement
of any job.
It has played a central role in
several studies of product-form queueing models
over the past decade; see e.g.\ \cite{AW12,AW14,ABDV19,ABV18,BC17,C19-1,C19-2,GR20,G16}.
We now explain how \gls{pands} queues
can be applied to analyze the performance
of existing and new
load-distribution and scheduling protocols
in such clusters of machines.
As an introductory example,
in Section~\ref{subsec:app-intro},
we consider a load-distribution
protocol whereby the decision of
assigning a job to a machine
is based on
the order in which
(slots in the buffers of)
machines have become idle.
We will see that this load-distribution protocol can be interpreted as a new scheduling protocol, called \emph{cancel-on-commit}, for a redundancy scheduling system.
We then explain how the queueing model
that describes the dynamics of this protocol
can be cast as a closed tandem network
of two \gls{pands} queues
like that of Section~\ref{subsec:two}. 
In Section~\ref{subsec:app-general},
we introduce a more general framework
that encompasses other
load-distribution and scheduling protocols,
and then we give two prototypical
examples of such protocols.
In all cases,
the dynamics can be described
using a closed tandem network of \gls{pands} queues
like that introduced in Section~\ref{subsec:two}, and
deriving the stationary distribution
of the system state is a direct application
of the results of Section~\ref{subsec:two},
provided that the associated Markov process
satisfies the appropriate irreducibility conditions.

\subsection{Assign-to-the-longest-idle-slot
and cancel-on-commit}
\label{subsec:app-intro}

We first consider a cluster
made of a dispatcher and a set of machines.
Each incoming job is
\textit{a priori} compatible with several machines
but is eventually assigned to
and processed by only one of these machines.
Following the same approach
as in the recent work~\cite{ABV18},
we introduce two variants of this cluster,
one in which the dispatcher
has a central buffer to store the jobs
that have not been committed
to a machine yet,
and another
in which these uncommitted jobs
are temporarily replicated
in the buffers of several machines.
In the former case,
we introduce an assignment protocol,
called first-come-first-served
and assign-to-the-longest-idle-\emph{slot}
(FCFS-ALIS), that generalizes the
first-come-first-served
and assign-to-the-longest-idle-\emph{server}
protocol introduced in~\cite{AW14}.
In the latter case,
we introduce a new
redundancy scheduling protocol,
called cancel-on-commit,
that generalizes the
cancel-on-start protocol.

\subsubsection{Assign-to-the-longest-idle-slot}
\label{subsubsec:alis}

We start with the variant where uncommitted jobs are stored in a central buffer.
Consider a two-level buffered cluster
consisting of a dispatcher and
a set $\Se = \{1, \ldots, S\}$ of machines.
For each $s \in \Se$,
machine~$s$ has a buffer
of length~$\ell_s \in \{1, 2, \ldots\}$
that contains all jobs
assigned (we also say committed) to this machine,
either waiting or in service.
Each machine processes the jobs
in its buffer
in \gls{fcfs} order and,
for each $s \in \Se$,
the service time of a job on machine~$s$
is exponentially distributed
with rate~$\mu_s$.
Incoming jobs enter the system via the dispatcher,
and the dispatcher is in charge
of assigning these jobs to
(the buffer of) a machine.
The dispatcher also has its own buffer
where incoming jobs can be stored
in case they cannot be immediately
assigned to a machine due to full buffers.
In this way,
each job present in the system
is either in the buffer of a machine,
in which case we say that it
has been assigned or committed to this machine,
or in the buffer of the dispatcher,
waiting for an assignment.

Each incoming job has a type
that determines the set of machines
to which this job can be assigned.
The set of job types is denoted
by $\K = \{1, \ldots, K\}$ and,
for each $k \in \K$,
type-$k$ jobs arrive according to
a Poisson process with rate $\nu_k$
and can be assigned to any machine
within the set $\Se_k \subseteq \Se$.
Conversely, for each $s \in \Se$,
we let $\K_s \subseteq \K$ denote
the set of job types that can be
assigned to machine~$s$
(that is, such that $s \in \Se_k$).
This defines a bipartite \emph{assignment} graph
between job types and machines,
in which there is an edge between
a type and a machine if
the jobs of this type can be
assigned to this machine;
in this case, we say that these jobs
are \emph{compatible} with the machine.
In the examples of this section,
job types will be identified
by letters rather than numbers
to avoid confusion.
In the assignment graph
of \figurename~\ref{fig:app-intro-assignment}
for instance,
type-$A$ jobs are compatible
with machines~$1$ and $3$
and type-$B$ jobs with machines~$2$ and $3$,
so that $\Se_A = \{1,3\}$,
$\Se_B = \{2,3\}$,
$\K_1 = \{A\}$, $\K_2 = \{B\}$,
and $\K_3 = \{A,B\}$.

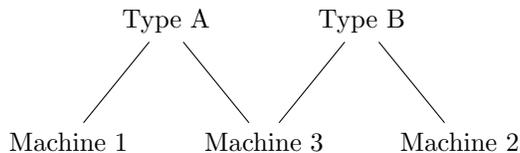
\begin{figure}[ht]
	\centering
	\begin{tikzpicture}
		\def\width{2.6cm}
		\def\height{1.6cm}
		
		\node (t1) {Type A};
		\node (t2)
		at ($(t1)+(\width,0)$) {Type B};
		
		\node (s1)
		at ($(t1)+(-0.5*\width,-\height)$)
		{Machine 1};
		\node (s3)
		at ($(t1)!.5!(t2)-(0,\height)$)
		{Machine 3};
		\node (s2)
		at ($(t2)+(0.5*\width,-\height)$)
		{Machine 2};
		
		\draw (t1) -- (s1);
		\draw (t1) -- (s3);
		\draw (t2) -- (s3);
		\draw (t2) -- (s2);
	\end{tikzpicture}
	\caption{A bipartite assignment graph
		between job types and machines.
		To avoid any confusion
		in the rest of this section,
		job types are identified
		by letters
		rather than numbers
		in the examples.
	}
	\label{fig:app-intro-assignment}
\end{figure}

An incoming type-$k$ job
is immediately assigned to
the buffer of a machine in $\Se_k$
if at least one of these buffers
has idle (that is, empty) slots,
otherwise the job is left unassigned
in the dispatcher's buffer.
We assume that, for each $k \in \K$,
the dispatcher's buffer can contain
at most $\ell_k \in \{0, 1, 2, \ldots\}$
unassigned type-$k$ jobs,
so that an incoming type-$k$ job
is rejected (and considered permanently lost)
if it arrives while there are
already $\ell_k$ unassigned type-$k$ jobs.
In other words, the dispatcher's buffer
consists of $\ell_k$ slots for type-$k$ jobs,
for each $k \in \K$,
and these slots cannot be occupied
by jobs of other types.
These values $\ell_k$ can be used
to differentiate service between job types.

When the job in service
on machine~$s$ completes service,
the oldest unassigned job
of a type in $\K_s$, if any,
is immediately removed
from the dispatcher's buffer
and added to the buffer of machine~$s$,
following which
machine~$s$ immediately starts processing
the oldest job in its buffer.
Conversely, when a type-$k$ job arrives,
the dispatcher applies the following procedure:
\begin{enumerate}[label=(\roman*)]
	\item \label{alis-1}
	if one or more machines
	in the set $\Se_k$
	have space in their buffer,
	the dispatcher selects the buffer slot,
	among those corresponding to
	these machines,
	that has been idle the longest,
	and assigns the job to the corresponding machine;
	\item \label{alis-2}
	otherwise, if there are currently
	fewer than
	$\ell_k$ unassigned type-$k$ jobs,
	the dispatcher puts the incoming job
	in its own buffer
	until this job can be assigned
	to the buffer of a machine in $\Se_k$;
	\item \label{alis-3}
	otherwise, the job is rejected.
\end{enumerate}
The assignment rule described
in step~\ref{alis-1} will be called \gls{alis},
in reference to the
assign-to-the-longest-idle-server (also ALIS)
assignment rule~\cite{AW12,AW14},
which corresponds to the special case
where $\ell_s = 1$ for each $s \in \Se$%
\footnote{%
	In the remainder,
	when we use the acronym ALIS (resp.\ FCFS-ALIS),
	we refer to our generalization of
	the ALIS (resp.\ FCFS-ALIS) protocol,
	in which the ``S'' refers to ``slot''.
	If we need to be more specific,
	we will replace ``ALIS''
	with either ``server-based ALIS''
	or ``slot-based ALIS''.
}.
Note that there cannot be
an unassigned job of a type in $\K_s$
while there is space in the buffer of machine~$s$,
as in this case this job would have been assigned to
the buffer of this machine earlier.
We assume that all assignment operations
occur immediately upon a job arrival
or a service completion.
Overall,
we obtain a protocol called \gls{fcfs-alis},
in which machines process jobs
in their buffer in \gls{fcfs} order,
and incoming jobs are assigned
to machines according to
the \gls{alis} assignment rule.
	
	Our ALIS assignment rule
	can be rephrased as follows in terms of tokens.
	This alternative description will, among others, be useful for the analysis of Section \ref{subsubsec:app-intro-det}.
	Assume that each machine sends
	a token to the dispatcher
	whenever a job completes service,
	and that the buffer keeps
	an ordered list of these tokens,
	with the oldest tokens at the head.
	Then, when a new job arrives,
	the dispatcher assigns this job
	to the compatible machine
	whose token appears earliest in this list,
	and deletes the corresponding token.
	Everything goes as if each job
	seized a token of a machine
	when it starts occupying a slot in its buffer
	and released this token upon service completion.
	Each token in the dispatcher's list
	corresponds to an idle slot
	in the buffer of the corresponding machine.
	The intuition behind
	this assignment rule is that,
	if a slot in a buffer
	has been idle for a long time,
	it is likely that
	the corresponding machine
	is relatively less loaded than others,
	so that we should assign
	the incoming job to this machine if possible.

\begin{remark}\label{rem:fcfs-alis}
	The assign-to-the-longest-idle-server (ALIS) protocol
	introduced in~\cite{AW12}
	corresponds to the degenerate case
	where $\ell_s = 1$ for each $s \in \Se$
	and $\ell_k = 0$ for each $k \in \K$,
	so that an incoming job is rejected if
	all its compatible machines
	are already busy.
	It was shown in \cite{AW12} that,
	in this case, performance is insensitive
	to the job size distribution beyond its mean.
	A protocol
	that generalizes this \gls{alis} protocol
	to a scenario where
	$\ell_s \in \{1, 2, \ldots\}$
	for each $s \in \Se$
	and preserves
	its insensitivity property
	was introduced in \cite{C19-1};
	this protocol is also a special case
	of the protocol that we have just defined.
	The first-come-first-served
	and assign-to-the-longest-idle-server (FCFS-ALIS)
	protocol
	introduced in \cite{AW14}
	corresponds to the degenerate case
	where $\ell_s = 1$ for each $s \in \Se$
	and $\ell_k = +\infty$ for each $k \in \K$.
	In practice, taking larger values
	for $\ell_s$ for $s \in \Se$ can be beneficial
	if there is a communication delay
	between the dispatcher and the machines,
	such that a job cannot enter
	service immediately
	after it is assigned to a machine.
	Furthermore, taking unequal values
	for $\ell_k$ for $k \in \K$
	can be useful
	to differentiate service between job types.
	In the remainder of this section,
	we shall assume that
	$\ell_k \in \{1, 2, \ldots\}$
	for each $k \in \K$
	and $\ell_s \in \{1, 2, \ldots\}$
	for each $s \in \Se$.
	Although the protocol
	defined above can be easily extended
	to the case $\ell_k = +\infty$,
	generalizing the queueing analysis
	of Section~\ref{subsubsec:app-intro-det}
	is less straightforward
	and will be left for future work.
\end{remark}

\subsubsection{Cancel-on-commit} \label{subsubsec:cancel-on-commit}

It was observed in \cite{ABV18}
that the server-based FCFS-ALIS protocol,
which is a special case of
the slot-based FCFS-ALIS protocol
introduced in Section~\ref{subsubsec:alis}
as per Remark~\ref{rem:fcfs-alis},
leads to the same dynamics as a redundancy scheduling protocol, called \emph{cancel-on-start}~\cite{G16, BC17}, whereby replicas of a job may be routed to multiple machines, and redundant replicas are canceled whenever a replica enters service at a machine. 
Using a similar approach,
we now show that the slot-based \gls{fcfs-alis} protocol
can be reinterpreted as a redundancy scheduling protocol,
called \emph{cancel-on-commit},
that generalizes the cancel-on-start protocol.
The basic idea is to let go of the dispatcher's central buffer for uncommitted jobs and, instead, replicate these jobs in the buffers of all their compatible machines.

Consider the following reinterpretation
of the cluster model of Section~\ref{subsubsec:alis}.
The cluster still consists of a dispatcher
and a set $\Se = \{1, \ldots, S\}$ of machines,
and incoming jobs have a type in
the set $\K = \{1, \ldots, K\}$
that determines the set of machines
to which they can be assigned.
We again let $\Se_k$ denote the set of machines
that are compatible with type-$k$ jobs,
for each $k \in \K$,
and $\K_s$ the set of job types
that are compatible with machine~$s$,
for each $s \in \Se$.
The main difference with Section~\ref{subsec:app-intro}
is that the dispatcher no longer has
a central buffer to store unassigned jobs.
Instead, for each $s \in \Se$,
machine~$s$ has a two-level buffer.
The first level of this buffer
consists of $\ell_s$ slots
occupied by jobs that have been
assigned or \emph{committed} to this machine.
The second level of this buffer
will contain unassigned job replicas,
as explained in the next paragraph.
From now on, the first (resp.\ second) level
of the buffer of machine~$s$ will be called its first-level (resp.\ second-level) buffer for brevity.
Each machine processes the jobs
in its first-level buffer
in \gls{fcfs} order.

When a type-$k$ job arrives,
the dispatcher immediately
sends a replica of this job
to each machine within the set~$\Se_k$.
What happens next depends on
the state of these machines,
and we distinguish three cases:
\begin{enumerate}[label=(\roman*)]
	\item \label{cac-1}
	at least one machine in $\mathcal{S}_k$
	has available space in its first-level buffer,
	\item \label{cac-2}
	none of the machines in $\Se_k$ has available space in its first-level buffer,
	but the second-level buffers
	of these machines each contain
	less than $\ell_k$ uncommitted
	type-$k$ jobs, or
	\item \label{cac-3}
	the first-level buffers of the machines in $\Se_k$ are all full,
	and their second-level buffers
	already contain $\ell_k$
	uncommitted type-$k$ jobs.
\end{enumerate}
In the first case,
we let $s$ denote the machine in $\Se_k$
with the slot in its first-level buffer
that has been idle the longest
out of all idle slots
in the first-level buffers
of the machines in $\Se_k$.
The replica sent to machine~$s$
takes this first-level slot,
and we say that the job
commits to machine~$s$.
All other replicas of the job
are immediately canceled.
The committed replica then awaits
its service by machine~$s$,
rendered in FCFS fashion,
and leaves upon service completion.
In case~\ref{cac-2}, each replica of the job
takes a slot in the second-level buffer
of a machine in $\Se_k$
and waits until it is either canceled
or committed to this machine.
Whenever a job completes service at a machine,
a first-level buffer slot becomes available.
This buffer slot is immediately taken
by the longest waiting replica
in the machine's second-level buffer, if any.
The corresponding job is committed
to the machine,
and all other replicas of this job are canceled.
Finally, in case~\ref{cac-3},
the incoming job is rejected
and considered permanently lost.

One can verify that this \emph{cancel-on-commit} protocol leads to the same dynamics as the \gls{fcfs-alis} protocol described in Section~\ref{subsubsec:alis}. The key difference, which does not impact the dynamics, is that uncommitted jobs wait in the buffers of their compatible machines instead of waiting in the dispatcher's centralized buffer.
We will see in Section~\ref{subsubsec:app-intro-det} that these two systems can be cast as a closed tandem of two \gls{pands} queues, so that Theorem~\ref{theo:two-steady-cd} also provides the stationary distribution of the job population in a redundancy scheduling system with the cancel-on-commit protocol.

\subsubsection{Interpretation as a closed tandem network of pass-and-swap queues}
\label{subsubsec:app-intro-det}

The objective of this section is to
cast the above-mentioned cluster model
as a closed tandem network of two \gls{pands} queues.
To this end, we first need to give
(yet) another perspective on
the dynamics of the system.

\paragraph{Token-based central-queue perspective}

The dynamics of the machine cluster
can also be described
by considering tokens,
as if each job present in the system
held a token that identifies
the slot occupied by this job.
More specifically,
the set of token classes
is $\K \sqcup \Se$,
where $\sqcup$ denotes
the disjoint union operator\footnote{%
	In practice, to make sure that the sets $\K$ and $\Se$ remain disjoint, job types and machines can be renumbered if necessary.
	No confusion will arise
	from this slight abuse of notation,
	as we will always be using
	the letter~$k$
	for a token class in $\K$,
	associated with a job type,
	and the letters~$s$ and~$t$
	for a token class in $\Se$,
	associated with a machine.
	In the examples,
	the index of a token class will be
	a letter if this token class
	is associated with a job type
	and a number if this token class
is associated with a machine.}.
There are $\ell_s$ class-$s$ tokens,
for each $s \in \Se$, and
$\ell_k$ class-$k$ tokens,
for each $k \in \K$.
The former tokens are those we already referred to in Section~\ref{subsubsec:alis}.
More particularly, focusing
on the cluster model
of Section~\ref{subsubsec:alis},
each class-$s$ token
corresponds to a specific slot
in the buffer of machine~$s$,
in the sense that a job holds this token
when it occupies the corresponding slot.
Similarly, each class-$k$ token
corresponds to a specific slot
that can be occupied by type-$k$ jobs
in the dispatcher's buffer.
Focusing on the cluster model of Section~\ref{subsubsec:cancel-on-commit},
each class-$s$ token
corresponds to a specific slot
in the first-level buffer of machine~$s$,
while each class-$k$ token
corresponds to a specific slot
that can be occupied by type-$k$ jobs
in the second-level buffers
of the machines in $\Se_k$.
In general, we can think of
a class-$s$ token
as a token that ``belongs'' to machine~$s$,
and of a class-$k$ token
as a token that ``belongs'' to type-$k$ jobs.
Each token is held by a job
if and only if this job
is occupying the corresponding slot,
otherwise the token is available.
When a type-$k$ job enters the system,
this job seizes an available token
of a class $s \in \Se_k$, if any,
otherwise it seizes
an available class-$k$ token, if any.
When a job completes service
at machine~$s \in \Se$,
its class-$s$ token is passed on to
an unassigned job
of a type $k \in \K_s$, if any,
in which case this unassigned job
releases the class-$k$ token that it was holding;
otherwise, the class-$s$ token is released.

The dynamics of the cluster
can be entirely described
by the movements of these tokens,
and long-term performance metrics,
like the mean sojourn time of jobs,
can be derived from
the long-term expected number of tokens
held by jobs present in the system.
These tokens will play the part
of customers in the closed tandem network
of two \gls{pands} queues
that we will now introduce.
The first queue
will contain tokens held by jobs
present in the system,
either assigned to a machine
or waiting for an assignment,
and these tokens will be ordered
according to the arrival order of jobs in the system.
In particular, contrary to
Sections~\ref{subsubsec:alis}
and~\ref{subsubsec:cancel-on-commit},
we will adopt a central-queue perspective
in which all (tokens held by)
jobs present in the system
are gathered in a single queue.
This representation has become
standard in product-form queueing
models representing
machine clusters~\cite{ABDV19,ABV18,BC17,GR20,G16}.
The second queue
of the closed tandem network
will contain available tokens,
and their order will partly
reflect their release order.
We will see that the dynamics induced by the
\gls{fcfs-alis} and cancel-on-commit protocols
are captured by this model,
provided that the \gls{pands} mechanism is
applied with a suitable swapping graph.
We first give an overview of
the closed tandem network
and then we will detail the dynamics
of each queue separately.

\paragraph{Closed tandem network of two \gls{pands} queues}

We consider a closed tandem network
of two \gls{pands} queues
like the one described
in Section~\ref{subsec:two}.
Customers represent tokens
and, in the remainder,
we will always refer to them as tokens.
The set of token classes
is $\I = \K \sqcup \Se$, and this
set has cardinality $I = K + S$.
The first queue contains
the tokens held by jobs
present in the system
and the second queue
contains the available tokens.
A token that leaves the first queue
immediately enters the second queue
as a token of the same class,
and conversely.
The overall number of class-$i$ tokens
in the network is $\ell_i$,
for each $i \in \I$.

Both \gls{pands} queues
have the same swapping graph,
so that we obtain a closed tandem network
of two \gls{pands} queues
like that described
in Section~\ref{subsec:two}.
The placement order is as follows: $s \prec k$
for each $s \in \Se$ and $k \in \K_s$
(or equivalently,
for each $k \in \K$ and $s \in \Se_k$).
The corresponding placement graph
is obtained from the assignment graph
introduced in Section~\ref{subsubsec:alis}
by orienting edges from
the (token classes that correspond to) machines
towards the
(token classes that correspond to) job types.
For example, the placement graph
associated with the assignment graph
of \figurename~\ref{fig:app-intro-assignment}
is shown in
\figurename~\ref{fig:app-intro-placement}.
The swapping graph of the network
is simply the underlying
undirected graph
of the placement graph.
This placement order guarantees that,
if a class-$k$ token is in the first queue,
then, for each $s \in \Se_k$,
all class-$s$ tokens
are also in the first queue
and precede this class-$k$ token.
This corresponds to the fact that,
in the cluster,
a type-$k$ job can only be unassigned
if the buffers of all machines
in $\Se_k$ are full.
We will come back to this interpretation
later when we specify the dynamics
of each queue in detail.
An example of a network state that adheres to this placement order is given by a state where all tokens are lined up in the second queue,
with first those of the classes in~$\K$ in an arbitrary order, and then those of the classes in~$\Se$, also in an arbitrary order. We assume that the network starts in such a state, which corresponds to an empty system in which all tokens are available.

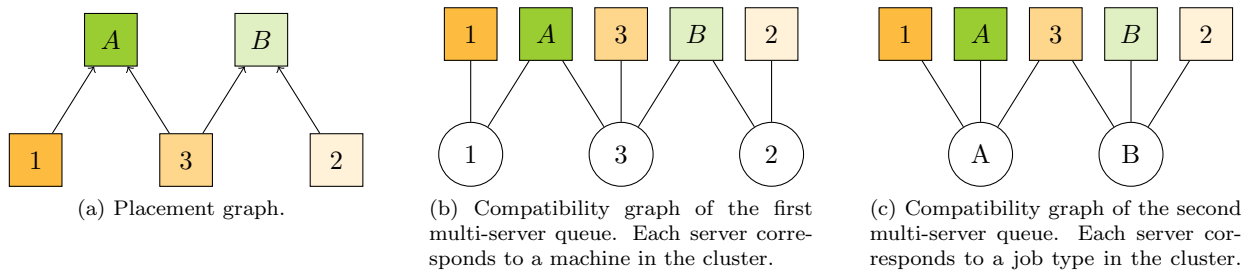
\begin{figure}[ht]
	\centering
	\subfloat[Placement graph.
	\label{fig:app-intro-placement}]{
		\begin{tikzpicture}
			\def\width{2cm}
			\def\height{1.6cm}
			
			\node[class, darkgreen] (c1) {$A$};
			\node[class, lightgreen] (c2)
			at ($(c1)+(\width,0)$)
			{$B$};
			
			\node[class, darkorange] (s1)
			at ($(c1)-(.5*\width,0)-(0,\height)$)
			{$1$};
			\node[class, neutralorange] (s3)
			at ($(c1)!.5!(c2)-(0,\height)$)
			{$3$};
			\node[class, lightorange] (s2)
			at ($(c2)+(.5*\width,0)-(0,\height)$) {$2$};
			
			\draw[<-] (c1) -- (s1);
			\draw[<-] (c1) -- (s3);
			\draw[<-] (c2) -- (s3);
			\draw[<-] (c2) -- (s2);
		\end{tikzpicture}
	}
	\hfill
	\subfloat[Compatibility graph
	of the first multi-server queue.
	Each server corresponds to
	a machine in the cluster.
	\label{fig:app-intro-det-msq1}]{
		\begin{tikzpicture}
			\def\width{2cm}
			\def\height{1.6cm}
			
			\node[class, darkorange] (cs1)
			{$1$};
			\node[class, neutralorange] (cs3)
			at ($(cs1)+(\width,0)$) {$3$};
			\node[class, lightorange] (cs2)
			at ($(cs3)+(\width,0)$) {$2$};
			
			\node[class, darkgreen] (c1)
			at ($(cs1)!.5!(cs3)$)
			{$A$};
			\node[class, lightgreen] (c2)
			at ($(cs2)!.5!(cs3)$)
			{$B$};
			
			\node[server] (s1)
			at ($(cs1)-(0,\height)$) {1};
			\node[server] (s3)
			at ($(cs3)-(0,\height)$) {3};
			\node[server] (s2)
			at ($(cs2)-(0,\height)$) {2};
			
			\draw (cs1) -- (s1);
			\draw (cs3) -- (s3);
			\draw (cs2) -- (s2);
			\draw (s1) -- (c1) -- (s3) -- (c2) -- (s2);
		\end{tikzpicture}
	}
	\hfill
	\subfloat[Compatibility graph
	of the second multi-server queue.
	Each server corresponds to
	a job type in the cluster.
	\label{fig:app-intro-det-msq2}]{
		\begin{tikzpicture}
			\def\width{2cm}
			\def\height{1.6cm}
			
			\node[class, darkorange] (cs1)
			{$1$};
			\node[class, neutralorange] (cs3)
			at ($(cs1)+(\width,0)$) {$3$};
			\node[class, lightorange] (cs2)
			at ($(cs3)+(\width,0)$) {$2$};
			
			\node[class, darkgreen] (c1)
			at ($(cs1)!.5!(cs3)$)
			{$A$};
			\node[class, lightgreen] (c2)
			at ($(cs2)!.5!(cs3)$)
			{$B$};
			
			\node[server] (s1)
			at ($(c1)-(0,\height)$) {A};
			\node[server] (s2)
			at ($(c2)-(0,\height)$) {B};
			
			\draw (c1) -- (s1);
			\draw (c2) -- (s2);
			\draw (cs1) -- (s1) -- (cs3)
			-- (s2) -- (cs2);
		\end{tikzpicture}
	}
	\caption{%
		Closed tandem network
		of two \gls{pands} queues
		associated with the assignment graph
		of \figurename~\ref{fig:app-intro-assignment}.
		The classes and servers associated
		with job types are identified by letters,
		and those associated with machines
		are identified by numbers
		(in accordance with
		\figurename~\ref{fig:app-intro-assignment}).
		The class colors are visual aids
		that help distinguish between
		the classes associated
		with job types (in green)
		and those associated
		with machines (in orange).
		The same class and server indexing
		and color code
		will be adopted in \figurename s~\ref{fig:app-network}
		and \ref{fig:app3-app1}.
	}
	\label{fig:app-intro-det}
\end{figure}

\paragraph{First queue}

A token is in the first queue
when it is held by a job present in the system,
whether this job is assigned to
the buffer of a machine or unassigned.
This queue is a multi-server queue,
like that of Example~\ref{ex:pands},
and each server represents
a machine in the cluster.
More specifically,
the set of servers is $\Se$ and,
for each $s \in \Se$,
the service rate of server~$s$
is equal to the service rate
$\mu_s$ of machine~$s$.
Using the same index set~$\Se$
for the set of servers in the first queue
and for the token classes associated with machines
may seem to be ambiguous at first,
but we will always specify whether
we are referring to a server $s \in \Se$
or to a token class $s \in \Se$.
For each $s \in \Se$,
a class-$s$ token
is compatible with
server~$s$ and with this server only.
Following the same notation
as in Example~\ref{ex:oi},
we let $\Se_s = \{s\}$.
Additionally, for each $k \in \K$,
the set of servers that can process
class-$k$ tokens is $\Se_k$,
which corresponds, in the cluster, to
the set of machines
to which type-$k$ jobs can be assigned.
In the end, for each $s \in \Se$,
server~$s$ can process
the tokens that belong to machine~$s$
plus the tokens that belong
to the job types in $\K_s$.
For example, the compatibility graph
of the first queue
in the tandem network associated
with the assignment graph of
\figurename~\ref{fig:app-intro-assignment}
is shown in \figurename~\ref{fig:app-intro-det-msq1}.
Each server processes
its compatible tokens
in \gls{fcfs} order.
Note that, for each $k \in \K$,
the set $\Se_k$ now plays
two roles in the tandem network:
it represents the set of servers
that can process class-$k$ tokens
in the first queue
as well as the set of token classes
that can be swapped with class-$k$ tokens.

Recall that, according to the placement order,
if a class-$k$ token is in the first queue,
then, for each $s \in \Se_k$,
all class-$s$ tokens are also in the first queue,
at positions that precede the position
of the class-$k$ token.
Given the compatibility graph,
this implies that a class-$k$ token
will actually never be \textit{in service}
in this queue.
Therefore, the only way that a class-$k$ token
leaves the first queue
is if a token of
a class $s \in \Se_k$ completes service
and ejects this class-$k$ token.
In the cluster, this means that
a job completes service on machine~$s$
and that the token
released by this job
is seized by a type-$k$ job that
was unassigned so far
(so that this type-$k$ job
releases its own class-$k$ token).

We let $c = (c_1, \ldots, c_n) \in \I^*$
denote the state of the first queue.
As observed before,
the placement order guarantees that,
if there is a token of a class $k \in \K$
at a position $p \in \{1, \ldots, n\}$
in the first queue,
then, for each $s \in \Se_k$,
each class-$s$ token is also
in the first queue,
at a position $q \in \{1, \ldots, p-1\}$
that precedes that
of the class-$k$ token.
Therefore,
the state space of the state
of the first queue is a strict subset $\C$
of the set of sequences $c \in \I^*$
such that $|c| \le \ell$,
where $\ell = (\ell_1, \ldots, \ell_I)$
is the vector that gives
the maximum number of tokens of each class.
The overall and per-token service rates
are still given by \eqref{eq:mu}
and \eqref{eq:delta-mu},
with the sets $\Se_k$ for $k \in \K$
and $\Se_s$ for $s \in \Se$
as defined above.
Because of the placement order,
\eqref{eq:delta-mu} simplifies to
$\Delta\mu(c_1, \ldots, c_p) = 0$
for each $p \in \{1, \ldots, n\}$
such that $c_p \in \K$.
Furthermore,
the compatibility graph
guarantees that,
for each $p \in \{1, \ldots, n\}$
such that $s = c_p \in \Se$, we have
$\Delta\mu(c_1, \ldots, c_p) = \mu_s$
if $|(c_1, \ldots, c_{p-1})|_s = 0$
and $\Delta\mu(c_1, \ldots, c_p) = 0$
otherwise.

Now assume that a token
in some position $p \in \{1, \ldots, n\}$
completes service
and let $s = c_p \in \Se$
denote this token's class.
Because of the \gls{pands} mechanism,
only one of these two types of transitions can occur:
\begin{enumerate}[label=(\roman*)]
	\item \label {item:trans-1} If the first queue
	contains a token of a class in $\K_s$
	(necessarily in position at least $p+1$
	because of the placement order),
	the class-$s$ token replaces
	the first of these tokens,
	say of class $k \in \K_s$,
	and the ejected class-$k$ token joins
	the second queue.
	In the cluster, this means that
	a token from machine~$s$
	is released by a departing job
	and is immediately seized by
	an unassigned type-$k$ job;
	this type-$k$ job releases
	its class-$k$ token,
	which is appended
	to the queue of available tokens.
	\item \label{item:trans-2}
	If there is no token
	with a class in $\K_s$
	in the first queue,
	the class-$s$ token
	leaves this queue
	and joins the second queue.
	In the cluster, this means that
	a token from machine~$s$ is released
	by a departing job
	and is immediately appended to
	the queue of available tokens
	because there is
	no unassigned job
	of a type in $\K_s$.
\end{enumerate}
In both cases, a token leaves
the first queue
and is added to the second,
meaning that a token is released
in the cluster.
Examples of transitions
are shown in \figurename~\ref{fig:app-network}
for the cluster
of \figurename~\ref{fig:app-intro-assignment}.
In the state of \figurename~\ref{fig:app-network-1},
all tokens of classes in $\Se$ are held
by jobs present in the system,
and there is also an unassigned type-$A$ job.
From \figurename~\ref{fig:app-network-1}
to \figurename~\ref{fig:app-network-2},
the oldest class-$2$ token completes service
in the first queue
and joins the second queue.
In the cluster,
this means that a job completes service
on machine~2 and
its token is added to the queue
of available tokens because
there is no unassigned job of
a compatible type.
This is a transition of type~\ref{item:trans-2}.
From \figurename~\ref{fig:app-network-2}
to \figurename~\ref{fig:app-network-3},
the oldest class-$3$ token
completes service in the first queue;
this token replaces the class-$A$ token,
which joins the second queue.
In the cluster,
this means that a job completes service
on machine~3 and
that its token is seized by
an unassigned type-$A$ job.
This is a transition of type~\ref{item:trans-1}.
The transition
from \figurename~\ref{fig:app-network-3}
to \figurename~\ref{fig:app-network-4},
triggered by a service completion
in the second queue,
will be commented on later.

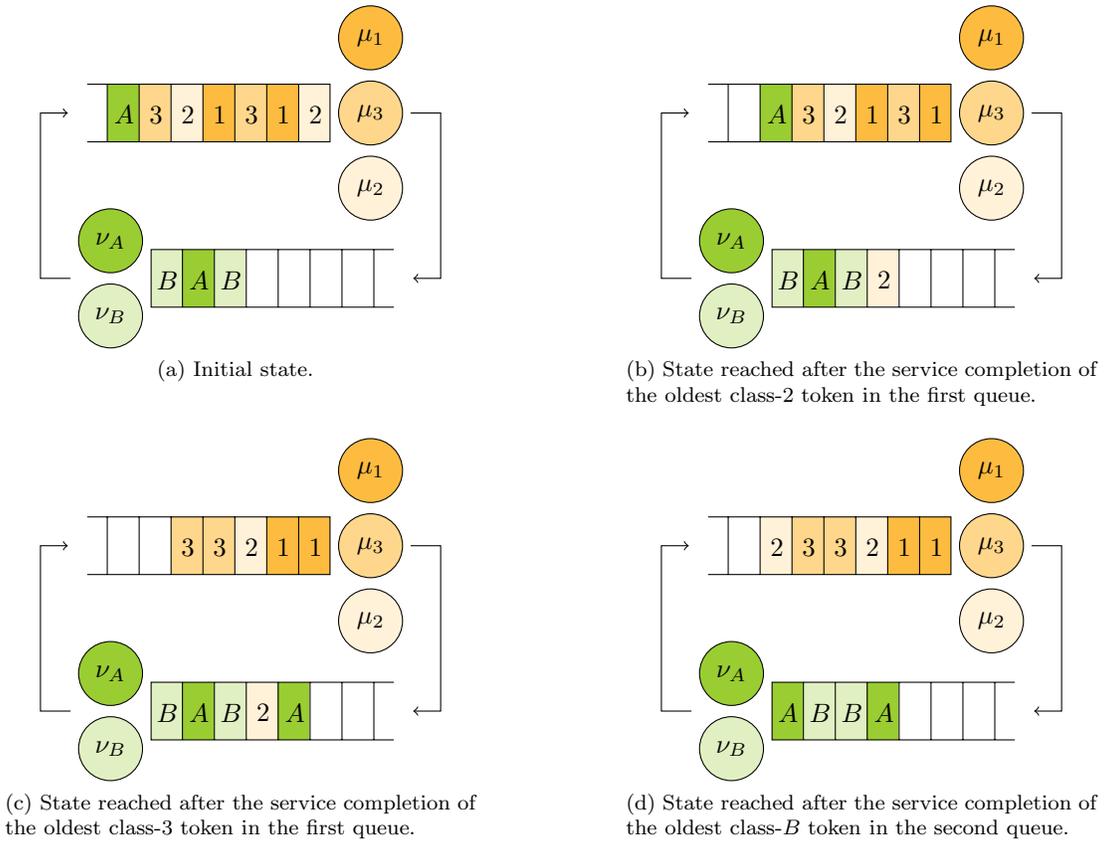
\begin{figure}[ht]
	\centering
	\subfloat[Initial state. \label{fig:app-network-1}]{%
		\begin{tikzpicture}
			\def\width{1cm}
			\def\height{2.2cm}
			
			\node[fcfs=8,
			rectangle split part fill
			={white, \darkgreen,
				\neutralorange, \lightorange, \darkorange,
				\neutralorange, \darkorange, \lightorange},
			] (queue1) {
				\nodepart{one}{\phantom{1}}
				\nodepart{two}{\phantom{1}}
				\nodepart{three}{\phantom{1}}
				\nodepart{four}{\phantom{1}}
				\nodepart{five}{\phantom{1}}
				\nodepart{six}{\phantom{1}}
				\nodepart{seven}{\phantom{1}}
				\nodepart{eight}{\phantom{1}}
			};
			\fill[white] ([xshift=-\pgflinewidth-1pt,
			yshift=-\pgflinewidth+1pt]queue1.north west)
			rectangle ([xshift=\pgflinewidth+4pt,
			yshift=\pgflinewidth-1pt]queue1.south west);
			
			\node at ($(queue1.two)+(.1cm,.09cm)$)
			{$A$};
			\node at ($(queue1.three)+(.1cm,.09cm)$)
			{$3$};
			\node at ($(queue1.four)+(.1cm,.09cm)$)
			{$2$};
			\node at ($(queue1.five)+(.1cm,.09cm)$)
			{$1$};
			\node at ($(queue1.six)+(.1cm,.09cm)$)
			{$3$};
			\node at ($(queue1.seven)+(.1cm,.09cm)$)
			{$1$};
			\node at ($(queue1.eight)+(.1cm,.09cm)$)
			{$2$};
			
			\node[server, anchor=west, fill=\darkorange]
			(mu1) at ($(queue1.east)+(.1cm,0)+(0,1cm)$)
			{$\mu_1$};
			\node[server, anchor=west, fill=\neutralorange]
			(mu3) at ($(queue1.east)+(.1cm,0)$)
			{$\mu_3$};
			\node[server, anchor=west, fill=\lightorange]
			(mu2) at ($(queue1.east)+(.1cm,0)-(0,1cm)$)
			{$\mu_2$};
			
			\node[fcfs=8,
			rectangle split part fill
			={\lightgreen, \darkgreen, \lightgreen, white},
			] (queue2) at ($(queue1)+(\width,0)-(0,\height)$) {
				\nodepart{one}{\phantom{1}}
				\nodepart{two}{\phantom{1}}
				\nodepart{three}{\phantom{1}}
				\nodepart{four}{\phantom{1}}
				\nodepart{five}{\phantom{1}}
				\nodepart{six}{\phantom{1}}
				\nodepart{seven}{\phantom{1}}
				\nodepart{eight}{\phantom{1}}
			};
			\fill[white] ([xshift=\pgflinewidth+1pt,
			yshift=-\pgflinewidth+1pt]queue2.north east)
			rectangle ([xshift=-\pgflinewidth-4pt,
			yshift=\pgflinewidth-1pt]queue2.south east);
			
			\node[server, anchor=east, fill=\darkgreen]
			(nu1) at ($(queue2.west)-(.1cm,0)+(0,.5cm)$)
			{$\nu_A$};
			\node[server, anchor=east, fill=\lightgreen]
			(nu2) at ($(queue2.west)-(.1cm,0)-(0,.5cm)$)
			{$\nu_B$};
			
			\node at ($(queue2.one)+(.1cm,.09cm)$)
			{$B$};
			\node at ($(queue2.two)+(.1cm,.09cm)$)
			{$A$};
			\node at ($(queue2.three)+(.1cm,.09cm)$)
			{$B$};
			
			\draw[->] ($(nu1.west)!.5!(nu2.west)-(.1cm,0)$)
			-- ($(nu1.west)!.5!(nu2.west)-(.5cm,0)$)
			|- ($(queue1.west)-(.1cm,0)$);
			\draw[->] ($(mu3.east)+(.1cm,0)$)
			-- ($(mu3.east)+(.5cm,0)$)
			|- ($(queue2.east)+(.1cm,0)$);
			
			\node at ($(mu1.east)+(.85cm,0)$) {};
			\node at ($(nu1.west)-(.85cm,0)$) {};
		\end{tikzpicture}%
	}
	\hskip 2cm
	\subfloat[State reached
	after the service completion
	of the oldest class-$2$ token
	in the first queue. \label{fig:app-network-2}]{%
		\begin{tikzpicture}
			\def\width{1cm}
			\def\height{2.2cm}
			
			\node[fcfs=8,
			rectangle split part fill
			={white, white, \darkgreen,
				\neutralorange, \lightorange, \darkorange,
				\neutralorange, \darkorange},
			] (queue1) {
				\nodepart{one}{\phantom{1}}
				\nodepart{two}{\phantom{1}}
				\nodepart{three}{\phantom{1}}
				\nodepart{four}{\phantom{1}}
				\nodepart{five}{\phantom{1}}
				\nodepart{six}{\phantom{1}}
				\nodepart{seven}{\phantom{1}}
				\nodepart{eight}{\phantom{1}}
			};
			\fill[white] ([xshift=-\pgflinewidth-1pt,
			yshift=-\pgflinewidth+1pt]queue1.north west)
			rectangle ([xshift=\pgflinewidth+4pt,
			yshift=\pgflinewidth-1pt]queue1.south west);
			
			\node at ($(queue1.three)+(.1cm,.09cm)$)
			{$A$};
			\node at ($(queue1.four)+(.1cm,.09cm)$)
			{$3$};
			\node at ($(queue1.five)+(.1cm,.09cm)$)
			{$2$};
			\node at ($(queue1.six)+(.1cm,.09cm)$)
			{$1$};
			\node at ($(queue1.seven)+(.1cm,.09cm)$)
			{$3$};
			\node at ($(queue1.eight)+(.1cm,.09cm)$)
			{$1$};
			
			\node[server, anchor=west, fill=\darkorange]
			(mu1) at ($(queue1.east)+(.1cm,0)+(0,1cm)$)
			{$\mu_1$};
			\node[server, anchor=west, fill=\neutralorange]
			(mu3) at ($(queue1.east)+(.1cm,0)$)
			{$\mu_3$};
			\node[server, anchor=west, fill=\lightorange]
			(mu2) at ($(queue1.east)+(.1cm,0)-(0,1cm)$)
			{$\mu_2$};
			
			\node[fcfs=8,
			rectangle split part fill
			={\lightgreen, \darkgreen, \lightgreen, \lightorange, white},
			] (queue2) at ($(queue1)+(\width,0)-(0,\height)$) {
				\nodepart{one}{\phantom{1}}
				\nodepart{two}{\phantom{1}}
				\nodepart{three}{\phantom{1}}
				\nodepart{four}{\phantom{1}}
				\nodepart{five}{\phantom{1}}
				\nodepart{six}{\phantom{1}}
				\nodepart{seven}{\phantom{1}}
				\nodepart{eight}{\phantom{1}}
			};
			\fill[white] ([xshift=\pgflinewidth+1pt,
			yshift=-\pgflinewidth+1pt]queue2.north east)
			rectangle ([xshift=-\pgflinewidth-4pt,
			yshift=\pgflinewidth-1pt]queue2.south east);
			
			\node[server, anchor=east, fill=\darkgreen]
			(nu1) at ($(queue2.west)-(.1cm,0)+(0,.5cm)$)
			{$\nu_A$};
			\node[server, anchor=east, fill=\lightgreen]
			(nu2) at ($(queue2.west)-(.1cm,0)-(0,.5cm)$)
			{$\nu_B$};
			
			\node at ($(queue2.one)+(.1cm,.09cm)$)
			{$B$};
			\node at ($(queue2.two)+(.1cm,.09cm)$)
			{$A$};
			\node at ($(queue2.three)+(.1cm,.09cm)$)
			{$B$};
			\node at ($(queue2.four)+(.1cm,.09cm)$)
			{$2$};
			
			\draw[->] ($(nu1.west)!.5!(nu2.west)-(.1cm,0)$)
			-- ($(nu1.west)!.5!(nu2.west)-(.5cm,0)$)
			|- ($(queue1.west)-(.1cm,0)$);
			\draw[->] ($(mu3.east)+(.1cm,0)$)
			-- ($(mu3.east)+(.5cm,0)$)
			|- ($(queue2.east)+(.1cm,0)$);
			
			\node at ($(mu1.east)+(.85cm,0)$) {};
			\node at ($(nu1.west)-(.85cm,0)$) {};
		\end{tikzpicture}%
	}
	\\
	\subfloat[State reached after
	the service completion
	of the oldest class-$3$ token
	in the first queue. \label{fig:app-network-3}]{%
		\begin{tikzpicture}
			\def\width{1cm}
			\def\height{2.2cm}
			
			\node[fcfs=8,
			rectangle split part fill
			={white, white, white,
				\neutralorange, \neutralorange, \lightorange,
				\darkorange, \darkorange},
			] (queue1) {
				\nodepart{one}{\phantom{1}}
				\nodepart{two}{\phantom{1}}
				\nodepart{three}{\phantom{1}}
				\nodepart{four}{\phantom{1}}
				\nodepart{five}{\phantom{1}}
				\nodepart{six}{\phantom{1}}
				\nodepart{seven}{\phantom{1}}
				\nodepart{eight}{\phantom{1}}
			};
			\fill[white] ([xshift=-\pgflinewidth-1pt,
			yshift=-\pgflinewidth+1pt]queue1.north west)
			rectangle ([xshift=\pgflinewidth+4pt,
			yshift=\pgflinewidth-1pt]queue1.south west);
			
			\node at ($(queue1.four)+(.1cm,.09cm)$)
			{$3$};
			\node at ($(queue1.five)+(.1cm,.09cm)$)
			{$3$};
			\node at ($(queue1.six)+(.1cm,.09cm)$)
			{$2$};
			\node at ($(queue1.seven)+(.1cm,.09cm)$)
			{$1$};
			\node at ($(queue1.eight)+(.1cm,.09cm)$)
			{$1$};
			
			\node[server, anchor=west, fill=\darkorange]
			(mu1) at ($(queue1.east)+(.1cm,0)+(0,1cm)$)
			{$\mu_1$};
			\node[server, anchor=west, fill=\neutralorange]
			(mu3) at ($(queue1.east)+(.1cm,0)$)
			{$\mu_3$};
			\node[server, anchor=west, fill=\lightorange]
			(mu2) at ($(queue1.east)+(.1cm,0)-(0,1cm)$)
			{$\mu_2$};
			
			\node[fcfs=8,
			rectangle split part fill
			={\lightgreen, \darkgreen, \lightgreen, \lightorange, \darkgreen, white},
			] (queue2) at ($(queue1)+(\width,0)-(0,\height)$) {
				\nodepart{one}{\phantom{1}}
				\nodepart{two}{\phantom{1}}
				\nodepart{three}{\phantom{1}}
				\nodepart{four}{\phantom{1}}
				\nodepart{five}{\phantom{1}}
				\nodepart{six}{\phantom{1}}
				\nodepart{seven}{\phantom{1}}
				\nodepart{eight}{\phantom{1}}
			};
			\fill[white] ([xshift=\pgflinewidth+1pt,
			yshift=-\pgflinewidth+1pt]queue2.north east)
			rectangle ([xshift=-\pgflinewidth-4pt,
			yshift=\pgflinewidth-1pt]queue2.south east);
			
			\node[server, anchor=east, fill=\darkgreen]
			(nu1) at ($(queue2.west)-(.1cm,0)+(0,.5cm)$)
			{$\nu_A$};
			\node[server, anchor=east, fill=\lightgreen]
			(nu2) at ($(queue2.west)-(.1cm,0)-(0,.5cm)$)
			{$\nu_B$};
			
			\node at ($(queue2.one)+(.1cm,.09cm)$)
			{$B$};
			\node at ($(queue2.two)+(.1cm,.09cm)$)
			{$A$};
			\node at ($(queue2.three)+(.1cm,.09cm)$)
			{$B$};
			\node at ($(queue2.four)+(.1cm,.09cm)$)
			{$2$};
			\node at ($(queue2.five)+(.1cm,.09cm)$)
			{$A$};
			
			\draw[->] ($(nu1.west)!.5!(nu2.west)-(.1cm,0)$)
			-- ($(nu1.west)!.5!(nu2.west)-(.5cm,0)$)
			|- ($(queue1.west)-(.1cm,0)$);
			\draw[->] ($(mu3.east)+(.1cm,0)$)
			-- ($(mu3.east)+(.5cm,0)$)
			|- ($(queue2.east)+(.1cm,0)$);
			
			\node at ($(mu1.east)+(.85cm,0)$) {};
			\node at ($(nu1.west)-(.85cm,0)$) {};
		\end{tikzpicture}%
	}
	\hskip 2cm
	\subfloat[State reached after
	the service completion
	of the oldest class-$B$ token
	in the second queue. \label{fig:app-network-4}]{%
		\begin{tikzpicture}
			\def\width{1cm}
			\def\height{2.2cm}
			
			\node[fcfs=8,
			rectangle split part fill
			={white, white, \lightorange,
				\neutralorange, \neutralorange, \lightorange,
				\darkorange, \darkorange},
			] (queue1) {
				\nodepart{one}{\phantom{1}}
				\nodepart{two}{\phantom{1}}
				\nodepart{three}{\phantom{1}}
				\nodepart{four}{\phantom{1}}
				\nodepart{five}{\phantom{1}}
				\nodepart{six}{\phantom{1}}
				\nodepart{seven}{\phantom{1}}
				\nodepart{eight}{\phantom{1}}
			};
			\fill[white] ([xshift=-\pgflinewidth-1pt,
			yshift=-\pgflinewidth+1pt]queue1.north west)
			rectangle ([xshift=\pgflinewidth+4pt,
			yshift=\pgflinewidth-1pt]queue1.south west);
			
			\node at ($(queue1.three)+(.1cm,.09cm)$)
			{$2$};
			\node at ($(queue1.four)+(.1cm,.09cm)$)
			{$3$};
			\node at ($(queue1.five)+(.1cm,.09cm)$)
			{$3$};
			\node at ($(queue1.six)+(.1cm,.09cm)$)
			{$2$};
			\node at ($(queue1.seven)+(.1cm,.09cm)$)
			{$1$};
			\node at ($(queue1.eight)+(.1cm,.09cm)$)
			{$1$};
			
			\node[server, anchor=west, fill=\darkorange]
			(mu1) at ($(queue1.east)+(.1cm,0)+(0,1cm)$)
			{$\mu_1$};
			\node[server, anchor=west, fill=\neutralorange]
			(mu3) at ($(queue1.east)+(.1cm,0)$)
			{$\mu_3$};
			\node[server, anchor=west, fill=\lightorange]
			(mu2) at ($(queue1.east)+(.1cm,0)-(0,1cm)$)
			{$\mu_2$};
			
			\node[fcfs=8,
			rectangle split part fill
			={\darkgreen, \lightgreen, \lightgreen, \darkgreen, white},
			] (queue2) at ($(queue1)+(\width,0)-(0,\height)$) {
				\nodepart{one}{\phantom{1}}
				\nodepart{two}{\phantom{1}}
				\nodepart{three}{\phantom{1}}
				\nodepart{four}{\phantom{1}}
				\nodepart{five}{\phantom{1}}
				\nodepart{six}{\phantom{1}}
				\nodepart{seven}{\phantom{1}}
				\nodepart{eight}{\phantom{1}}
			};
			\fill[white] ([xshift=\pgflinewidth+1pt,
			yshift=-\pgflinewidth+1pt]queue2.north east)
			rectangle ([xshift=-\pgflinewidth-4pt,
			yshift=\pgflinewidth-1pt]queue2.south east);
			
			\node[server, anchor=east, fill=\darkgreen]
			(nu1) at ($(queue2.west)-(.1cm,0)+(0,.5cm)$)
			{$\nu_A$};
			\node[server, anchor=east, fill=\lightgreen]
			(nu2) at ($(queue2.west)-(.1cm,0)-(0,.5cm)$)
			{$\nu_B$};
			
			\node at ($(queue2.one)+(.1cm,.09cm)$)
			{$A$};
			\node at ($(queue2.two)+(.1cm,.09cm)$)
			{$B$};
			\node at ($(queue2.three)+(.1cm,.09cm)$)
			{$B$};
			\node at ($(queue2.four)+(.1cm,.09cm)$)
			{$A$};
			
			\draw[->] ($(nu1.west)!.5!(nu2.west)-(.1cm,0)$)
			-- ($(nu1.west)!.5!(nu2.west)-(.5cm,0)$)
			|- ($(queue1.west)-(.1cm,0)$);
			\draw[->] ($(mu3.east)+(.1cm,0)$)
			-- ($(mu3.east)+(.5cm,0)$)
			|- ($(queue2.east)+(.1cm,0)$);
			
			\node at ($(mu1.east)+(.85cm,0)$) {};
			\node at ($(nu1.west)-(.85cm,0)$) {};
		\end{tikzpicture}%
	}
	\caption{Closed tandem network
		of two \gls{pands} queues
		associated with the cluster
		of \figurename~\ref{fig:app-intro-assignment},
		assuming that $\ell_k = 2$
		for each $k \in \K$
		and $\ell_s = 2$
		for each $s \in \Se$.%
	}
	\label{fig:app-network}
\end{figure}

\paragraph{Second queue}

We now provide a symmetric description
for the second queue,
which contains available tokens.
This queue is again a multi-server queue
like that of Example~\ref{ex:pands},
but the servers correspond to job types
and not to machines.
More specifically,
the set of servers is $\K$ and,
for each $k \in \K$,
the service rate of server~$k$
is equal to $\nu_k$,
the arrival rate of type-$k$ jobs
in the cluster.
Again, even though
we use the same set $\K$ to index
the set of servers in the second queue
and the set of token classes associated with job types,
we will always specify whether we are referring
to a server~$k \in \K$ or a token class $k \in \K$.
For each $k \in \K$,
the set of servers that can process
class-$k$ tokens is $\K_k = \{k\}$.
Also, for each $s \in \Se$,
the set of servers that can
process class-$s$ tokens is $\K_s$,
corresponding to
the set of job types that can seize
a token from machine~$s$ in the cluster.
In the end, for each $k \in \K$,
server~$k$ can process
the tokens that belong to type-$k$ jobs
plus the tokens that belong
to the machines in $\Se_k$.
For example, the compatibility graph
of the second queue
in the tandem network associated with
the assignment graph
of \figurename~\ref{fig:app-intro-assignment}
is shown
in \figurename~\ref{fig:app-intro-det-msq2}.
Each server processes its
compatible tokens
in \gls{fcfs} order.

Due to the placement order,
if a class-$s$ token is in the second queue,
then, for each $k \in \K_s$,
all class-$k$ tokens
are also in the second queue,
at positions that precede the position
of the class-$s$ token.
Given the compatibility graph,
this implies that a class-$s$ token
will never be \textit{in service}
in this queue.
The only way a class-$s$ token
leaves this queue is
if a token of a class $k \in \K_s$
completes service
and ejects this class-$s$ token.
In the cluster,
this means that
a type-$k$ job enters
and seizes a token from machine~$s$.

We let $d = (d_1, \ldots, d_m) \in \I^*$
denote the state of the second queue.
As observed before,
the placement order guarantees that,
if there is a token of a class $s \in \Se$
at some position $p \in \{1, \ldots, m\}$
in the second queue,
then, for each $k \in \K_s$,
each class-$k$ token is also in the second queue,
at a position $q \in \{1, \ldots, p-1\}$
that precedes that of this class-$s$ token.
Therefore, the state space of
the state of the second queue
is a strict subset $\D$ of the set of sequences
$d \in \I^*$ such that $|d| \le \ell$.
The overall service rate in this queue
is equal to the sum of the arrival rates
of the job types that can seize
at least one available token, given by
\begin{equation} \label{eq:nu}
	\nu(d)
	= \sum_{k \in \bigcup_{p=1}^m \K_{d_p}} \nu_k.
\end{equation}
For each $p \in \{1, \ldots, m\}$,
the token in position~$p$
can be seized or moved
by the incoming jobs that are compatible
with this token
but not with the older available tokens,
and these jobs arrive at rate
\begin{equation} \label{eq:delta-nu}
	\Delta\nu(d_1,\ldots,d_p)
	= \sum_{k \in \K_{d_p} \setminus \bigcup_{q=1}^{p-1} \K_{d_q}} \nu_k.
\end{equation}
The functions $\nu$ and $\Delta\nu$
play the same role for the second queue
as the functions $\mu$ and $\Delta\mu$,
given by~\eqref{eq:mu} and \eqref{eq:delta-mu},
for the first queue.
Again because of the placement order,
\eqref{eq:delta-nu} simplifies to
$\Delta\nu(d_1, \ldots, d_p) = 0$
for each $p \in \{1, \ldots, m\}$
such that $d_p \in \Se$.
The compatibility graph also
guarantees that,
for each $p \in \{1, \ldots, m\}$
such that $k = d_p \in \K$,
we have $\Delta\nu(d_1, \ldots, d_p) = \nu_k$
if $|(d_1, \ldots, d_{p-1})|_k = 0$
and $\Delta\nu(d_1, \ldots, d_p) = 0$
otherwise.

Now assume that a token
in some position $p \in \{1, \ldots, m\}$ completes service.
If $k = c_p \in \K$
denotes this token's class,
then, because of the \gls{pands} mechanism,
only one of these two transitions can occur:
\begin{enumerate}[label=(\roman*)]
	\item \label{item:trans-3} If the second queue
	contains a token of a class in $\Se_k$
	(necessarily in position at least $p+1$
	because of the placement order),
	the class-$k$ token replaces
	the first of these tokens,
	say of class $s \in \Se_k$,
	and the ejected class-$s$ token joins
	the first queue.
	In the cluster, this means that
	an incoming type-$k$ job
	seizes a token from machine~$s$
	because this was the oldest
	available token of a compatible machine.
	\item \label{itemi:trans-4} If there is no token
	of a class in $\Se_k$ in the second queue,
	the class-$k$ token leaves this queue
	and is added to the first queue.
	In the cluster, this means that
	a type-$k$ job enters
	and does not find any available token
	from a machine in $\Se_k$,
	so that this job seizes
	a class-$k$ token
	and will hold this token until it is assigned
	to the buffer of a machine in $\Se_k$.
\end{enumerate}
In both cases, a token leaves
the second queue
and joins the first,
meaning that
a new job arrives and
seizes a token in the cluster.
An example of a type-\ref{item:trans-3} transition
is shown in \figurename~\ref{fig:app-network}
for the cluster
of \figurename~\ref{fig:app-intro-assignment}.
Indeed, from \figurename~\ref{fig:app-network-3}
to \figurename~\ref{fig:app-network-4},
the oldest class-$B$ token completes service
in the second queue;
this token replaces the class-$2$
token, which joins the first queue.
In the cluster,
this means that a type-B job
enters and seizes a token of machine~2.

\begin{remark}
	The second queue is degenerate
	in the sense that,
	accounting for the placement
	and compatibility graph,
	class-$k$ tokens are the only tokens that
	can be processed by server~$k$.
	Therefore,
	if we let $\A = \{k \in \K: |d|_k > 0\}$,
	then, for each $k \in \A$,
	the service rate of
	the oldest class-$k$ token
	is~$\nu_k$,
	irrespective of the order
	of the tokens
	in state~$d$
	(provided that this state
	adheres to the placement order).
	This implies that the relative order
	of the tokens
	of the classes in $\K$ in state $d$
	modifies neither their service rates
	nor the departure rate of the tokens
	of the classes in $\Se$.
	On the contrary, in general,
	the relative order of the tokens
	of the classes in $\Se$
	modifies their departure rate.
	A similar remark
	could be made for the first queue
	by exchanging the roles
	of the sets $\K$ and $\Se$.
\end{remark}

\subsection{Generalization
	to other resource-management protocols}
\label{subsec:app-general}

The important thing
to remember from Section~\ref{subsec:app-intro}
is that we can describe the dynamics
of tokens in a machine cluster
using a closed tandem network
of two \gls{pands} queues,
so that the first queue contains
tokens held by jobs present in the cluster
and the other queue contains available tokens.
In Section~\ref{subsubsec:app-general-def},
we propose a more general framework
based on the same idea.
This framework extends
the example of Section~\ref{subsec:app-intro}
in two ways: it allows not only
for more general
compatibility constraints
between job types and machines,
but also for multiple levels
of preferences between tokens.
Sections~\ref{subsubsec:app-distributed}
and \ref{subsubsec:app-hierarchical}
give a prototypical example
for each extension.

\subsubsection{Queueing model}
\label{subsubsec:app-general-def}

Consider a closed tandem network
of two \gls{pands} queues
like that described in Section~\ref{subsec:two}.
Let $\I = \{1, \ldots, I\}$ denote
the set of token classes
and $\prec$ the placement order
of this network.
Recall that, for each $i, j \in \I$
such that $i \prec j$,
class-$i$ tokens
precede (resp.\ succeed) class-$j$ tokens
in the first (resp.\ second) queue.
The swapping graph of the queues
is simply the underlying undirected graph
of the placement graph.
For each $i \in \I$, let
$\ell_i$ denote the number
of class-$i$ tokens in the network.
We assume that both \gls{pands} queues
are multi-server queues,
like that described in Example~\ref{ex:pands}.
Their compatibility graphs will
be described in the next paragraphs.
The applications that we have
in mind again involve tokens in a cluster,
and in these applications the first queue
of the tandem network will contain
tokens held by jobs present in the system
and the second queue
will contain available tokens. 

Let us first describe
the compatibility constraints in the first queue,
as we did in Example~\ref{ex:oi}.
Let $\Se = \{1, \ldots, S\}$ denote
the set of servers in this first queue
and, for each $s \in \Se$,
$\mu_s$ the service rate of server~$s$.
For each class $i \in \I$
that is minimal
with respect to the placement order $\prec$
(that is,  there is
no class $j \in \I$
with $i \succ j$),
we let $\Se_i \subseteq \Se$ denote
the set of servers that can process
class-$i$ tokens in the first queue.
This defines a bipartite graph between
the set of minimal classes
and the set of servers.
The set of servers that can process
non-minimal classes
is defined by an ascending recursion
over the placement order.
More specifically,
for each class $i \in \I$
that is not minimal with respect to
the placement order,
the set of servers that can
process class-$i$ tokens is
$\Se_i =
\bigcup_{j \in \I: j \prec i} \Se_j$.
Going back to the example of
\figurename~\ref{fig:app-intro-det},
we have that the set of
classes is $\{1,2,3,A,B\}$ and the set of
minimal classes
is $\{1,2,3\}$;
the sets of servers associated
with these classes in the first queue are
$\Se_{1} = \{1\}$, $\Se_{2} = \{2\}$,
and $\Se_{3} = \{3\}$,
while the sets of servers associated
with the classes that are not minimal are
$\Se_{A} = \Se_{1} \cup \Se_{3} = \{1,3\}$
and
$\Se_{B} =\Se_{2} \cup \Se_{3} = \{2,3\}$.
The state of the first queue
is denoted by $c = (c_1, \ldots, c_n)$
and the overall and individual service rates
in this queue are given
by~\eqref{eq:mu} and \eqref{eq:delta-mu},
respectively.

Similarly,
we let $\K = \{1, \ldots, K\}$ denote
the set of servers in the second queue and,
for each $k \in \K$,
$\nu_k$ the service rate of server~$k$.
For each class $i \in \I$ that is maximal
with respect to the placement order~$\prec$
(that is, there is no class $j \in \I$
with $i \prec j$),
we let $\K_i \subseteq \K$ denote
the set of servers
that can process class-$i$ tokens.
This defines a bipartite graph between
the set of maximal classes
and the set of servers.
The set of servers that can process
non-maximal classes
is defined by a descending recursion
over the placement order.
More specifically,
for each class $i \in \I$
that is not maximal,
the set of servers that can process
class-$i$ tokens is
$\K_i = \bigcup_{j \in \I: i \prec j} \K_j$.
By again considering the example of
\figurename~\ref{fig:app-intro-det},
we have that the set of maximal classes
is $\{A,B\}$;
the sets of servers associated
with these classes
in the second queue are
$\K_{A} = \{A\}$ and $\K_{B} = \{B\}$,
while the sets of servers associated
with the classes that are not maximal are
$\K_{1} = \K_{A} = \{A\}$,
$\K_{2} = \K_{B} = \{B\}$,
and $\K_{3} = \K_{A} \cup \K_{B} = \{A,B\}$.
The state of the second queue
is denoted by $d = (d_1, \ldots, d_m)$
and the overall and individual service rates
in this queue are given
by~\eqref{eq:nu} and \eqref{eq:delta-nu},
respectively.

In this new framework,
the placement order describes
not only priorities between classes
but also compatibilities
between classes and servers.
Using this observation,
we will now see that
the structure of the closed tandem network
can be described more compactly by a mixed graph
(that is, a graph with both
directed and undirected edges).
The mixed graph associated with the model
of \figurename~\ref{fig:app-intro-det}
is shown in \figurename~\ref{fig:app3-app1}.
The subgraph induced in this mixed graph
by the set of classes
describes the placement order.
The subgraph induced
by the set of minimal classes
and the set of machines,
as shown at the bottom
of \figurename~\ref{fig:app3-app1},
describes the compatibilities
between the minimal classes
and the servers of the first queue.
The set of servers that can serve
a non-minimal class
is the union of the sets of servers
that can serve the ancestors of this class.
Similarly, the subgraph induced
by the sets of maximal classes
and the set of job types,
as shown at the top
of \figurename~\ref{fig:app3-app1},
describes the compatibilities
between the maximal classes
and the servers of the second queue.
The set of servers that can serve
a non-maximal class
is the union of the sets of servers
that can serve the descendants of this class.
\figurename s~\ref{fig:app3-app1-nopriority}
and \ref{fig:app3-hierarchical}
show more elaborate examples of mixed graphs
that will be studied
in Sections~\ref{subsubsec:app-distributed}
and \ref{subsubsec:app-hierarchical}.

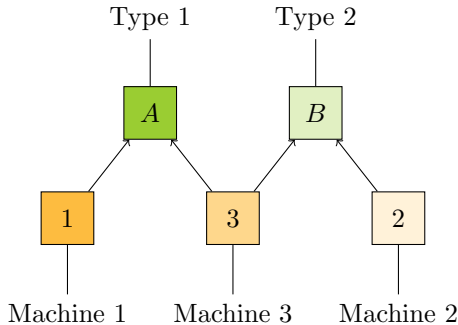
\begin{figure}[ht]
	\centering
	\begin{tikzpicture}
		\def\width{2.2cm}
		\def\height{1.4cm}
		
		\node[class, darkgreen]
		(c1) {$A$};
		\node[class, lightgreen]
		(c2) at ($(c1)+(\width,0)$)
		{$B$};
		
		\node[class, darkorange] (cA)
		at ($(c1)-(0.5*\width,0)-(0,\height)$)
		{$1$};
		\node[class, neutralorange] (cB)
		at ($(c1)!.5!(c2)-(0,\height)$)
		{$3$};
		\node[class, lightorange] (cC)
		at ($(c2)+(0.5*\width,0)-(0,\height)$)
		{$2$};
		
		\node (t1)
		at ($(c1)+(0,.9*\height)$) {Type~1};
		\node (t2)
		at ($(c2)+(0,.9*\height)$) {Type~2};
		\node (s1)
		at ($(cA)-(0,.9*\height)$) {Machine~1};
		\node (s3)
		at ($(cB)-(0,.9*\height)$) {Machine~3};
		\node (s2)
		at ($(cC)-(0,.9*\height)$) {Machine~2};
		
		\draw[-] (t1) -- (c1);
		\draw[-] (t2) -- (c2);
		\draw[<-] (c1) -- (cA);
		\draw[<-] (c1) -- (cB);
		\draw[<-] (c2) -- (cB);
		\draw[<-] (c2) -- (cC);
		\draw[-] (cA) -- (s1);
		\draw[-] (cB) -- (s3);
		\draw[-] (cC) -- (s2);
	\end{tikzpicture}
	\caption{Mixed graph associated
		with the example
		of Section~\ref{subsec:app-intro}.}
	\label{fig:app3-app1}
\end{figure}

As in Section~\ref{subsec:app-intro},
some tokens may never be in service
in a queue.
In fact, in the first queue,
the only tokens that can be in service
are those of the classes
that are minimal with respect
to the placement order.
A token of a class $i \in \I$
that is not minimal
can only leave this queue
upon the service completion of
a token of a minimal class~$j \in \I$
such that $j \prec i$.
Similarly, only tokens
of classes that are maximal
with respect to the placement order
can be in service in the second queue.
A token of a class $i \in \I$
that is not maximal
can only leave this queue
upon the service completion of
a token of a maximal class~$j \in \I$
such that $i \prec j$.

Applying the results
of Section~\ref{subsubsec:two-steady}
allows us to directly derive
a closed-form expression for
the stationary distribution
of the network state.
We adopt the notation of that section.
In particular, the state space
of the Markov process associated with
the network state $(c;d)$
is denoted by $\Sigma$
and characterized by~\eqref{eq:two-partition}.
Assuming that this Markov process
is irreducible,
it follows from
Theorem~\ref{theo:two-steady-cd} that
its stationary distribution is given by
\begin{equation*}
	\pi(c;d) =
	\frac1G \left(
	\prod_{p = 1}^n \frac1{\mu(c_1, \ldots, c_p)}
	\right) \left(
	\prod_{p = 1}^m \frac1{\nu(d_1, \ldots, d_p)}
	\right),
	\quad \forall (c; d) \in \Sigma,
\end{equation*}
where the constant~$G$ follows from normalization.

We now consider two examples
that illustrate the descriptive
power of this new framework.
Section~\ref{subsubsec:app-distributed}
gives an extension
of the introductory example
of Section~\ref{subsec:app-intro}
to a cluster where jobs can be
distributed over several machines.
Section~\ref{subsubsec:app-hierarchical}
looks at a token-based hierarchical
load-distribution protocol.
These two examples
can be considered independently.

\subsubsection{Distributed processing}
\label{subsubsec:app-distributed}

As in Section~\ref{subsec:app-intro},
we consider a cluster
that consists of a dispatcher and
a set $\Se = \{1, \ldots, S\}$ of machines.
The set of job types
is denoted by $\K = \{1, \ldots, K\}$ and,
for each $k \in \K$,
type-$k$ jobs arrive according to
an independent Poisson process
with rate $\nu_k$
and have independent and
exponentially distributed sizes
with unit mean.
An incoming job may be
assigned to the buffer(s) of
one (or more) machine(s),
left unassigned for now, or rejected,
depending on the type of this job
and on the system state.

The difference with Section~\ref{subsec:app-intro}
is that an incoming job is no longer
assigned to a single machine; instead,
it is assigned to all machines in a group.
More specifically, if a job is assigned to a group of machines, this means that this job is added to the buffer of every machine in this group, and that these machines
will subsequently be able to process
this job in parallel.
A job is said to be unassigned (or uncommitted)
if it has not been assigned to a group yet.
We let $\T = \{1, \ldots, T\}$ denote
the set of group indices and, for each $t \in \T$,
$\Se_t \subseteq \Se$ the set of machines
that belong to group~$t$
and $\K_t \subseteq \K$
the set of job types that can be
assigned to group~$t$.
With a slight abuse of notation,
we also let $\T_s \subseteq \T$ denote
the set of groups that include machine~$s$,
for each $s \in \Se$,
and $\T_k \subseteq \T$
the set of groups to which type-$k$ jobs can be assigned,
for each $k \in \K$.
This defines a tripartite
\textit{assignment} graph
between job types, groups, and machines,
as shown in
\figurename~\ref{fig:app-distributed-assignment}.
The introductory example
of Section~\ref{subsec:app-intro}
corresponds to the special case
where there is a one-to-one correspondence
between groups and machines,
that is, $T = S$ and
$\Se_t = \{t\}$ for each $t \in \T$.

\begin{figure}[ht]
	\centering
	\begin{tikzpicture}
		\def\width{2.6cm}
		\def\height{1.6cm}
		
		\node (t1) {Type A};
		\node (t2)
		at ($(t1)+(\width,0)$) {Type B};
		
		\node (p1)
		at ($(t1)-(0,\height)$) {Group 1};
		\node (p2)
		at ($(t2)-(0,\height)$) {Group 2};
		
		\node (s1)
		at ($(p1)+(-0.5*\width,-\height)$)
		{Machine 1};
		\node (s3)
		at ($(p1)!.5!(p2)-(0,\height)$)
		{Machine 3};
		\node (s2)
		at ($(p2)+(0.5*\width,-\height)$)
		{Machine 2};
		
		\draw (t1) -- (p1);
		\draw (t1) -- (p2);
		\draw (t2) -- (p2);
		\draw (p1) -- (s1);
		\draw (p1) -- (s3);
		\draw (p2) -- (s3);
		\draw (p2) -- (s2);
	\end{tikzpicture}
	\caption{A tripartite assignment graph
		between job types, groups, and machines.
		We have $\Se_1 = \{1,3\}$,
		$\Se_2 = \{2,3\}$,
		$\K_1 = \{A\}$, and $\K_2 = \{A,B\}$.
	}
	\label{fig:app-distributed-assignment}
\end{figure}
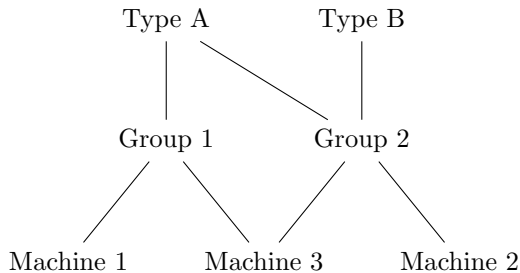

Each machine processes the jobs in its buffer
in \gls{fcfs} order, while ignoring other jobs. In particular, a job may be in service on multiple machines if it is at the head of their buffers.
For each $t \in \T$,
if a job assigned to group~$t$
is in service on a subset $\Se' \subseteq \Se_t$
of the machines of this group,
the departure rate
of this job is $\sum_{s \in \Se'} \mu_s$.

We generalize the \gls{alis} assignment rule
introduced in Section~\ref{subsubsec:alis}
as follows.
For each $t \in \T$,
the assignments of jobs to group~$t$
are regulated via
$\ell_t$ tokens, so that
each job seizes
one of these tokens
when it is assigned to group~$t$
and releases this token
upon service completion.
The dispatcher keeps
a list of available tokens,
sorted in their release order,
so that the longest available token
is at the head of this list.
An incoming type-$k$ job
can be assigned to group~$t$
if and only if
a token of this group is
available. As before,
at most $\ell_k$ type-$k$ jobs
can be left unassigned
if no token of a compatible group
is available upon their arrival.
Now, when a type-$k$ job arrives
in the system,
one of the following events occurs:
\begin{enumerate}[label=(\roman*)]
	\item if one or more tokens
	of the groups in $\T_k$ are available,
	the job seizes the one of these tokens
	that has been available the longest
	and is assigned to the corresponding group
	(so that the job is added
	to the buffers of all machines
	in the group);
	\item otherwise, if there are currently
	fewer than
	$\ell_k$ unassigned type-$k$ jobs,
	the incoming job is left unassigned
	until it can be assigned to
	one of its compatible groups;
	\item otherwise, the job is rejected.
\end{enumerate}
When a job assigned to group~$t$
completes service,
this job leaves the system immediately.
Its token is seized by
the oldest unassigned job of a type
in $\K_t$, if any,
otherwise it is added
to the dispatcher's list
of available tokens.
Furthermore, the machines
that were processing this job
immediately start processing
the next job in their buffer, if any.
This protocol can be seen
as a generalization of that
introduced in \cite{C19-1}
to a scenario where incoming jobs
are left unassigned
(instead of being rejected)
in the absence of available compatible tokens.
This is also a generalization
of the \gls{fcfs-alis} protocol,
in which each machine
processes the jobs in its buffer
in \gls{fcfs} order,
and each job is assigned to \emph{groups}
of machines using the above generalization
of the \gls{alis} assignment rule.

Following the same approach as in
Section~\ref{subsubsec:cancel-on-commit},
we can reinterpret
this generalization
of the \gls{fcfs-alis} protocol
as a generalized
redundancy scheduling protocol
that combines the
cancel-on-commit protocol
introduced in Section~\ref{subsubsec:cancel-on-commit}
and the cancel-on-complete protocol.
Indeed, in the above cluster,
everything works as if
an incoming type-$k$ job
were at first replicated over
all machines of the set
$\bigcup_{t \in \T_k} \Se_t$
and eventually committed to
a \emph{subset} of these machines,
those belong to a given group $t \in \T_k$.
In practice, this means that
a replica of an incoming type-$k$ job
is sent to every machine
within the set $\bigcup_{t \in \T_k} \Se_t$,
and, once the job is committed
to a group $t \in \T_k$,
the replicas sent to the machines
that are not in the set $\Se_t$
are canceled.
Subsequently, the remaining replicas
on the machines in $\Se_t$
are canceled whenever
one of them completes service.
In particular, several replicas
may be in service at the same time.

The dynamics of this system
can be described by the queueing model
of Section~\ref{subsubsec:app-general-def}
as follows.
We again describe the dynamics
of the cluster by looking at tokens,
so that all jobs present in the system
(and not only those assigned to a group)
held one of these tokens.
The set of token classes
is $\I = \K \sqcup \T$
and there are
$\ell_i$ class-$i$ tokens,
for each $i \in \I$.
The placement order is defined by
$t \prec k$
for each $t \in \T$ and $k \in \K_t$
(or equivalently,
$t \prec k$ for each
$k \in \K$ and $t \in \T_k$).
The minimal classes are those
associated with machine groups
and the maximal classes are those
associated with job types.
For each $t \in \T$,
the set of servers that can process
class-$t$ tokens in the first queue
is $\Se_t$,
corresponding to the set of machines
that belong to group~$t$.
For each $k \in \K$,
the set of servers that can process
class-$k$ tokens in the second queue
is $\{k\}$.
The mixed graph associated
with the example of
\figurename~\ref{fig:app-distributed-assignment}
is shown in
\figurename~\ref{fig:app3-app1-nopriority}.

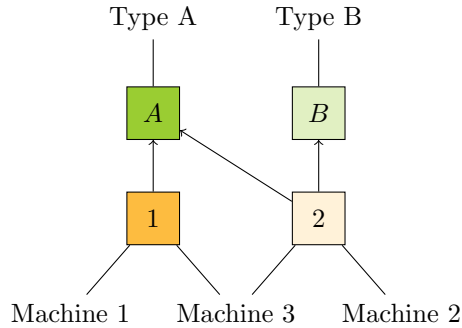
\begin{figure}[ht]
	\centering
	\begin{tikzpicture}
		\def\width{2.2cm}
		\def\height{1.4cm}
		
		\node[class, darkgreen] (c1)
		{$A$};
		\node[class, lightgreen] (c2)
		at ($(c1)+(\width,0)$) {$B$};
		\node[class, darkorange] (c3)
		at ($(c1)-(0,\height)$) {$1$};
		\node[class, lightorange] (c4)
		at ($(c2)-(0,\height)$) {$2$};
		
		\node (t1)
		at ($(c1)+(0,.9*\height)$) {Type A};
		\node (t2)
		at ($(c2)+(0,.9*\height)$) {Type B};
		\node (s1)
		at ($(c3)-(.5*\width,0)-(0,.9*\height)$)
		{Machine~1};
		\node (s3)
		at ($(c3)!.5!(c4)-(0,.9*\height)$) {Machine~3};
		\node (s2)
		at ($(c4)+(.5*\width,0)-(0,.9*\height)$)
		{Machine~2};
		
		\draw[-] (t1) -- (c1);
		\draw[-] (t2) -- (c2);
		\draw[<-] (c1) -- (c3);
		\draw[<-] (c1) -- (c4);
		\draw[<-] (c2) -- (c4);
		\draw[-] (c3) -- (s1);
		\draw[-] (c3) -- (s3);
		\draw[-] (c4) -- (s3);
		\draw[-] (c4) -- (s2);
	\end{tikzpicture}
	\caption{%
		Mixed graph associated
		with the cluster of
		\figurename~\ref{fig:app-distributed-assignment}.
	}
	\label{fig:app3-app1-nopriority}
\end{figure}

\subsubsection{Hierarchical load distribution}
\label{subsubsec:app-hierarchical}

Let~$H$ denote a positive integer.
We consider a cluster
that consists of a dispatcher
and $2^{H-1}$ machines, and we
denote by $\Se = \{1, \ldots, 2^{H-1}\}$
the set of machines.
Jobs arrive according to a Poisson process
with a positive rate~$\nu$.
Each incoming job is compatible with all machines
but will eventually be
assigned to and processed by a single machine.
Every machine has a buffer
of length~1, so that the only job assigned to this machine is also in service on this machine.
For each $s \in \Se$,
the service time of a job on machine~$s$
is exponentially distributed
with a positive rate~$\mu_s$.
The job arrivals within the system
are regulated via $2^H - 1$ tokens
numbered from 1 to $2^H - 1$.
A job that has not been assigned
to a machine yet
holds a token numbered
from 1 to $2^{H-1} - 1$,
while, for each
$s \in \{ 1, \ldots, 2^{H-1}\}$,
the job in service on machine~$s$
holds token~$2^{H-1} + s - 1$.
Initially, when the system
is empty of jobs,
all tokens are arranged in ascending order
in a list kept by the dispatcher,
with token~1 at the head of the list
and token~$2^H - 1$ at the end.
If a new job arrives
and there is at least one available token,
this job seizes the token
obtained by applying
the \gls{pands} mechanism
in the queue of available tokens,
starting from the token
at the head of the queue,
with the following swapping rule:
for each $i \in \{1, \ldots, 2^{H-1} - 1\}$,
token~$i$ can be swapped with
tokens $2i$ and $2i + 1$.
An incoming job is rejected if no token is available.
Conversely, a service completion
triggers the following chain reaction:
if token~$i$ is released by a job,
this token is seized
by the job that holds token $\lfloor i/2 \rfloor$
(so that this token is in turn released
and can be seized by another job), if any,
otherwise it is added
to the list of available tokens.

Priorities between tokens
can be represented by
a perfect binary tree of height~$H-1$
such that, for each
$i \in \{1, \ldots, 2^{H-1} - 1\}$,
the children of node~$i$
are nodes~$2i$ and $2i + 1$.
\figurename~\ref{fig:binary-tree} shows
an example with $H = 3$.
Leaf nodes correspond
to tokens held by jobs in service
on a machine.
For each $h \in \{1, \ldots, H-1\}$,
the nodes at depth~$h$
in the tree correspond to
tokens $2^{h-1}$ to $2^h - 1$.
A job holding one of these tokens
is $H-h$ steps away
from entering service on a machine.
Indeed, if a job holds a token
$i \in \{1, \ldots, 2^{H-1} - 1\}$
and a token that belongs to
the subtree rooted at node~$i$ is released,
this job will seize
either token $2i$ or token $2i + 1$,
thus getting one step closer
to entering service
on a machine.

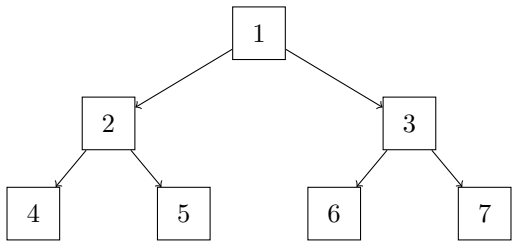
\begin{figure}[ht]
	\centering
	\begin{tikzpicture}
		\def\width{1cm}
		\def\height{1.2cm}
		
		\node[class] (c1) {1};
		
		\node[class] (c2)
		at ($(c1)-(2*\width,0)-(0,\height)$) {2};
		\node[class] (c3)
		at ($(c1)+(2*\width,0)-(0,\height)$) {3};
		
		\node[class] (c4)
		at ($(c2)+(-\width,-\height)$) {4};
		\node[class] (c5)
		at ($(c2)+(\width,-\height)$) {5};
		\node[class] (c6)
		at ($(c3)+(-\width,-\height)$) {6};
		\node[class] (c7)
		at ($(c3)+(\width,-\height)$) {7};
		
		\draw[->] (c1) -- (c2);
		\draw[->] (c1) -- (c3);
		\draw[->] (c2) -- (c4);
		\draw[->] (c2) -- (c5);
		\draw[->] (c3) -- (c6);
		\draw[->] (c3) -- (c7);
	\end{tikzpicture}
	\caption{A perfect binary tree
		of height $H - 1 = 2$.}
	\label{fig:binary-tree}
\end{figure}

The corresponding queueing model,
based on the framework
of Section~\ref{subsubsec:app-general-def},
is defined as follows.
The set of token classes is
$\I = \{1, \ldots, 2^H - 1\}$.
For each $i \in \I$,
there is a single class-$i$ token
which corresponds, in the cluster,
to token~$i$.
The placement order is defined as follows:
for each $i \in \{1, \ldots, 2^{H-1} - 1\}$,
$i \succ 2i$ and $i \succ 2i + 1$
(so that the placement graph is obtained
by reversing edges
in the perfect binary tree
defined in the previous paragraph).
In particular,
if the queue of available tokens is not empty,
the token at the head of this queue
is necessarily token~$1$.
The set of servers
in the first queue is
$\Se = \{1, \ldots, 2^{H-1}\}$.
The set of minimal token classes
is $\{2^{H-1}, 2^{H-1} + 1, \ldots, 2^H - 1\}$
and, for each $s \in \{1, \ldots, 2^{H-1}\}$,
class $2^{H-1} + s - 1$
is compatible with server~$s$.
In the second queue,
there is a single server of rate $\nu$.
The only maximal class is class~$1$
and this class is compatible with this server.
The mixed graph associated with
the perfect binary tree of
\figurename~\ref{fig:binary-tree}
is shown in \figurename~\ref{fig:app3-hierarchical}.

\begin{figure}[b]
	\centering
	\begin{tikzpicture}
		\def\width{1cm}
		\def\height{1.2cm}
		
		\node[class] (c1) {1};
		
		\node[class] (c2)
		at ($(c1)-(2*\width,0)-(0,\height)$) {2};
		\node[class] (c3)
		at ($(c1)+(2*\width,0)-(0,\height)$) {3};
		
		\node[class] (c4)
		at ($(c2)+(-\width,-\height)$) {4};
		\node[class] (c5)
		at ($(c2)+(\width,-\height)$) {5};
		\node[class] (c6)
		at ($(c3)+(-\width,-\height)$) {6};
		\node[class] (c7)
		at ($(c3)+(\width,-\height)$) {7};
		
		\node (t1)
		at ($(c1)+(0,.9*\height)$) {Type A};
		\node (s1)
		at ($(c4)-(0,.9*\height)$) {Machine~1};
		\node (s3)
		at ($(c5)-(0,.9*\height)$) {Machine~2};
		\node (s2)
		at ($(c6)-(0,.9*\height)$) {Machine~3};
		\node (s4)
		at ($(c7)-(0,.9*\height)$) {Machine~4};
		
		\draw[-] (t1) -- (c1);
		\draw[<-] (c1) -- (c2);
		\draw[<-] (c1) -- (c3);
		\draw[<-] (c2) -- (c4);
		\draw[<-] (c2) -- (c5);
		\draw[<-] (c3) -- (c6);
		\draw[<-] (c3) -- (c7);
		\draw[-] (c4) -- (s1);
		\draw[-] (c5) -- (s3);
		\draw[-] (c6) -- (s2);
		\draw[-] (c7) -- (s4);
	\end{tikzpicture}
	\caption{Mixed graph associated with
		the perfect binary tree
		of \figurename~\ref{fig:binary-tree}.}
	\label{fig:app3-hierarchical}
\end{figure}
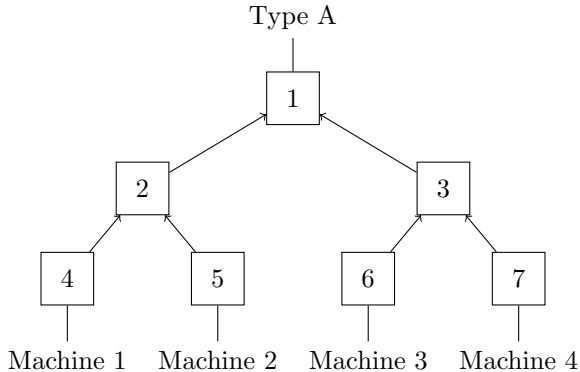

This hierarchical load-distribution strategy
could be generalized by considering
a perfect $a$-ary tree, with $a \ge 2$,
or a directed rooted tree,
so that each node represents a class of tokens
and the tokens associated with leaf nodes
give access to machines.
By combining this idea with that of
Section~\ref{subsubsec:app-distributed},
we could also consider
a directed acyclic graph
and associate
a job type with one or more nodes without ancestor
and a machine or a group of machines
with each node without descendant.
As for previous cluster models,
it is also possible to propose an alternative interpretation of this model
using redundancy scheduling.

\section{Conclusion} \label{sec:ccl}

In this paper, we introduced pass-and-swap (\gls{pands}) queues, an extension of order-independent (OI) queues in which,
upon a service completion,
customers move along the queue and swap positions with other customers, depending on compatibilities defined by a so-called swapping graph.
We showed that a stable \gls{pands} queue is quasi-reversible and that, surprisingly, its product-form stationary distribution
is independent of its swapping graph.
We then studied networks of \gls{pands} queues.
Although deriving the stationary distribution of open networks is a straightforward application of quasi-reversibility,
the case of closed networks is more intricate because the Markov process describing the network state over time is not necessarily irreducible.
For closed networks with one or two queues
and a deterministic routing process,
we observed that the \gls{pands} mechanism allows for the enforcement of priorities between classes, in the sense that a customer cannot leave a queue
before all customers of the classes with higher priority leave it. Finally, we showed that such closed networks describe the dynamics of the loss variants of several token-based load-distribution protocols, such \gls{fcfs-alis} and multiple redundancy-scheduling protocols.

This work suggests that we still do not have a complete picture of all queueing dynamics that lead to a product-form stationary distribution, which leaves open an important avenue for further study. 
Another open question is formed by the irreducibility of the Markov process underlying closed networks.
While we established irreducibility of this Markov process under the condition that, at any point in time, each customer has a positive service rate, 
the characterization of irreducibility properties of the Markov process underlying general closed networks, and their impact on the stationary distribution (along with its product-form nature), remains an open question. 
A different direction of further research entails the applications of \gls{pands} queues. In particular, in Section~\ref{sec:app}, we regarded applications based on multi-server queues as defined in Example~\ref{ex:oi}. Although applications with arbitrary customer-server compatibilities, such as load-balancing and resource-management protocols in computer systems, form the motivation for this work, we believe that \gls{pands} queues can be successfully applied to other systems involving priorities. This would require the use of more general \gls{pands} queues than just multi-server queues.

\paragraph{Acknowledgement}
We are thankful to Thomas Bonald
and Fabien Mathieu
for their useful comments
and for coining
the name ``pass-and-swap queues''.
We thank Sem Borst
for some helpful discussions
and for valuable remarks on an earlier draft of this paper.
The authors also wish to thank an anonymous associate editor for several astute remarks about the contents and exposition of the paper, and for suggesting the \emph{cancel-on-commit} redundancy protocol (including its name).
The research of Jan-Pieter Dorsman is supported by the Netherlands Organisation for Scientific Research (NWO) through Gravitation-grant NETWORKS-024.002.003.


\begin{thebibliography}{10}
	
	\bibitem{AKRW18}
	I.J.B.F. Adan, I.~Kleiner, R.~Righter, and G.~Weiss.
	\newblock {FCFS} parallel service systems and matching models.
	\newblock {\em Performance Evaluation}, 127-128:253--272, 2018.
	
	\bibitem{AW12}
	I.J.B.F. Adan and G.~Weiss.
	\newblock A loss system with skill-based servers under assign to longest idle
	server policy.
	\newblock {\em Probability in the Engineering and Informational Sciences},
	26(3):307--321, 2012.
	
	\bibitem{AW14}
	I.J.B.F. Adan and G.~Weiss.
	\newblock A skill based parallel service system under {FCFS}-{ALIS} — steady
	state, overloads, and abandonments.
	\newblock {\em Stochastic Systems}, 4(1):250--299, 2014.
	
	\bibitem{ABDV19}
	U.~Ayesta, T.~Bodas, J.L. Dorsman, and I.M. Verloop.
	\newblock A token-based central queue with order-independent service rates.
	\newblock 2020.
	\newblock To appear in \textit{Operations Research}.
	
	\bibitem{ABV18}
	U.~Ayesta, T.~Bodas, and I.M. Verloop.
	\newblock On a unifying product form framework for redundancy models.
	\newblock {\em Performance Evaluation}, 127-128:93--119, 2018.
	
	\bibitem{bcmp}
	F.~Baskett, K.M. Chandy, R.R. Muntz, and F.G. Palacios.
	\newblock Open, closed, and mixed networks of queues with different classes of
	customers.
	\newblock {\em Journal of the {ACM}}, 22(2):248--260, 1975.
	
	\bibitem{BK95}
	S.A. Berezner, C.F. Kriel, and A.E. Krzesinski.
	\newblock Quasi-reversible multiclass queues with order independent departure
	rates.
	\newblock {\em Queueing Systems}, 19(4):345--359, 1995.
	
	\bibitem{BK96}
	S.A. Berezner and A.E. Krzesinski.
	\newblock Order independent loss queues.
	\newblock {\em Queueing Systems}, 23(1-4):331--335, 1996.
	
	\bibitem{BC17}
	T.~Bonald and C.~Comte.
	\newblock Balanced fair resource sharing in computer clusters.
	\newblock {\em Performance Evaluation}, 116:70--83, 2017.
	
	\bibitem{BD11}
	R.J. Boucherie and N.M.~van Dijk, editors.
	\newblock {\em Queueing networks: {A} fundamental approach}.
	\newblock International Series in Operations Research \& Management Science.
	Springer {US}, 2011.
	
	\bibitem{C11}
	X.~Chao.
	\newblock Networks with customers, signals, and product form solution.
	\newblock In {\em Queueing Networks: {A} Fundamental Approach}, International
	Series in Operations Research \& Management Science, pages 217--267.
	Springer, Boston, MA, 2011.
	
	\bibitem{C19-1}
	C.~Comte.
	\newblock Dynamic load balancing with tokens.
	\newblock {\em Computer Communications}, 144:76--88, 2019.
	
	\bibitem{C19-2}
	C.~Comte.
	\newblock {\em Resource management in computer clusters: algorithm design and
		performance analysis}.
	\newblock phdthesis, Institut Polytechnique de Paris, 2019.
	
	\bibitem{GR20}
	K.S. Gardner and R.~Righter.
	\newblock Product forms for {FCFS} queueing models with arbitrary server-job
	compatibilities: an overview.
	\newblock {\em Queueing Systems}, 96(1):3--51, 2020.
	
	\bibitem{G16}
	K.S. Gardner, S.~Zbarsky, S.~Doroudi, M.~Harchol-Balter, E.~Hyytiä, and
	A.~Scheller-Wolf.
	\newblock Queueing with redundant requests: Exact analysis.
	\newblock {\em Queueing Systems}, 83(3-4):227--259, 2016.
	
	\bibitem{Jackson1}
	J.R. Jackson.
	\newblock Networks of waiting lines.
	\newblock {\em Operations Research}, 5(4):518--521, 1957.
	
	\bibitem{kelly}
	F.P. Kelly.
	\newblock {\em Reversibility and stochastic networks}.
	\newblock Cambridge University Press, 2011.
	
	\bibitem{KW82}
	F.P. Kelly and J.~Walrand.
	\newblock Networks of quasi-reversible nodes.
	\newblock In {\em Applied Probability-Computer Science: The Interface Volume
		1}, Progress in Computer Science, pages 3--29. Birkhäuser Boston, 1982.
	
	\bibitem{K11}
	A.E. Krzesinski.
	\newblock Order independent queues.
	\newblock In R.J. Boucherie and N.M.~van Dijk, editors, {\em Queueing networks:
		{A} fundamental approach}, number 154 in International Series in Operations
	Research \& Management Science, pages 85--120. Springer US, 2011.
	
	\bibitem{mshcc}
	A.E. Krzesinski and R.~Schassberger.
	\newblock Product form solutions for multiserver centers with hierarchical
	concurrency constraints.
	\newblock {\em Probability in the Engineering and Informational Sciences},
	6(2):147--156, 1992.
	\newblock Publisher: Cambridge University Press.
	
	\bibitem{msccc}
	J.-Y. {Le}~Boudec.
	\newblock A {BCMP} extension to multiserver stations with concurrent classes of
	customers.
	\newblock {\em SIGMETRICS Performance Evaluation Review}, 14(1):78--91, 1986.
	
	\bibitem{MBM20}
	P.~Moyal, A.~Busic, and J.~Mairesse.
	\newblock A product form for the general stochastic matching model.
	\newblock {\em {arXiv}:1711.02620 [math]}, 2020.
	
	\bibitem{M72}
	R.R. Muntz.
	\newblock Poisson departure processes and queueing networks.
	\newblock {\em {IBM} Thomas J. Watson Research Centre}, 1972.
	
	\bibitem{S99}
	R.~Serfozo.
	\newblock {\em Introduction to stochastic networks}.
	\newblock Stochastic Modelling and Applied Probability. Springer-Verlag, 1999.
	
	\bibitem{D11}
	N.M. Van~Dijk.
	\newblock On practical product form characterizations.
	\newblock In R.J. Boucherie and N.M.~van Dijk, editors, {\em Queueing networks:
		{A} fundamental approach}, International Series in Operations Research \&
	Management Science, pages 1--83. Springer {US}, Boston, {MA}, 2011.
	
\end{thebibliography}

\appendix
\section*{Appendix}

\section{Proof of Theorem~\ref{theo:pands}} \label{app:pands}

Consider a \gls{pands} queue as defined in Section~\ref{subsec:pands-def},
with a set $\I = \{1,\ldots,I\}$ of customer classes,
per-class arrival rates $\lambda_1, \ldots, \lambda_I$,
and a rate function $\mu$.
Also, for each $i \in \I$,
let $\I_i \subseteq \I$ denote
the set of customer classes
that can be swapped with class~$i$.
As announced in the sketch of proof
that followed Theorem~\ref{theo:pands},
our objective is to prove
that the balance function~$\Phi$
defined by~\eqref{eq:oi-Phic}
satisfies~\eqref{eq:pands-partial-balance-3}.

\paragraph{Rewriting \eqref{eq:pands-partial-balance-3}}

We first need to specify,
for each $c = (c_1,\ldots,c_n) \in \I^*$
and $i \in \I$,
all transitions that lead to state~$c$
by the departure of a class-$i$ customer.
To this end, we will identify
all states $d \in \I^*$
and positions $p \in \{1, \ldots, n+1\}$
such that $\delta_p(d) = (c, i)$.
The following notation will be convenient.
For each $c = (c_1,\ldots,c_n) \in \I^*$
and $d = (d_1,\ldots,d_m) \in \I^*$,
we let $c,d = (c_1,\ldots,c_n,d_1,\ldots,d_m)$
denote the state obtained by concatenation.
If~$d$ contains a single class-$i$ customer,
that is $d = (i)$,
then we simply write $c,i$ for $c,(i)$
and $i,c$ for $(i),c$.
For each sequence
$c = (c_1,\ldots,c_n) \in \I^*$
and positions $p, q \in \{1, \ldots, n\}$
with $p \leq q$,
we let $c_{p \ldots q} = (c_p,\ldots,c_q)$.
Finally, we adopt the convention that
$c_{p \ldots q} = \emptyset$ if $p > q$.

Now that all required notation is introduced,
we proceed with the identification.
Let $c = (c_1,\ldots,c_n) \in \I^*$
and $i \in \I$.
Furthermore, we set $q_0 = n+1$ and $i_0 = i$.
Moving from tail to head in state~$c$
(that is, from position~$n$ to position~$1$),
we determine the positions and classes
of the customers that may be involved in
a transition that leads to state~$c$
by the departure of a class-$i$ customer.
We now distinguish between multiple cases, based on the total number~$v$ of customers that move during the transition:

\begin{description}
	\item[Case $v = 1$]
	A single customer was involved
	in the transition,
	namely the customer of class~$i_0 = i$ that left.
	By definition of
	the \gls{pands} mechanism,
	this customer is the one that completed service
	and it could not replace any subsequent customer in the queue.
	Therefore, if state~$c$ contains any class
	that can be swapped with class~$i$,
	then the departing customer of class~$i$ was necessarily
	in a position $p \in \{q_1+1,\ldots,n+1\}$ before the transition,
	where $q_1$ is the largest integer
	$q \in \{1, \ldots, n\}$
	such that $c_q \in \I_{i_0}$.
	If state~$c$ does not contain any such customer,
	we let $q_1 = 0$.
	In both cases, before the departure,
	the queue could be any state of the form
	$d = c_{1 \ldots p-1}, i_0,
	c_{p \ldots n}$,
	where $p \in \{q_1+1,\ldots,n+1\}$.
	\item[Case $v = 2$]
	If two customers were involved in the transition,
	this means that the departing customer of class~$i$
	was ejected by a second customer
	whose service was completed.
	The \gls{pands} mechanism
	and the symmetric property
	of the swapping relation
	impose that this second customer is the one we have just identified,
	in position $q_1$, and that $q_1 \ge 1$.
	We let $i_1 = c_{q_1} \in \I_{i_0}$ denote the class of this second customer.
	By the same argument as before,
	this second customer
	could be in any position
	$p \in \{q_2 + 1, \ldots, q_1\}$
	before the transition,
	where $q_2$ is the largest integer
	$q \in \{1, \ldots, q_1 - 1\}$
	such that $c_q \in \I_{i_1}$, if any,
	and $q_2 = 0$ otherwise.
	In both cases, before the departure,
	the queue could be in any state of the form
	$d = c_{1 \ldots p-1}, i_1,
	c_{p \ldots q_1-1},
	i_0, c_{q_1+1 \ldots n}$,
	where $p \in \{q_2+1, \ldots, q_1\}$.
	\item[Case $v = 3$]
	The departing customer
	was ejected by a second customer,
	which was ejected by a third customer
	whose service was completed.
	Pursuing the previous reasoning,
	we can show that the second involved customer
	is that of class $i_1 = c_{q_1}$, in position $q_1$,
	and the third involved customer is that
	of class $i_2 = c_{q_2}$,
	in position $q_2$,
	assuming that $1 \le q_2 < q_1$.
	Before the transition,
	this third customer could be in any position $p \in \{q_3+1, \ldots, q_2\}$,
	where $q_3$ is the largest integer
	$q \in \{1, \ldots, q_2 - 1\}$
	such that $c_q \in \I_{i_2}$, if any,
	and $q_3 = 0$ otherwise.
	Before the departure,
	the queue could be in any state
	$ d = c_{1 \ldots, p-1},
	i_2, c_{p \ldots q_2-1}, i_1,
	c_{q_2+1 \ldots q_1-1},
	i_0, c_{q_1+1 \ldots n} $,
	where $p \in \{q_3+1,\ldots,q_2\}$.
\end{description}
Continuing on, we build
a decreasing sequence
$n+1 = q_0 > q_1 > q_2 >
\ldots > q_{u-1} > q_u = 0$
of positions in state~$c$
using the recursion
$q_v = \max\{q \le q_{v-1} - 1:
c_q \in \I_{i_{v-1}}\}$
for each $v \in \{1, \ldots, u-1\}$.
The recursion stops
when the set
$\{q \le q_{v-1} - 1:
c_q \in \I_{i_{v-1}}\}$ is empty,
in which case we let $u = v$ and $q_u = 0$.
This integer $u$ gives
the maximum number of customers
that can be involved in the transition
(including the departing class-$i$ customer).
We also define a sequence
$i_0 = i, i_1 = c_{q_1}, i_2 = c_{q_2},
\ldots, i_{u-1} = c_{q_{u-1}}$
of classes.
In the end, the states~$d$
that lead to state~$c$
by a departure of a class-$i$ customer are those of the form
$$
d = c_{1 \ldots p-1}, i_v, c_{p \ldots q_v-1}
, i_{v-1}, c_{q_v+1 \ldots q_{v-1}-1}
, \ldots,
c_{q_3+1 \ldots q_2-1}, i_1,
c_{q_2+1 \ldots q_1-1}, i_0,
c_{q_1+1 \ldots n}
$$
with $v \in \{0,\ldots,u-1\}$
and $p \in \{q_{v+1}+1, q_{v+1}+2,
\ldots, q_v\}$,
where~$p$ gives the position
of the customer, of class~$i_v$,
whose service was actually completed.
This implies that~\eqref{eq:pands-partial-balance-3}
can be rewritten as follows:
\begin{align}
	\nonumber
	\Phi(c)
	= \sum_{v=0}^{u-1} \sum_{p=q_{v+1}+1}^{q_v}
	&
	\begin{aligned}[t]
		\Phi(
		&c_{1 \ldots p-1} , i_v , c_{p \ldots q_v-1}
		, i_{v-1} , c_{q_v+1 \ldots q_{v-1}-1} , i_{v-2} , \\
		& \ldots ,
		c_{q_3+1 \ldots q_2-1} , i_1 ,
		c_{q_2+1 \ldots q_1-1} , i_0 ,
		c_{q_1+1 \ldots n}
		)
	\end{aligned} \\
	\label{eq:proof-partial-balance-2}
	&\times \Delta\mu(c_{1 \ldots p-1} , i_v),
\end{align}
That the balance function~$\Phi$
defined by~\eqref{eq:oi-Phic}
satisfies~\eqref{eq:proof-partial-balance-2}
is shown in the following lemma,
which concludes the proof
of Theorem~\ref{theo:pands}.

\begin{lemma} \label{lem:proof-pands}
	The function $\Phi$ defined by~\eqref{eq:oi-Phic}
	satisfies~\eqref{eq:proof-partial-balance-2}
	for each integers $n \ge 0$
	and $u \in \{1, \ldots, n+1\}$,
	state~$c = (c_1, \ldots, c_n) \in \I^*$,
	class~$i \in \I$,
	and decreasing integer sequence
	$q_0, q_1, \ldots, q_u$ with $q_0 = n+1$ and $q_u = 0$,
	where $i_0 = i, i_1 = c_{q_1}, i_2 = c_{q_2},
	\ldots, i_{u-1} = c_{q_{u-1}}$.
\end{lemma}

\begin{proof}[Proof of the lemma]
	Our proof is by induction
	on the maximum number~$u \ge 1$
	of customers involved in the transition.
	More specifically,
	we show that the following statement
	holds for each positive integer~$u$:
	\begin{displayquote}
		Equation~\eqref{eq:proof-partial-balance-2} is satisfied
		for each integer~$n \ge u-1$,
		state~$c = (c_1, \ldots, c_n) \in \I^*$,
		class~$i \in \I$,
		and decreasing integer sequence
		$q_0, q_1, \ldots, q_u$ with $q_0 = n+1$ and $q_u = 0$.
	\end{displayquote}
	Before we proceed to the proof by induction,
	recall that
	$\Phi$ satisfies the following equation,
	which is a rewritten version
	of Equation~\eqref{eq:oi-partial-balance-3}
	shown in the proof of Theorem~\ref{theo:oi}:
	\begin{equation} \label{eq:proof-partial-balance-3}
		\Phi(c) = \sum_{p=1}^{n+1} \Phi(c_{1 \ldots p-1} , i , c_{p \ldots n})
		\, \Delta\mu(c_{1 \ldots p-1} , i),
		\quad \forall n \ge 0,
		\quad \forall c = (c_1, \ldots, c_n) \in \I^*,
		\quad \forall i \in \I.
	\end{equation}
	
	\paragraph{\normalfont \textit{Base step.}}
	With $u = 1$,
	\eqref{eq:proof-partial-balance-2}
	is equivalent to~\eqref{eq:proof-partial-balance-3}
	with $i = i_0$.
	As we have just mentioned,
	it was already shown that
	$\Phi$ satisfies this equation.
	
	\paragraph{\normalfont \textit{Induction step.}}
	Now let $u \ge 2$ and assume that the statement
	is valid for each
	$u' \in \{1, 2, \ldots, u-1\}$.
	Consider an integer $n \ge u - 1$,
	a state
	$c = (c_1,\ldots,c_n) \in \I^*$,
	a class $i \in \I$,
	and a decreasing integer sequence
	$q_0, q_1, \ldots, q_u$ with $q_0 = n+1$ and $q_u = 0$.
	Also let
	$i_0 = i, i_1 = c_{q_1}, i_2 = c_{q_2},
	\ldots, i_{u-1} = c_{q_{u-1}}$.
	We first apply~\eqref{eq:proof-partial-balance-3}
	to state~$c$ and class $i = i_0$
	and split the sum into two parts to obtain
	\begin{align*}
		\Phi(c)
		&=
		\begin{aligned}[t]
			&\sum_{p=1}^{q_1} \Phi(c_{1 \ldots p-1} , i_0 ,
			c_{p \ldots n})
			\, \Delta\mu(c_{1 \ldots p-1} , i_0)
			+ \sum_{p=q_1+1}^{n+1} \Phi(c_{1 \ldots p-1} , i_0 ,
			c_{p \ldots n})
			\, \Delta\mu(c_{1 \ldots p-1} , i_0).
		\end{aligned}
	\end{align*}
	Using the definition~\eqref{eq:oi-Phic} of $\Phi$
	and the fact that $\mu$ is order independent,
	we rewrite the first sum differently:
	\begin{align}
		\nonumber
		\Phi(c)
		={}&\left( \prod_{p=q_1}^{n} \frac1{\mu( c_{1 \ldots p} , i_0 )} \right)
		\sum_{p=1}^{q_1} \Phi(c_{1 \ldots p-1} , i_0 , c_{p \ldots q_1-1})
		\, \Delta\mu(c_{1 \ldots p-1} , i_0) \\
		\label{eq:proof-partial-balance-4}
		&+ \sum_{p=q_1+1}^{n+1} \Phi(c_{1 \ldots p-1} , i_0 , c_{p \ldots n})
		\, \Delta\mu(c_{1 \ldots p-1} , i_0).
	\end{align}
	But applying \eqref{eq:proof-partial-balance-3}
	to state $c_{1 \ldots q_1-1}$
	and class $i_0$ yields
	$$
	\Phi(c_{1 \ldots q_1-1})
	= \sum_{p=1}^{q_1} \Phi(c_{1 \ldots p-1},
	i_0 , c_{p \ldots q_1-1})
	\, \Delta\mu(c_{1 \ldots p-1} , i_0),
	$$
	so that~\eqref{eq:proof-partial-balance-4} can be rewritten as
	\begin{align} \label{eq:proof-partial-balance-5}
		\Phi(c)
		&=
		\begin{aligned}[t]
			\left( \prod_{p=q_1}^{n} \frac1{\mu( c_{1 \ldots p} , i_0 )} \right)
			\Phi(c_{1 \ldots q_1-1})
			+ \sum_{p=q_1+1}^{n+1} \Phi(c_{1 \ldots p-1} , i_0 , c_{p \ldots n})
			\, \Delta\mu(c_{1 \ldots p-1} , i_0).
		\end{aligned}
	\end{align}
	Now we apply the induction assumption
	to the positive integer $u' = u - 1$,
	with the integer $n' = q_1 - 1$,
	the state $c' = c_{1 \ldots q_1-1}$,
	the class~$i_1$,
	the decreasing sequence
	$q'_0 = q_1 = n' + 1$, $q'_1 = q_2$, \ldots, $q'_{u-2} = q_{u-1}$,
	$q'_{u'} = q'_{u-1} = q_u = 0$,
	and the indices $i'_0 = i_1$, $i'_1 = i_2$,
	\ldots, $i'_{u'-1} = i'_{u-2} = i_{u-1}$.
	We can verify that $n'\ge u' - 1$
	because the sequence
	$q_1$, $q_2$, \ldots, $q_u$
	is decreasing with $q_u = 0$,
	so that
	$q_1 \ge q_2 + 1 \ge q_3 + 2
	\ge \ldots \ge q_u + (u-1) = u-1$.
	For this setting,
	\eqref{eq:proof-partial-balance-2} implies that 
	\begin{align*}
		\Phi(c_{1 \ldots q_1-1})
		= \sum_{v=1}^{u-1} \sum_{p=q_{v+1}+1}^{q_v}
		&
		\begin{aligned}[t]
			\Phi(
			&c_{1 \ldots p-1} , i_v , c_{p \ldots q_v-1}
			, i_{v-1} , c_{q_v+1 \ldots q_{v-1}-1} , i_{v-2} , \\
			& \ldots ,
			c_{q_4+1 \ldots q_3-1} , i_2 ,
			c_{q_3+1 \ldots q_2-1} , i_1 ,
			c_{q_2+1 \ldots q_1-1}
			)
		\end{aligned} \\
		& \times \Delta\mu(c_{1 \ldots p-1} , i_v). \nonumber
	\end{align*}
	Note that the first sum ranges from
	$1$ to $u-1$,
	and not from $0$
	to $u' - 1 = u - 2$,
	as a result of rewriting.
	Doing the substitution in~\eqref{eq:proof-partial-balance-5} yields
	\begin{align*}
		\Phi(c)
		&=
		\begin{aligned}[t]
			&\left( \prod_{p=q_1}^{n} \frac1{\mu( c_{1 \ldots p} , i_{0} )} \right)
			\sum_{v=1}^{u-1} \sum_{p = q_{v+1}+1}^{q_v}
			\begin{aligned}[t]
				&\begin{aligned}[t]
					\Phi(
					&c_{1 \ldots p-1} , i_v , c_{p \ldots q_v-1}
					, i_{v-1} , c_{q_v+1 \ldots q_{v-1}-1} , i_{v-2} , \\
					& \ldots ,
					c_{q_4+1 \ldots q_3-1} , i_2 ,
					c_{q_3+1 \ldots q_2-1} , i_1 ,
					c_{q_2+1 \ldots q_1-1}
					)
				\end{aligned} \\
				&\times \Delta\mu(c_{1 \ldots p-1} , i_v)
			\end{aligned}
			\\
			&+ \sum_{p=q_1+1}^{n+1} \Phi(c_{1 \ldots p-1} , i_0 , c_{p \ldots n})
			\, \Delta\mu(c_{1 \ldots p-1} , i_0).
		\end{aligned}
	\end{align*}
	We again apply~\eqref{eq:oi-Phic}
	and the fact that $\mu$ is order independent
	to move the product back into the first sum,
	so that we obtain
	\begin{align*}
		\Phi(c)
		&=
		\begin{aligned}[t]
			& \sum_{v=1}^{u-1} \sum_{p = q_{v+1}+1}^{q_v}
			\begin{aligned}[t]
				&\begin{aligned}[t]
					\Phi(
					&c_{1 \ldots p-1} , i_v , c_{p \ldots q_v-1}
					, i_{v-1} , c_{q_v+1 \ldots q_{v-1}-1} , i_{v-2} , \\
					& \ldots ,
					c_{q_3+1 \ldots q_2-1} , i_1 ,
					c_{q_2+1 \ldots q_1-1} , i_0 ,
					c_{q_1+1 \ldots n}
					)
				\end{aligned} \\
				&  \times \Delta\mu(c_{1 \ldots p-1} , i_v)
			\end{aligned}
			\\
			&+ \sum_{p=q_1+1}^{n+1} \Phi(c_{1 \ldots p-1} , i_0 , c_{p \ldots n})
			\, \Delta\mu(c_{1 \ldots p-1} , i_0).
		\end{aligned}
	\end{align*}
	We conclude by observing that
	the second sum corresponds
	to the missing term $v = 0$
	in the first sum.
\end{proof}

\section{Proof of Theorem~\ref{theo:stability}} \label{app:stability}

We first prove that \eqref{eq:stability} is a necessary condition for stability by arguing that the Markov process describing the state of the queue over time cannot be ergodic in the absence of this condition. Then, we prove that this condition is sufficient, by comparing the \gls{pands} queue to a degenerate queue with pessimistic service rates.

\paragraph{Necessary condition}

Assume that there is a non-empty set $\A \subseteq \I$
such that $\limmu(\A) \le \sum_{i \in \A} \lambda_i$.
Since $\mu$ is non-decreasing,
this means that $\mu(x) \le \sum_{i \in \A} \lambda_i$
for each $x \in \N^I$
such that $\{i \in \I: x_i > 0\} \subseteq \A$.
Combining this inequality with \eqref{eq:oi-Phic} yields that, for any such $x$,
and for each $c \in \I^*$ such that $|c| = x$,
we have
\begin{equation*}
	\Phi(c)
	\ge \left( \frac1{\sum_{i \in \A} \lambda_i} \right)^{|c|_1 + \ldots + |c|_I}
	= \left( \frac1{\sum_{i \in \A} \lambda_i} \right)^{x_1 + \ldots + x_I},
\end{equation*}
which implies
\begin{align} \label{eq:Phiineq}
	\sum_{\substack{c \in \I^*: |c| = x}}
	\Phi(c) \prod_{i \in \I} {\lambda_i}^{|c|_i}
	\ge \sum_{\substack{c \in \I^*: |c| = x}}
	\; \prod_{i \in \A} 
	\left( \frac{\lambda_i}{\sum_{j \in \A} \lambda_j} \right)^{x_i}
	= \binom{x_1 + \ldots + x_I}{x_1, \ldots, x_I}
	\prod_{i \in \A}
	\left( \frac{\lambda_i}{\sum_{j \in \A} \lambda_j} \right)^{x_i}.
\end{align}
It follows that
\begin{align*}
	\sum_{c \in \I^*} \Phi(c) \prod_{i \in \I} {\lambda_i}^{|c|_i}
	&= \sum_{x \in \N^I}
	\; \sum_{\substack{c \in \I^*: |c| = x}}
	\Phi(c) \prod_{i \in \I} {\lambda_i}^{|c|_i}, \\
	&\ge \sum_{\substack{x \in \N^I: \\ \{i \in \I: x_i > 0\} \subseteq \A}}
	\; \sum_{\substack{c \in \I^*: |c| = x}}
	\Phi(c) \prod_{i \in \I} {\lambda_i}^{|c|_i}, \\
	&\ge \sum_{\substack{x \in \N^I: \\ \{i \in \I: x_i > 0\} \subseteq \A}}
	\binom{x_1 + \ldots + x_I}{x_1, \ldots, x_I}
	\prod_{i \in \A}
	\left( \frac{\lambda_i}{\sum_{j \in \A} \lambda_j} \right)^{x_i}, \\
	&= \sum_{n=0}^\infty \left( \sum_{ i\in \A} \frac{\lambda_i}{\sum_{j \in \A} \lambda_j} \right)^n.
\end{align*}
In the first inequality, we restricted the outer sum so that we can apply \eqref{eq:Phiineq}. In the final equality, we used the multinomial theorem, stating that, for each positive integers $n$ and $N$ and reals $\rho_1, \rho_2, \ldots, \rho_N$, we have
$$
(\rho_1+\ldots+\rho_N)^n = \sum_{x_1+\ldots+x_N =n} \binom{n}{x_1, \ldots, x_N}\prod_{i=1}^N \rho_i^{x_i}.
$$
Since $\sum_{i\in\A} \frac{\lambda_i}{\sum_{j \in \A} \lambda_j} = 1$, the final expression amounts to infinity, so that $\sum_{c \in \I^*} \Phi(c) \prod_{i \in \I} {\lambda_i}^{|c|_i} = \infty$. This ensures that \eqref{eq:oi-stability} is not satisfied, so that the Markov process on $\I^*$ cannot be ergodic.

\paragraph{Sufficient condition}

Assuming that \eqref{eq:stability} is satisfied,
we prove stability in two steps.
We first introduce a second \gls{pands} queue
with the same set~$\I$ of classes
and arrival rates $\lambda_1, \ldots, \lambda_I$
as the original \gls{pands} queue,
but with a rate function $\hat\mu$
such that $\hat\mu(x) \le \mu(x)$ for each $x \in \N^I$.
We will refer to this second \gls{pands} queue
as the \textit{degenerate} \gls{pands} queue,
as the service rate received by the customers of each class
only depends on the number of customers of this class.
Then we will show that the degenerate \gls{pands} queue is stable.
Since the degenerate \gls{pands} queue has more pessimistic service rates than the original \gls{pands} queue, this also implies that the original \gls{pands} queue is stable, as we will see below.

We first introduce several quantities
that will be useful to define
the degenerate \gls{pands} queue.
Since $\bar\mu$ satisfies~\eqref{eq:stability},
there exists an $m \in \N$ such that
\begin{equation*}
	\sum_{i \in \A} \lambda_i < \mu(m e_\A),
	\quad \forall \A \subseteq \I: \A \neq \emptyset.
\end{equation*}
We can also find
$\hat\lambda = (\hat\lambda_1, \ldots, \hat\lambda_I) \in \R_+^I$
such that $\lambda_i < \hat\lambda_i$
for each $i \in \I$, and
\begin{equation*}\label{eq:lahatineq}
	\sum_{i \in \A} \hat\lambda_i < \mu(m e_\A),
	\quad \forall \A \subseteq \I: \A \neq \emptyset.
\end{equation*}
For instance, we can choose
$$
\hat\lambda_i = \lambda_i
+ \frac12
\min_{\substack{\A \subseteq \I: i \in \A}}
\left(
\frac{ \mu(m e_\A) - \sum_{j \in \A} \lambda_j }{ |\A| }
\right),
\quad \forall i \in \I.
$$
Finally, we let
\begin{equation} \label{eq:proof-delta}
	\delta = \frac1I \, \min\left(
	\min_{x \in \N^I \setminus \{0\}} (\mu(x)),
	\min_{\A \subseteq \I: \A \neq \emptyset} \left(
	\mu(m e_\A) - \sum_{i \in \A} \hat\lambda_i
	\right)
	\right).
\end{equation}
The definitions of $\mu$ and $\hat\lambda$
guarantee that $\delta > 0$.
In the degenerate \gls{pands} queue,
$\delta$ will be the service rate
of the customer classes
that have fewer than $m$ present customers.
As we will see later,
choosing this value of $\delta$ ensures that
the service rate of the degenerate queue
is always smaller or equal to that of the original queue.

The degenerate \gls{pands} queue is defined as follows.
Just like the original \gls{pands} queue,
the set of customer classes is $\I = \{1, \ldots, I\}$
and the per-class arrival rates are $\lambda_1, \ldots, \lambda_I$.
But the rate function $\hat\mu$ of this new queue
is defined on $\N^I$ by
$\hat\mu(x) = \sum_{i \in \I} \hat\mu_i(x_i)$, with
\begin{equation*}
	\hat\mu_i(x_i) = \begin{cases}
		0
		&\text{if $x_i = 0$,} \\
		\min(\delta, \hat\lambda_i)
		&\text{if $x_i = 1, 2, \ldots, m-1$,} \\
		\hat\lambda_i
		&\text{if $x_i = m, m+1, \ldots$} \\
	\end{cases}
\end{equation*}
In this way, for each $i \in \I$,
the oldest class-$i$ customer is served at rate $\min(\delta, \hat\lambda_i)$
and the $m$-th oldest class-$i$ customer
is served at rate $\max(\hat\lambda_i - \delta, 0)$.
The service rate of other class-$i$ customers is zero.
It follows that, for each $x \in \N^I \setminus \{0\}$,
we have $\hat\mu(x) \le \mu(x)$.
Indeed,
\begin{itemize}
	\item if $x_i < m$ for each $i \in \I$, then
	$$
	\hat\mu(x)
	\le \sum_{\substack{i \in \I: x_i > 0}} \delta
	\le \sum_{\substack{i \in \I: x_i > 0}} \frac1I \mu(x)
	\le \mu(x),
	$$
	where the first inequality follows from the definition of $\hat\mu$
	and the second holds by~\eqref{eq:proof-delta};
	\item otherwise, with $\A = \{i \in \I: x_i \ge m\}$,
	we have $\A \neq \emptyset$ and $x \ge m e_\A$, so that
	\begin{align*}
		\hat\mu(x)
		\le \sum_{i \in \A} \hat\lambda_i
		+ \sum_{\substack{i \in \I \setminus \A: x_i > 0}} \delta
		\le \sum_{i \in \A} \hat\lambda_i
		+ \sum_{\substack{i \in \I \setminus \A: x_i > 0}}
		\frac{\mu(m e_\A) - \sum_{j \in \A} \hat\lambda_j}I
		\le \mu(m e_\A)
		\le \mu(x),
	\end{align*}
	where the first inequality follows from the definition of $\hat\mu$,
	the second holds by~\eqref{eq:proof-delta},
	and the fourth follows from the monotonicity of $\mu$.
\end{itemize}
We let $\hat\Phi$ denote the balance function
of the degenerate \gls{pands} queue,
as defined in \eqref{eq:oi-Phic}.
Theorems~\ref{theo:oi} and \ref{theo:pands} now
guarantee that the original \gls{pands} queue is stable
whenever the degenerate \gls{pands} queue is.
More particularly, by~\eqref{eq:oi-Phic},
we have for each $(c_1, \ldots, c_n) \in \I^*$:
\begin{align*}
	\Phi(c_1, \ldots, c_n)
	= \prod_{p=1}^n \frac1{\mu(c_1, \ldots, c_p)}
	\le \prod_{p=1}^n \frac1{\hat\mu(c_1, \ldots, c_p)}
	= \hat\Phi(c_1, \ldots, c_n),
\end{align*}
which implies that
\begin{align*}
	\sum_{c \in \I^*} \Phi(c) \prod_{i \in \I} {\lambda_i}^{|c|_i}
	\le \sum_{c \in \I^*} \hat\Phi(c) \prod_{i \in \I} {\lambda_i}^{|c|_i}.
\end{align*}
Therefore, according to \eqref{eq:oi-stability}, the original \gls{pands} queue is stable whenever the degenerate \gls{pands} queue is. It therefore only remains to show that the degenerate \gls{pands} queue is stable. To prove this, we first write:
Therefore, according to \eqref{eq:oi-stability}, the original \gls{pands} queue is stable whenever the degenerate \gls{pands} queue is. It therefore only remains to show that the degenerate \gls{pands} queue is stable. To prove this, we first write:
\begin{align*}
	\sum_{c \in \I^*} \hat\Phi(c)
	\prod_{i \in \I} {\lambda_i}^{|c|_i}
	= \sum_{x \in \{0,1,\ldots,m-1\}^I}
	\sum_{\substack{c \in \I^*: \\ |c| = x}} \hat\Phi(c)
	\prod_{i \in \I} {\lambda_i}^{x_i}
	+ \sum_{\substack{\A \subseteq \I: \\ \A \neq \emptyset}}
	\sum_{\substack{x \in \N^I: \\
			x_i \ge m, \forall i \in \A, \\
			x_i < m, \forall i \notin \A}}
	\sum_{\substack{c \in \I^*: \\ |c| = x}} \hat\Phi(c)
	\prod_{i \in \I} {\lambda_i}^{x_i}.
\end{align*}
The first sum on the right-hand side
is finite because it has a finite number of terms.
The second sum is also finite because,
for each non-empty set $\A \subseteq \I$, we have
\begin{align*}
	\sum_{\substack{x \in \N^I: \\
			x_i \ge m, \forall i \in \A, \\
			x_i < m, \forall i \notin \A\phantom{,}}}
	\sum_{\substack{c \in \I^*: \\ |c| = x}} \hat\Phi(c)
	\prod_{i \in \I} {\lambda_i}^{x_i}
	&= \sum_{\substack{y \in \N^I: \\ y_i = 0, \forall i \notin \A}}
	\sum_{\substack{z \in \N^I : \\
			z_i = 0, \forall i \in \A, \\
			z_i < m, \forall i \notin \A}}
	\sum_{\substack{c \in \I^*: \\ |c| = m e_\A + y + z}} \hat\Phi(c)
	\prod_{i \in \I} {\lambda_i}^{(m e_\A + y + z)_i}, \\
	&= \sum_{\substack{y \in \N^I: \\ y_i = 0, \forall i \notin \A}}
	\sum_{\substack{z \in \N^I : \\
			z_i = 0, \forall i \in \A, \\
			z_i < m, \forall i \notin \A}}
	\sum_{\substack{c \in \I^*: \\ |c| = m e_\A + z}} \hat\Phi(c)
	\prod_{i \in \A} \left( \frac1{\hat\lambda_i} \right)^{y_i}
	\prod_{i \in \I} {\lambda_i}^{(m e_\A + y + z)_i}, \\
	&= \left(
	\prod_{i \in \A}
	\sum_{y_i = 0}^{+\infty}
	\left( \frac{\lambda_i}{\hat\lambda_i} \right)^{y_i}
	\right)
	\sum_{\substack{z \in \N^I: \\
			z_i = 0, \forall i \in \A, \\
			z_i < m, \forall i \notin \A}}
	\sum_{\substack{c \in \I^*: \\ |c| = m e_\A + z}} \hat\Phi(c)
	\prod_{i \in \I} {\lambda_i}^{(m e_\A + z)_i}
	< +\infty.
\end{align*}
The first equality is obtained by substitution.
The second equality follows from the fact that,
using \eqref{eq:oi-Phic} and the definition of $\hat\mu$,
we can prove by induction over
$n = x_1 + \ldots + x_I$ that,
for each $x \in \N^I$, we have
\begin{equation*}
	\sum_{\substack{c \in \I^*: |c| = x}}
	\hat\Phi(c)
	= \prod_{i \in \I}
	\bigg(
	\frac1{\min(\delta, \hat\lambda_i)}
	\bigg)^{\min(x_i, m)}
	\left( \frac1{\hat\lambda_i} \right)^{\max(x_i - m,0)}.
\end{equation*}
The third equality is obtained by rearranging terms. The inequality follows from the fact that
$\lambda_i < \hat\lambda_i$ for each $i \in \I$, so that the product between large parentheses is finite; the rest of the expression is a sum of a finite number of terms, each of which is finite.

\section{Proofs of the propositions in Section~\ref{sec:closed}} \label{app:closed}

In this section,
we give the proofs of
Propositions~\ref{prop:one-order-closed},
\ref{prop:one-order-irreducible},
\ref{prop:two-order-closed},
and \ref{prop:two-order-irreducible}
stated in Section~\ref{sec:closed}.

\theoremstyle{plain}
\newtheorem*{prop:one-order-closed}{Proposition~\ref{prop:one-order-closed}}
\begin{prop:one-order-closed}
	If the initial state of the closed \gls{pands} queue adheres to the placement order $\prec$, then any state reached by applying the \gls{pands} mechanism also adheres to this placement order.
\end{prop:one-order-closed}

\begin{proof}
	Let $c = (c_1, \ldots, c_n)$
	denote the initial state of the queue
	and assume that $c$ adheres
	to the placement order $\prec$.
	Let $p \in \{1, \ldots, n\}$ such that $\Delta\mu(c_1, \ldots, c_p) > 0$
	and consider the transition induced by the service completion
	of the customer in position~$p$.
	In the course of this transition,
	one or more customers are moved
	from the head towards the tail of the queue,
	the last one being moved to
	the last position.
	We now argue that the state
	reached after this transition
	still adheres to $\prec$,
	after which the proposition
	follows immediately, since application of the \gls{pands} mechanism only consists of a number of such transitions.
	
	We first show that
	the customer that completes service,
	of class~$c_p$,
	does not pass over any customer
	of a class~$i$ such that $c_p \prec i$.
	If there is no integer $p' \in \{p+1, \ldots, n\}$
	such that $c_p \prec c_{p'}$,
	the conclusion is immediate.
	Now assume that there is such an integer
	and let $q$ denote the smallest integer
	in $\{p+1, \ldots, n\}$
	such that $c_p \prec c_q$.
	We will show that:
	\begin{enumerate}[label=(\roman*)]
		\item \label{proof1}
		classes $c_p$ and $c_q$
		are neighbors in the swapping graph, and
		\item \label{proof2}
		there is no $r \in \{p+1, \ldots, q-1\}$
		such that classes $c_p$ and $c_r$ are neighbors
		in the swapping graph.
	\end{enumerate}
	By definition of the \gls{pands} mechanism,
	this will imply that
	the customer that completes service
	at position~$p$ replaces
	the customer at position~$q$ in state~$c$,
	so that, after the transition,
	the prefix of length $q-1$
	of the new state is
	$(c_1, \ldots, c_{p-1}, c_{p+1}, \ldots, c_{q-1}, c_p)$.
	By definition of~$q$, this prefix
	still adheres to the placement order, and since the rest of the state does not change, the complete state will as well.
	The same reasoning can be repeated
	for each customer that is moved
	by applying the \gls{pands} mechanism.
	
	We first prove
	property~\ref{proof1} by contradiction.
	Assume that this property is not satisfied.
	By definition of the placement order,
	this implies that
	there is a class $i \in \I$
	such that $c_p \prec i \prec c_q$.
	Since state~$c$ adheres
	to the placement order,
	this implies that
	all class-$i$ customers
	are between positions~$p$ and $q$
	in state~$c$.
	In particular, there is an
	$r \in \{p+1, \ldots, q-1\}$
	such that $c_r = i$
	and, therefore, $c_p \prec c_r$,
	which contradicts the minimality of $q$.
	Therefore, property~\ref{proof1} is satisfied.
	We now prove property~\ref{proof2},
	again by contradiction.
	If this property were not satisfied,
	there would be an $r \in \{p+1, \ldots, q-1\}$
	such that classes $c_p$ and $c_r$
	are neighbors in the swapping graph.
	By definition of the placement order,
	this implies that either $c_p \prec c_r$
	or $c_r \prec c_p$.
	Since $p < r$ and state~$c$
	adheres to the placement order,
	the only possibility is that $c_p \prec c_r$,
	which again contradicts the minimality of~$q$.
	Therefore, property~\ref{proof2} is satisfied.
\end{proof}

\theoremstyle{plain}
\newtheorem*{prop:one-order-irreducible}{Proposition~\ref{prop:one-order-irreducible}}
\begin{prop:one-order-irreducible}
	Assume that $\Delta\mu(c) > 0$
	for each $c \in \I^*$.
	All states that
	adhere to the same placement order
	and correspond to the same macrostate
	form a single closed communicating class
	of the Markov process
	associated with the queue state.
\end{prop:one-order-irreducible}

\begin{proof}
	Given Proposition~\ref{prop:one-order-closed},
	it suffices to show that,
	for all states
	$c = (c_1, \ldots, c_n)$ and
	$d = (d_1, \ldots, d_n)$ that
	adhere to the same placement order~$\prec$
	and satisfy $|c| = |d| = \ell$,
	state~$d$ can be reached
	from state~$c$
	with a positive probability.
	
	If $c = d$, the conclusion is immediate.
	Now assume that $c \neq d$.
	We will construct a path of states
	$c^{0}, c^{1}, \ldots, c^{K-1}, c^{K}$,
	with $c^{0} = c$ and $c^{K} = d$,
	that the queue traverses with a positive probability,
	provided that it starts in state~$c$.
	We argue that such a path
	$c^{0}, c^{1}, \ldots, c^{K}$
	is attained by the following algorithm:
	\begin{description}
		\item[Step 1] Set $k = 0$ and $c^{0} = c$.
		\item[Step 2] Determine the smallest integer $p \in \{1, \ldots, n\}$
		such that $c^{k}_p \neq d_p$.
		\item[Step 3] Let $c^{k+1}$ denote
		the state reached when, in state $c^{k}$,
		the customer in position~$p$ completes service
		and the \gls{pands} mechanism is applied.
		\item[Step 4] Set $k = k + 1$.
		If $c^{k} = d$, then
		$K = k$ and the algorithm terminates.
		Otherwise, go to step 2.
	\end{description}
	The idea behind this algorithm is as follows.
	Step~2 identifies
	the first position in state~$c^{k}$
	at which the class of the customer
	does not coincide
	with that of the customer
	at the same position in state~$d$.
	This position is denoted by~$p$.
	Since $(c^k_1, \ldots c^k_{p-1})
	= (d_1, \ldots, d_{p-1})$,
	the customers in positions $1$ to $p - 1$
	need not have their position altered.
	Now let~$r$ denote the smallest integer
	in $\{p+1, \ldots, n\}$
	such that $c^{k}_r = d_p$.
	We will show in the next paragraph that,
	due to the service completion step,
	the customer in position~$r$ in state~$c^{k}$
	is one step closer to (or even attains)
	position~$p$ in state~$c^{k+1}$
	compared to state~$c^{k}$.
	This suffices to prove
	that the algorithm terminates.
	Step~4 makes sure that the two states
	are equal to each other,
	otherwise it initiates a new \gls{pands} transition.
	
	We now prove that,
	if $r$ denotes the smallest integer
	in $\{p+1, \ldots, n\}$
	such that $c^{k}_r = d_p$,
	then $c^{k+1}_{r-1} = c^{k}_r$.
	This is equivalent to proving that
	the customer in position~$r$
	in state~$c^{k}$
	is not ejected in the course of
	the transition described in step~3.
	To prove this,
	it is sufficient to show that,
	for each $q \in \{p, \ldots, r-1\}$,
	classes $c^k_q$ and $c^k_r$
	cannot be swapped with one another,
	that is, are not neighbors
	in the swapping graph.
	Let $q \in \{p, \ldots, r-1\}$.
	The adherence of state~$d$
	to the placement order
	implies that $d_{q'} \not\prec d_p$
	for each $q' \in \{p+1, \ldots, n\}$.
	As $(d_1, \ldots, d_{p-1})
	= (c^k_1, \ldots, c^k_{p-1})$
	and $d_p = c^k_r$,
	this implies that
	$c^{k}_{q} \nprec c^{k}_r$.
	It also follows
	from Proposition~\ref{prop:one-order-closed}
	that state $c^{k}$
	adheres to the placement order,
	so that $c^{k}_r \nprec c^{k}_q$.
	Therefore, we have
	$c^{k}_{q} \nprec c^{k}_r$
	and $c^{k}_r \nprec c^{k}_q$,
	which, by definition
	of a placement order,
	implies that
	classes $c^{k}_q$ and $c^{k}_r$
	are not neighbors in the swapping graph.
\end{proof}

\theoremstyle{plain}
\newtheorem*{prop:two-order-closed}{Proposition~\ref{prop:two-order-closed}}
\begin{prop:two-order-closed}
	If the initial network state
	adheres to the placement order $\prec$,
	then any state reached
	by applying the \gls{pands} mechanism
	to either of the two queues
	also adheres to this placement order.
\end{prop:two-order-closed}

\begin{proof}
	By symmetry,
	it suffices to prove that,
	if a network state adheres
	to the placement order~$\prec$,
	then any network state reached by
	a service completion in the first queue
	also adheres to this placement order.
	Consider a state~$(c;d)$
	that adheres to the placement order
	and let $c = (c_1, \ldots, c_n)$
	and $d = (d_1, \ldots, d_m)$.
	Let $c' = (c'_1, \ldots, c'_{n-1})$
	denote the state of the first queue
	right after a service completion in this queue
	and~$i$ the class of the customer
	that departs this queue.
	The state of the network
	right after the transition
	is $(c';d')$
	with $d' = (d_1, \ldots, d_m, i)$.
	Applying Proposition~\ref{prop:one-order-closed}
	to the first queue
	yields that state
	$(c'_1, \ldots, c'_{n-1}, i)$
	adheres to the placement order,
	from which we can derive that
	properties~\ref{order-1} and \ref{order-3}
	are satisfied by the new network state.
	Finally, the fact that state~$(c;d)$ satisfies properties~\ref{order-2} and \ref{order-3} implies that state~$d'$ satisfies property~\ref{order-2}.
\end{proof}

\theoremstyle{plain}
\newtheorem*{prop:two-order-irreducible}{Proposition~\ref{prop:two-order-irreducible}}
\begin{prop:two-order-irreducible}
	Assume that either
	$\Delta\mu(c) > 0$ for each $c \in \I^*$
	or $\Delta\nu(d) > 0$ for each $d \in \I^*$ (or both).
	All states that adhere
	to the same placement order
	and correspond to the same macrostate
	form a single closed communicating class
	of the Markov process
	associated with the network state.
\end{prop:two-order-irreducible}

\begin{proof}
	Without loss of generality,
	we assume that
	$\Delta\mu(c) > 0$
	for each $c \in \I^*$.
	The case where
	$\Delta\nu(d) > 0$
	for each $d \in \I^*$
	is solved by exchanging
	the roles of the two queues.
	Given the result of
	Proposition~\ref{prop:two-order-closed},
	it suffices to show that,
	for all states
	$(c;d)$ and $(c';d')$
	that adhere to the same placement order
	and correspond to the same macrostate,
	state $(c';d')$ can be reached
	from state $(c;d)$
	with a positive probability.
	
	Consider two states
	$(c;d)$ and $(c';d')$
	that adhere to the same
	placement order
	and satisfy $|c| + |d| = |c'| + |d'|$.
	The numbers of customers
	in states $c$, $d$, $c'$, and $d'$
	are denoted by $n$, $m$, $n'$, and $m'$,
	respectively.
	We now build a series of transitions
	that leads from state~$(c;d)$
	to state~$(c';d')$
	with a positive probability.
	
	First let $(c'';\emptyset)$
	denote the state
	reached from state~$(c;d)$ by having,
	$m$ times in a row,
	the customer at the head
	of the second queue
	complete service.
	Proposition~\ref{prop:two-order-closed}
	guarantees that
	state~$(c'';\emptyset)$
	adheres to the placement order
	so that, by property~\ref{order-1},
	state~$c''$ adheres
	to the placement order.
	Since we assumed that state~$(c';d')$
	adheres to the placement order,
	we also have that state
	$(c'_1, \ldots, c'_{n'}, d'_{m'}, \ldots, d'_1)$ adheres to the placement order.
	Therefore, it follows from
	Proposition~\ref{prop:one-order-irreducible}
	that, if the first queue
	evolved in isolation,
	as in Section~\ref{subsec:one},
	it would be possible to reach
	state
	$(c'_1, \ldots, c'_{n'}, d'_{m'}, \ldots, d'_1)$
	from state $c''$
	with positive probability.
	We can adapt the algorithm in this proposition
	to prove that, in the tandem network, state
	$(c'_1, \ldots, c'_{n'}, d'_{m'}, \ldots, d'_1;
	\emptyset)$
	can also be reached from state $(c'';\emptyset)$
	with a positive probability: it suffices to add a
	transition, after step~3,
	that consists of the service completion
	of the (only) customer
	in the second queue
	(so that this customer joins the back
	of the first queue).
	Once state
	$(c'_1, \ldots, c'_{n'}, d'_{m'}, \ldots, d'_1;
	\emptyset)$ is reached,
	it suffices to have
	the customer at the back of the first queue
	complete service
	$m'$ times in a row.
	Since a service completion
	at the final position of a queue
	does not trigger
	any \gls{pands} movement,
	the network state
	$(c';d')$ is reached,
	which concludes the proof.
\end{proof}

\section{Closed pass-and-swap queues with non-adhering initial states}\label{app:irreducibility}

As mentioned in Remark \ref{remark2}, a product-form stationary distribution can also be found for closed \gls{pands} queues in which the initial state does not adhere to a placement order.
To do so, we first associate, with each closed \gls{pands} queue, another closed \gls{pands} queue. We call this other queue the associated \emph{isomorphic queue}.
This isomorphic queue has
the same dynamics as the original queue
but its set of customer classes is different.
The initial state of this isomorphic queue does adhere to a placement order by construction, so that Propositions~\ref{prop:one-order-closed} and \ref{prop:one-order-irreducible} and Theorem~\ref{theo:one-picd} can be applied. This in its turn leads to a stationary distribution for the original closed \gls{pands} queue in the general case.

\subsection{The isomorphic queue}

We first define, for any closed \gls{pands} queue, its associated isomorphic queue. If the initial state of the original queue contains a single customer of each class, as in the example of Section~\ref{subsubsec:one-example}, its associated isomorphic queue is the queue itself.
We now describe how the isomorphic queue is constructed if the initial state of the original queue contains two or more customers of the same class.

Let $c = (c_1, \ldots, c_n)$ denote the initial state of the queue and consider a class~$i \in \I$ and two positions $p, q \in \{1, \ldots, n\}$ such that $c_p = c_q = i$ and $p < q$. We introduce an extra class $i'$ (so that $\I$ is replaced with $\I \cup \{i'\}$) that has the same characteristics as class~$i$. More specifically, we impose that $\Delta\mu(d_1, \ldots, d_m,i) = \Delta\mu(d_1, \ldots, d_m ,i')$ for each $d = (d_1, \ldots, d_m) \in \I^*$. In the swapping graph, for each $k \in \I \setminus \{i,i'\}$, we add an edge between classes~$i'$ and $k$ if and only if there is an edge between classes~$i$ and $k$.

Moving towards a setting where all customers have different classes, we alter the initial state~$c$ by changing the class of the customer in position~$q$ from $i$ to $i'$. While $c_p$ and $c_q$ are not equal anymore, the definition of class~$i'$ guarantees that the dynamics of the queue remain the same. This procedure can be repeated with newly selected class~$i$ and positions~$p$ and $q$ as long as there are at least two customers with the same class in state~$c$. The queue obtained once all customers have different classes is called the isomorphic queue. If $\bar c$ is the initial state of the isomorphic queue obtained by repeating this procedure, we say that state $\bar c$ in the isomorphic queue corresponds to state $c$ in the original queue.

\begin{example}\label{ex:isomorphic}
	We now illustrate the construction 	of an isomorphic queue by means of the closed \gls{pands} queue	depicted in \figurename~\ref{fig:exampleIsomorphic}. This queue has six customers belonging to three classes.
	There are two class-1 customers, three class-2 customers, and one class-3 customer. The initial state of the queue, shown in \figurename~\ref{fig:iso-2}, is $(c_1, c_2, c_3, c_4, c_5, c_6)	= (1,2,1,2,2,3)$.	This state does not adhere	to any placement order because customers of classes~1 and~2 are interleaved. To construct the isomorphic queue, we progressively eliminate pairs of equal customer classes. For example, since $c_1 = c_3 = 1$,	we introduce an extra class~$1'$ such that $\Delta\mu(d_1, \ldots, d_m, 1) = \Delta\mu(d_1, \ldots, d_m, 1')$ for each state $d = (d_1, \ldots, d_m) \in \I^*$, with $\I = \{1,2,3\}$. Moreover, in the swapping graph, we add edges between class~$1'$ and classes~$2$ and~$3$. Finally, we change the class	of the customer in position~3 to~$1'$. This procedure has no effect	on the future dynamics of the queue	but the customers in positions~1 and 3 are now the only members	of their respective classes. The result is not yet an isomorphic queue since, for example, the customers in positions~2 and~4 are both of class~2. We therefore iterate this procedure,	changing the class of the customer in position~4 into class~$2'$ and adding an edge between	class~$2'$ and classes 1, $1'$, and~3	in the swapping graph.
	After this action, only the customers in positions~2 and~5 belong to the same class. Changing the class of the customer	in position~5 to an extra class~$2''$, with the same characteristics as class~2, yields the isomorphic queue shown in	\figurename~\ref{fig:exampleIsomorphic}.
	State $(1, 2, 1', 2', 2'', 3)$ in the isomorphic queue now corresponds to state $(1,2,1,2,2,3)$	in the original queue.
	
	\begin{figure}[ht]
		\centering
		\subfloat[Swapping graph
		of the original queue.
		\label{fig:iso-1}]{%
			\begin{tikzpicture}
				\def\width{2.1cm}
				\def\height{1.7cm}
				
				\node[class, fill=green!60] (1) {1};
				\node[class, fill=orange!60] (2)
				at ($(1)+(\width,0)$) {2};
				\node[class, fill=yellow!60] (3)
				at ($(1)!.5!(2)-(0,\height)$) {3};
				
				\draw (1) -- (2) -- (3) -- (1);
				
				\node at ($(1)-(1.9cm,0)$) {};
				\node at ($(2)+(1.9cm,0)$) {};
				\node at ($(queue.south)-(0,1cm)$) {};
			\end{tikzpicture}
		}
		\qquad
		\subfloat[A state of the original queue.
		\label{fig:iso-2}]{%
			\begin{tikzpicture}
				\def\width{1.2cm}
				\def\height{1.8cm}
				
				\node[fcfs=6,
				rectangle split part fill
				={yellow!60, orange!60, orange!60, green!60, orange!60, green!60},
				] (queue) {
					\nodepart{one}{3}
					\nodepart{two}{2}
					\nodepart{three}{2}
					\nodepart{four}{1}
					\nodepart{five}{2}
					\nodepart{six}{1}
				};
				
				\node[server, anchor=west] (mu)
				at ($(queue.east)+(.1cm,0)$) {};
				
				\draw[->] ($(mu.east)+(.1cm,0)$)
				-- ($(mu.east)+(.3cm,0)$)
				|- ($(queue.south west)-(.3cm,.2cm)$)
				-- ($(queue.west)-(.3cm,0)$)
				-- ($(queue.west)-(.1cm,0)$);
			\end{tikzpicture}
		}
		\\
		\subfloat[Swapping graph of the isomorphic queue.
		\label{fig:iso-3}]{%
			\begin{tikzpicture}
				\def\width{2.3cm}
				\def\height{1.8cm}
				
				\node[class, fill=\darkgreen] (1) {1};
				\node[class, fill=\lightgreen] (11)
				at ($(1)+(.22cm,.85cm)$) {$1'$};
				\node[class, fill=\neutralorange] (2)
				at ($(1)+(\width,0)$) {$2'$};
				\node[class, fill=\darkorange] (22)
				at ($(1)+(\width,0)+(-.22cm,.85cm)$) {2};
				\node[class, fill=\lightorange] (222)
				at ($(1)+(\width,0)+(.22cm,-.85cm)$) {$2''$};
				\node[class, fill=yellow!60] (3)
				at ($(1)!.5!(2)-(0,\height)$) {3};
				
				\draw (1) -- (2);
				\draw (1) -- (22);
				\draw (1) -- (222);
				\draw (11) -- (2);
				\draw (11) -- (22);
				\draw (11) -- (222);
				\draw (1) -- (3);
				\draw (11) -- (3);
				\draw (2) -- (3);
				\draw (22) -- (3);
				\draw (222) -- (3);
				
				\node at ($(1)-(1.8cm,0)$) {};
				\node at ($(2)+(1.8cm,0)$) {};
			\end{tikzpicture}
		}
		\qquad
		\subfloat[State of the isomorphic queue
		corresponding to \figurename~\ref{fig:iso-2}.
		\label{fig:iso-4}]{%
			\begin{tikzpicture}
				\def\width{1.2cm}
				\def\height{1.8cm}
				
				\node[fcfs=6,
				rectangle split part fill
				={yellow!60, \lightorange, \neutralorange, \lightgreen, \darkorange, \darkgreen},
				] (queue) {
					\nodepart{one}{3}
					\nodepart{two}{\enspace}
					\nodepart{three}{\enspace}
					\nodepart{four}{\enspace}
					\nodepart{five}{2}
					\nodepart{six}{1}
				};
				
				\node
				at ($(queue.two)+(.1cm,0)$)
				{$2''$};
				\node
				at ($(queue.three)+(.1cm,0)$)
				{$2'$};
				\node
				at ($(queue.four)+(.1cm,0)$)
				{$1'$};
				
				\node[server, anchor=west] (mu)
				at ($(queue.east)+(.1cm,0)$) {};
				
				\draw[->] ($(mu.east)+(.1cm,0)$)
				-- ($(mu.east)+(.3cm,0)$)
				|- ($(queue.south west)-(.3cm,.2cm)$)
				-- ($(queue.west)-(.3cm,0)$)
				-- ($(queue.west)-(.1cm,0)$);
			\end{tikzpicture}
		}
		\caption{A closed \gls{pands} queue
			and its isomorphic queue.
			The colors and shades
			are visual aids that help distinguish classes.%
		}
		\label{fig:exampleIsomorphic}
	\end{figure}
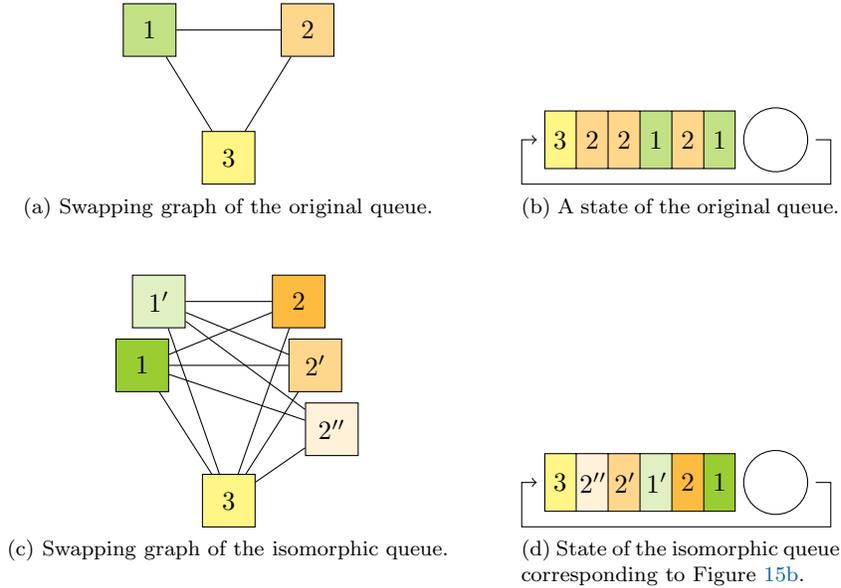
\end{example}

While the isomorphic queue has by construction the same dynamics as the original \gls{pands} queue, it has the following useful property.

\begin{lemma}\label{lem:isomorphicUnique}
	Every state of an isomorphic queue adheres to a unique placement order. 
\end{lemma}
\begin{proof}
	To prove this lemma, we show how to construct a placement graph for any state, which defines the placement order to which the state adheres, and moreover show that it is unique. Recall that a placement graph is an acyclic orientation of the swapping graph, and consists of as many vertices as there are customers in the isomorphic queue. Therefore, to construct the placement graph, the edges of the swapping graph need to be given an orientation. Recall that all customer classes appear exactly once in a state of an isomorphic queue. Therefore, we have that, for any two customer classes $i$ and $j$ that are connected by an edge in the swapping graph, the placement order $\prec$ of the placement graph to be constructed should satisfy $i \prec j$ or $j \prec i$, depending on whether or not the class-$i$ customer is nearer to the front of the queue than the class-$j$ customer. Orienting each edge accordingly yields a directed graph. Due to the transitivity of the order of customers in the state, this directed graph must be acyclic, and therefore it is a placement graph to which the state adheres. Note that, if any edge's orientation corresponding to this placement order were flipped, the corresponding pair of positions in the state would violate the new orientation. This proves the uniqueness, finalizing the proof.
\end{proof}

\subsection{Stationary distribution}

Now that the isomorphic queue has been introduced,
we can derive the stationary distribution
for a closed \gls{pands} queue
of which the initial state is not necessarily adhering.
We can do this because
the initial state of the isomorphic queue
does necessarily adhere to a placement order, say $\prec$.
Let $\bar\ell$ be the macrostate
of the isomorphic queue corresponding
to the initial macrostate $\ell$
of the closed \gls{pands} queue.
We denote by $\bar \C$ the set of states $\bar c$
in the isomorphic queue that adhere to $\prec$
and satisfy $|\bar c| = \bar\ell$.
Also let $\C$ denote
the set of states
of the original closed \gls{pands} queue
to which the states
in $\bar \C$ correspond.
Finally, for each $c \in \C$,
let $\bar\C_c\subset\bar \C$ denote
the set of states $\bar c\in\bar \C$ in the isomorphic queue
that correspond to state~$c$ in the original queue.
Note that $\bar\C_c$ may consist
of multiple elements and that,
considering all $c \in \C$,
the sets $\bar\C_c$ form a partition of $\bar \C$.
For example, a state $(1,2,2,3)$
of a closed \gls{pands} queue
may have corresponding
states $(1,2,2',3)$ and $(1,2',2,3)$
in the isomorphic queue,
both adhering to the same placement order.
Importantly, since all states in $\C$
correspond to the same macrostate $\ell$,
the sets $\bar\C_c$ for all $c \in \C$
have the same cardinality.
These definitions allow us to derive
the stationary distribution of
the Markov process associated with
the state of the original \gls{pands} queue.

\begin{theorem} \label{theo:iso-picd}
	The results of Theorem~\ref{theo:one-picd}
	remain valid if $\C$
	refers to the set of states
	of the original \gls{pands} queue
	to which the isomorphic states
	in $\bar\C$ correspond.
\end{theorem}

\begin{proof}
	By Lemma~\ref{lem:isomorphicUnique}, the initial state of the isomorphic queue must adhere to a placement order that we denote by $\prec$.
	Therefore, applying Theorem~\ref{theo:one-picd} to the isomorphic queue shows that the stationary distribution of the Markov process associated with its state is given by
	\begin{equation*}
		\bar \pi(\bar c)
		= \frac{\bar\Phi(\bar c)}
		{\sum_{\bar d \in \bar \C} \bar\Phi(\bar d)}
		= \frac{\bar\Phi(\bar c)}
		{\sum_{d \in \C} \sum_{\bar d \in \bar\C_d}
			\bar\Phi(\bar d)},
		\quad \forall \bar c \in \bar\C,
	\end{equation*}
	where $\bar\Phi(\bar c) = \Phi(c)$
	for each $c \in \C$ and $\bar c \in \bar\C_c$.
	
	By construction of the isomorphic queue, the dynamics of the original \gls{pands} queue and its isomorphic queue are the same. As such, the stationary probability of the original queue residing in state~$c$ is equal to the stationary probability of the isomorphic queue residing in any state of $\bar\C_c$, leading to:
	\begin{equation*}\label{eq:surjectiveStat}
		\pi(c)
		= \sum_{\bar c\in\bar\C_c} \bar\pi(\bar c)
		= \frac{\sum_{\bar c \in \bar\C_c} \bar\Phi(\bar c)}
		{\sum_{d \in \C} \sum_{\bar d \in \bar\C_d}
			\bar\Phi(\bar d)},
		\quad \forall c \in \C.
	\end{equation*}
	Equation~\eqref{eq:one-pic} follows by recalling that
	$\bar\Phi(\bar c) = \Phi(c)$
	for each $\bar c \in \bar\C_c$
	and that all sets $\bar\C_c$
	have the same cardinality.
\end{proof}

\begin{remark}
	By Remark~\ref{remark1}, the isomorphic queue cannot have transient states. Since an isomorphic queue with identical dynamics can be constructed for any closed \gls{pands} queue, this implies that the Markov process associated with the state of any closed \gls{pands} queue, regardless of any adherence of its initial state, cannot have transient states either. As a result, Theorem~\ref{theo:iso-picd} now implies a partition of the complete state space $\I^*$ in closed communicating classes, each of which corresponds to a set $\bar \C$ defined by a combination of a macrostate $\bar\ell$ and a particular placement order $\prec$ in the isomorphic queue.
\end{remark}

\end{document}